\documentclass[conference]{IEEEtran}
\ifCLASSINFOpdf
\else
\fi

\usepackage{hyperref}
\usepackage{graphicx, xcolor}
\usepackage{algorithm}
\usepackage{algpseudocode}
\usepackage{amsmath, amsthm}
\usepackage{amsfonts}
\usepackage{url}
\newtheorem{theorem}{Theorem}[section] 
\usepackage{pifont}
\newcommand{\cmark}{\ding{51}}
\newcommand{\xmark}{\ding{55}}
\usepackage{subcaption}
\newtheorem{proposition}[theorem]{Proposition}
\newtheorem{lemma}[theorem]{Lemma}
\newtheorem{corollary}[theorem]{Corollary}
\newtheorem{definition}[theorem]{Definition}

\newcommand{\name}{Dependency Triad}
\newcommand{\acronym}{DT}

\begin{document}
%

\title{\name: A Metric to Quantify the Dependencies Between Attributes for Local Differential Privacy}




%




\author{\IEEEauthorblockN{Sandaru Jayawardana\IEEEauthorrefmark{1},
Sennur Ulukus\IEEEauthorrefmark{2},
Ming Ding\IEEEauthorrefmark{3}and
Kanchana Thilakarathna\IEEEauthorrefmark{1}}
\IEEEauthorblockA{\IEEEauthorrefmark{1}The University of Sydney, Australia}
\IEEEauthorblockA{\IEEEauthorrefmark{2}University of Maryland, USA}
\IEEEauthorblockA{\IEEEauthorrefmark{3} Technology, CSIRO, Australia}}

\maketitle

\begin{abstract}
Collecting multidimensional user data is essential for extracting rich insights across various applications. 
Local Differential Privacy (LDP) has emerged as a de facto standard for mitigating privacy risks in such scenarios. 
A key challenge in privacy-preserving multidimensional data collection lies in inter-attribute dependencies, as they can inadvertently reveal correlated information and increase privacy vulnerabilities.
Therefore, accurately measuring correlation-induced privacy leakage (CPL) is essential for privacy analysis and privacy-utility trade-off.
However, existing CPL analysis solutions either require accurate prior knowledge or face scalability challenges for large numbers of attributes and high-cardinality attributes. 
These limit their practical applicability in real data.
To address this research gap, we propose a novel metric, ``\emph{\name}'' (\acronym), which summarizes the pairwise dependency information relevant to CPL using three parameters and yields a \emph{constant-time} conservative estimator of pairwise CPL. \acronym\ explicitly models uncertainty in prior distributional knowledge through its parameters, delivering robust leakage estimates. 
Moreover, its robustness to sparse distributions makes it particularly suitable for high-cardinality attributes, while the pairwise formulation serves as a tractable building block for assessing total leakage in multidimensional settings.
Extensive experiments on both synthetic and real datasets demonstrate that \acronym\ consistently estimates CPL across diverse dependency regimes and prior uncertainties. 

\end{abstract}


%
\IEEEpeerreviewmaketitle

\section{Introduction}\label{section:introduction}

In today's information age, collecting \textit{multidimensional} user records (i.e., records which contain multiple features/attributes) is crucial for applications such as \textit{personalized services}~\cite{rappor_original, L-SRR_location_based_LBS} and \textit{data crowd-sourcing}~\cite{location_based_survey_LBS, PCKV_key_value}. 
These records often include sensitive attributes---\textit{personally identifiable information (PII)}~\cite{Local_differential_privacy_and_its_applications_A_survey}, \textit{financial}~\cite{fintech_privacy_review}, and \textit{health}~\cite{SecureMA_participant_privacy_genetic_association}---creating substantial privacy risks. 
$\varepsilon$-local differential privacy ($\varepsilon$-LDP)~\cite{LDP_duchi} is the de facto standard to provide privacy when curators are untrusted, or users will not share raw data~\cite{Local_Differential_Privacy_in_Practice}. 

An LDP mechanism is a randomized algorithm that perturbs data to bound worst-case leakage by $\varepsilon$ ($\geq 0$). 
Here, a smaller $\varepsilon$ yields stronger privacy (more noise) but typically results in lower data utility (data accuracy); therefore, experts tune $\varepsilon$ to trade off between data utility and privacy. 
LDP has been deployed at scale by industry leaders such as Google~\cite{rappor_original}, Microsoft~\cite{microsoft_ldp}, and Apple~\cite{apple2017privacy}.

Multidimensional data inherently contain inter-dependencies~\cite {CALM, PCKV_key_value}, making privacy preservation challenging because correlation may reveal more information~\cite{no_free_lunch,paper_1}. 
Studies have proposed various methods to apply LDP in multidimensional data---e.g., privacy budget splitting~\cite{SPL_RS_RS_PLUS_FAKE}, random-sampling~\cite{RS_plus_FD}, joint perturbation~\cite{PCKV_key_value}, etc.---to overcome correlation-induced privacy leakage (CPL) impact. 
These approaches primarily assume that the prior knowledge of the attribute dependencies is unknown/unavailable. 
However, studies highlight that prior knowledge of such correlations is ubiquitous today; for example, through {universal correlation information}~\cite{impact_of_prior_knowledge}, {publicly available datasets}~\cite{secon_Collecting_High-Dimensional_Correlation_LDP}, and {learning distributions from a subset of the population}~\cite{ AAA}. 
\emph{Therefore, when prior knowledge is available, leveraging it to analyze CPL enables more nuanced applications, including evaluating the privacy guarantee and calibrating the privacy budgets of the mechanisms}~\cite{data_correlation_location_similar_LDP, paper_1,secon_Collecting_High-Dimensional_Correlation_LDP}.

Pioneer work~\cite{correlated_DP} uses conventional correlation metrics---e.g., \textit{mutual information (MI)}~\cite{covers_mutual_information}, \textit{Pearson correlation coefficient (PCC)}~\cite{PCC_new}---to estimate CPL. 
However, recent studies~\cite{paper_1, data_correlation_location_similar_LDP} claim that these conventional correlation metrics are unable to accurately characterize the CPL, even when ground-truth probability distributions between attributes are available. 
Recent studies propose algorithms to estimate CPL between attributes using probability distributions to address limitations of conventional correlation metrics. 
For instance, suppose we need to compute the CPL of an attribute $X_k$ caused by a correlated attribute $\hat{X}$. 
Let $a$ and $b$ be the alphabet sizes\footnote{We refer to the cardinality of attributes as alphabet size.} of $X_k$ and $\hat{X}$, respectively. 
Study~\cite{data_correlation_location_similar_LDP} (ITS) and study~\cite{secon_Collecting_High-Dimensional_Correlation_LDP} (CBP) quantify CPL only for special-case LDP mechanisms whose randomization can be parameterized by two perturbation-probability levels (e.g., k-RR~\cite{Local_differential_privacy_and_its_applications_A_survey}, Unary encoding~\cite{LDP_Frequent_Itemset_Mining}), which limits applicability to more general mechanisms (e.g., staircase mechanism~\cite{Staircase_Mechanism}). These algorithms have time complexities of $O(ab)$ and $O(2^a)$, respectively. 
A recent study~\cite{paper_1} proposes two algorithms for estimating CPL. 
The first algorithm, HCC-1, estimates the actual CPL using the transition probabilities of the LDP mechanism and the joint probability distribution between attributes.
The second algorithm, HCC-2, estimates the supremum value of the CPL for any LDP mechanism, using a privacy budget of $\varepsilon$ and a joint probability distribution between attributes. 
These algorithms have time complexities of $O(a^2b^2)$ and $O(a^2b\log b)$, respectively. Study~\cite{Quantifying_DP} proposes a precompute-based approach (LTM) to solve a similar optimization as in~\cite{paper_1}, which uses piecewise functions to compute the solution in $O(\log ab)$. 

All these existing algorithms assume access to \emph{ground-truth} distributional information about the data, which is typically \emph{unavailable}, especially in local data perturbation settings like LDP. 
Known distribution (i.e., prior knowledge) might mismatch the actual distribution due to various reasons, such as outdated knowledge, demographic differences, and noisy data. 
Therefore, the estimated CPL may differ from the actual leakage, leading to potential privacy vulnerabilities. 
Furthermore, real-world applications, such as privacy budget calibration, typically require assessing CPL at various privacy budget levels to identify the optimal privacy budget (i.e., privacy-utility trade-off) while keeping the entire distribution (e.g., in~\cite{paper_1, secon_Collecting_High-Dimensional_Correlation_LDP, data_correlation_location_similar_LDP}) or piecewise functions (space complexity of $O(ab)$ for two attributes). 
Therefore, these algorithms are \textbf{not} scalable well with larger alphabet sizes and datasets due to \textit{high time and space complexities}. 
Therefore, their practical applicability is limited by the requirement of ground-truth distributional information and their high space and computational complexity. 
Table~\ref{table:summary_CPL_algos} summarizes the existing solutions for computing CPL. 

\begin{table}[t]
\centering
\caption{Comparison between different CPL computation metrics and algorithms.}
\begin{tabular}{lcccc}
\hline
Algorithm/ & Scalability & Distribution Bias & Supporting \\ 
Metric & & Incomparability & Mechanisms \\\hline
MI~\cite{covers_mutual_information}/PCC~\cite{PCC_new} & \cmark & \xmark & N/A\\
ITS~\cite{data_correlation_location_similar_LDP}  & \xmark & \xmark & Two-level LDP \\
CBP~\cite{secon_Collecting_High-Dimensional_Correlation_LDP}  & \xmark & \xmark & GRR \\
HCC-1~\cite{paper_1} & \xmark & \xmark & Simple LDP \\
HCC-2~\cite{paper_1} & \xmark & \xmark & Any LDP \\
LTM~\cite{Quantifying_DP} & \xmark & \xmark & Any LDP\\
\hline
\end{tabular}
\label{table:summary_CPL_algos}
\end{table}

\noindent \textbf{Problem Statement -} In summary, although quantifying CPL is a promising solution to evaluate privacy guarantee and maximize privacy-utility trade-off in multidimensional data, existing metrics and algorithms have \textbf{two key shortcomings} that limit their practical usability on real-world data: 
\begin{enumerate}
    \item [$\bullet$] \textit{\textbf{P1 - Limited to specific LDP mechanisms or scalability limitations.}}

    \item [$\bullet$] \textit{\textbf{P2 - Requires ground truth probability distributions between the attributes}}.
\end{enumerate}
Next, let us consider a motivational example.

\begin{figure}[hbt]
  \centering
  \includegraphics[trim={2cm 1cm .2cm 0cm},clip,width=1.15\linewidth]{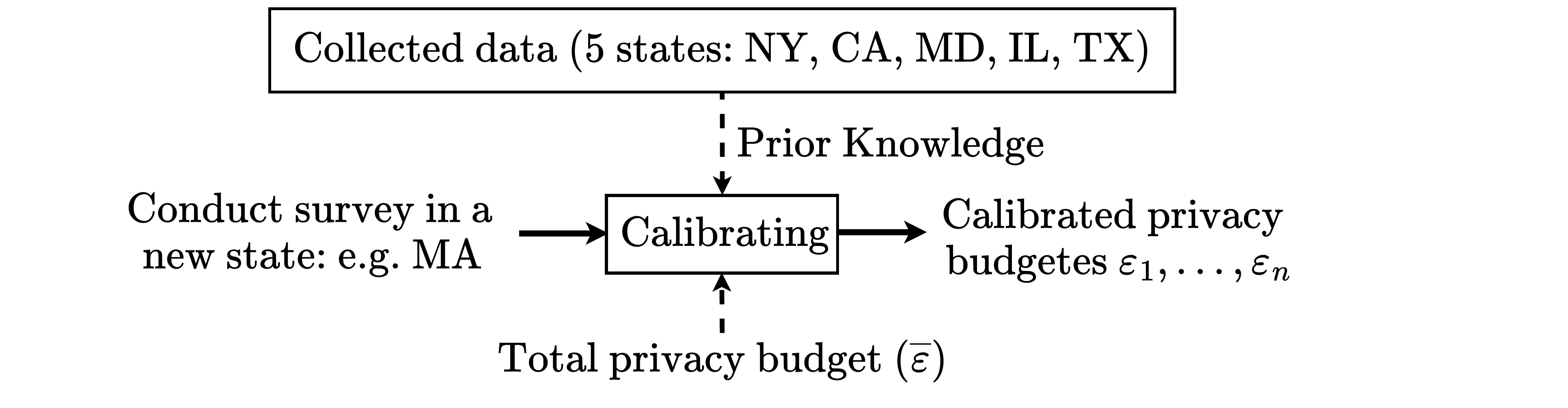}
  \caption{Motivational example. 
  $\overline{\varepsilon}$ denotes the maximum privacy leakage allowed for each attribute collected.} 
  \label{fig:us_income_survey}
\end{figure}

\noindent\textbf{Motivation Example - }Suppose we conduct a state-wise household income survey in the United States under LDP. 
Data from an initial subset of states (e.g., five states) can be collected using correlation-agnostic baselines such as simple privacy–budget splitting (SPL) or random sampling~\cite{ldp_Frequency_Estimation}. 
As the survey progresses, we can estimate joint attribute dependencies from initially collected data and treat them as \emph{prior knowledge}. 
This prior knowledge is then used to optimize attribute-wise privacy–budget allocation, thereby improving the privacy–utility trade-off (Figure~\ref{fig:us_income_survey}). 
In this setting, distribution shifts primarily arise from \emph{demographic differences} such as development level. 
Then, multiple sources (previously collected state data) provide the basis for learning distributions and their potential uncertainty.
However, currently available CPL analysis algorithms \textbf{cannot} utilize this distributional uncertainty information for privacy leakage analysis, which could lead to privacy vulnerabilities due to incorrect privacy leakage estimates.

\noindent{\textbf{Proposed Solution - }}
We first generalize the optimization problems studied in~\cite{paper_1,Quantifying_DP} to settings where the underlying data distribution is uncertain (Section~\ref{section:motivation}). 
In this regime, computing the optimal CPL is substantially more expensive---with worst-case exponential-time complexity---thereby motivating heuristic methods to estimate CPL. 
We then ask whether the probability distribution can be compressed into a low-dimensional summary that preserves the dependency structure relevant to CPL, enabling fast evaluation. 
Conceptually, this approach mirrors the use of conventional correlation metrics, but seeks to leverage their computational and storage advantages while mitigating their accuracy gaps.

Our main finding is that the pairwise dependency information relevant to CPL assessment (including distribution uncertainty) can be captured by \emph{three parameters}. 
Using these \emph{dependency statistics}---\emph{$(\alpha,\beta,\delta)$–\name\ (\acronym)}---we construct a constant-time estimator (using computed parameters $\alpha,\beta$ and $\delta$) that provides a tight upper bound on CPL for any privacy budget $\varepsilon$. 
Here, the nonnegative coefficients $\alpha$ and $\beta$ quantify complementary aspects of dependence: intuitively, $\alpha$ governs leakage behavior in the low privacy regime (larger $\varepsilon$), while $\beta$ governs the high privacy regime (smaller $\varepsilon$). 
The factor $\delta\in[0,1)$ corrects for sparsity\footnote{Probability distributions become sparse when attribute alphabets increase~\cite{SPM,dss}.} by capturing mass concentration in the distribution (We provide an empirical interpretation of the \acronym\ parameters in Section~\ref{section:Interpreting CPL Through DT Parameters}). 
Given \acronym, Algorithm~\ref{algo:APL_cal_using_db} estimates CPL in \emph{constant time} $O(1)$. 
\acronym\ only occupies \emph{constant space} $O(1)$, in contrast to storing the full $a\times b$ joint distribution at $O(ab)$ to keep dependency information. 
Therefore, constant time and space requirements address \textbf{P1}. 
Experts can calibrate $(\alpha,\beta,\delta)$ from prior, potentially uncertain, distributional knowledge, enabling deployment in real-world LDP data collection and privacy analysis pipelines (addressing \textbf{P2}). 
Finally, we prove that the \acronym-based estimator is an upper bound on CPL, ensuring that reported leakage never understates the true privacy loss and thereby overcoming privacy vulnerabilities of existing methods.

\noindent{\textbf{Our Contributions - }}Next, we summarize our key contributions in this paper as follows:
\begin{enumerate}
    \item[$\bullet$] We show that the pairwise dependency information relevant to CPL assessment, including distributional uncertainty, can be summarised by three parameters. 
    We propose this three-parameter representation, termed the \emph{\name} (\acronym), as a novel metric for modeling attribute dependencies in LDP privacy-leakage analysis. 
    Rather than fully characterizing the underlying joint distribution, \acronym\ serves as a \emph{conservative surrogate} for pairwise CPL, enabling tractable leakage assessment when the full distribution is \emph{unavailable or uncertain}. 
    We further use pairwise analysis as a practical building block to estimate total leakage in multidimensional data. To the best of our knowledge, \acronym\ is the first metric specifically designed to summarize dependencies for LDP privacy leakage analysis in this setting.
    
    \item[$\bullet$] We theoretically prove that \acronym\ possesses several desirable properties---configurable tolerance to uncertainty in prior distributions, constant-time complexity, and compatibility with sparse probability distributions.
    
    \item[$\bullet$] We validate the theoretical guarantees via extensive empirical evaluation on both synthetic and real datasets. 
    Moreover, we demonstrate how to incorporate \emph{distributional uncertainties} in privacy analysis and \emph{calibrate privacy budgets} in real-world data collection scenarios, highlighting the potential practical applicability of \acronym.
\end{enumerate}

\noindent{\textbf{Paper Structure - }}The rest of the paper is organized as follows. First, we review related definitions and theorems in Section~\ref{section:background}. In Section~\ref{section:motivation}, we extend the current optimization problem to address distribution uncertainty and highlight the associated computational challenges. In Section~\ref{section:methodology}, we explain the design steps of \emph{\acronym} to quantify CPL. In Section~\ref{section:Dependency Budget}, we formally define \acronym\ and present algorithms to compute the \acronym,\ and CPL using \acronym. 
In Section~\ref{section:evaluation}, we evaluate the proposed \emph{\acronym} using synthetic and real datasets. Finally, Section~\ref{section:related_works} summarizes the related works and Section~\ref{section:conclusion} concludes the paper.

\section{Preliminaries}\label{section:background}

First, we present the LDP definition in Section~\ref{section:local_differnetial_privacy}. 
Section~\ref{section:Privacy Leakage Caused by Correlated Attributes} defines the privacy leakages in correlated data and presents an existing algorithm to compute the CPL. 
Finally, we discuss and formulate potential distributional uncertainties in prior knowledge and their effects on CPL analysis in Section~\ref{section:Inaccuracies in Distributional Information}.

\subsection{Local Differential Privacy}\label{section:local_differnetial_privacy}

$\varepsilon$-LDP definition appears in various forms in the literature. In this work, we adopt the definition proposed by Duchi et al.~\cite{LDP_duchi}.

\begin{definition}[$\varepsilon$-Local Differential Privacy] \label{defn:LDP}
    Let the input alphabet be $\mathcal{X}$ and the output alphabet be $\mathcal{Y}$. 
    A mechanism $Q$ is $\varepsilon$-locally differentially private if, 
    \begin{equation}
        \sup_{S\subseteq\mathcal{Y},\ x,x'\in\mathcal{X}} \frac{Q(S|x)}{Q(S|x')} \leq e^\varepsilon,
    \end{equation}
    where $Q(S|x) = p(Y=y_i, \ y_i \in S \ | \ X = x)$ is the privatization mechanism. 
    The function $p(\cdot)$ denotes the probability of an event. 
    Here, $x \in \mathcal{X}$, $y_i \in \mathcal{Y} $ and $\varepsilon \geq 0$.
\end{definition}

\subsection{Privacy Leakages in Correlated Attributes}\label{section:Privacy Leakage Caused by Correlated Attributes}

Let $X = \{X_1,\dots, X_n\}$ be $n$ correlated attributes, and let $Y = \{Y_1, \dots,\linebreak Y_n\}$ be perturbed versions of $X$ which are obtained with $\mathcal{M} = \{\mathcal{M}_1, \dots, \mathcal{M}_n\}$ $\varepsilon$-LDP mechanisms respectively. 
Suppose $\mathcal{X}_1, \dots, \mathcal{X}_n$ are alphabets of inputs and $\mathcal{Y}_1,\dots, \mathcal{Y}_n$ are alphabets of outputs respectively. 
Then, we can define ``\emph{direct privacy leakage}'', ``\emph{correlation-induced privacy leakage}'', and ``\emph{total privacy leakage}'' terms~\cite{paper_1}.

\subsubsection{Direct Privacy Leakage (DPL)}

DPL is the extent to which an adversary can infer information about a user's original data from its perturbed data (without correlated attributes' perturbed data). 
Let any attribute $X_k \in X$ be perturbed using an $\varepsilon$-LDP mechanism $\mathcal{M}_k \in \mathcal{M}$. 
Then, the DPL of attribute $X_k$ is at most $\varepsilon$.

\subsubsection{Correlation-Induced Privacy Leakage (CPL)}

CPL is defined as the privacy leakage of an attribute that is caused by correlated attributes. 
Let $L_{Z \rightarrow X_k}$ denote the CPL experienced by any attribute $X_k \in X$ due to a subset of other correlated attributes $Z \subseteq X\setminus \{X_k\}$. 
Let us take $Y_k \in Y$ as the perturbed output of $X_k$, and $W \subseteq Y \setminus \{Y_k\}$ is the perturbed output of $Z$, and $\mathcal{W}$ is the output alphabet of $W$. 
Then $L_{Z \rightarrow X_k}$ can be calculated as 
\begin{equation}\label{eqn:privacy_leakage_additional_xk}
    L_{Z \rightarrow X_k} = \ln \sup_{w\in \mathcal{W},\ x,x'\in\mathcal{X}_k} \frac{p(w|x)}{p(w|x')}.
\end{equation}
Here, the function $p(\cdot)$ denotes the probability of an event.
\subsubsection{Total Privacy Leakage (TPL)}

TPL refers to the privacy leakage of an attribute, including its DPL and CPL. 
Formally, let $L_{X_k}$ denote the TPL of any attribute $X_k \in X$. Computing the exact TPL is challenging due to the curse of dimensionality. 
As an alternative, using the sequential composition theorem~\cite{LDP_survey_composition_theorem}, we can compute an upper bound for $L_{X_k}$, $\overline{L}_{X_k}$ as
\begin{equation}\label{eqn:total_privacy_leakage_composition}
    \overline{L}_{X_k} = \varepsilon + \sum_{x \in X \setminus \{X_k\}} L_{x \rightarrow X_k}.
\end{equation}
Here, $\varepsilon$ is the privacy budget of $\mathcal{M}_k$ mechanism. 
Therefore, we can compute the TPL in any system by using the CPLs between attribute pairs and the privacy budgets of the mechanisms. 
CPL caused by any attribute $\Hat{X} \in X$ about $X_k$ is given by,
\begin{equation}\label{eqn:expanded_additional_privacy_leakage}
    \begin{split}
        L_{\Hat{X} \rightarrow X_k} &= \ln{\sup_{y\in \mathcal{\Hat{Y}},\ x,x'\in\mathcal{X}_k} \left(\frac{\sum_{u \in \mathcal{\Hat{X}}} p( y|u)p(u|x)}{\sum_{v \in \mathcal{\Hat{X}}} p( y|v)p(v|x')}\right)}.
    \end{split}
\end{equation}
Here, $\hat{\mathcal{Y}}$ is the alphabet of perturbed output of $\hat{X}$. 
Next, let us analyze the CPL between two correlated attributes as we can estimate the TPL of any attribute using CPLs and privacy budgets as stated in~\eqref{eqn:total_privacy_leakage_composition}.

This is a \emph{tractable building block} to compute total privacy leakage, which is commonly adopted by recent studies.

\subsubsection{Calculating the CPL Between two Correlated Attributes}\label{section:calculate correlation-induced privacy leakage between two correlated attributes}

Study \cite{paper_1} shows that solving~\eqref{eqn:expanded_additional_privacy_leakage} is a maximization problem of a fraction of linear functions given by, 
\begin{equation}
\ln\sup_{y\in \mathcal{\Hat{Y}},\ x,x'\in\mathcal{X}_k} \frac{C^TG}{C^TG'}.
\end{equation}
Here, $C = (p(y|\hat{x}_1)\text{,}\dots\text{, } p(y|\hat{x}_b))$, $G = \left(G_i\right)^b_{i=1} = (p(\hat{x}_1|x)\text{,}\dots\text{, } \linebreak p(\hat{x}_b|x))$, $G'= \left(G'_i\right)^t_{i=1} = (p(\hat{x}_1|x')\text{,}\dots\text{, } p(\hat{x}_b|x'))$. 
The set $\mathcal{\Hat{X}} = \{\hat{x}_1, \dots, \hat{x}_b\}$ denotes the alphabet of the attribute $\Hat{X}$, with $|\mathcal{\Hat{X}}|=b$, and $x,x' \in \mathcal{X}_k$ where $\mathcal{X}_k$ is the alphabet of $X_k$, with $| \mathcal{X}_k|=a$. 
This maximization problem is equivalent to a simplified form, as shown in~\eqref{eqn:optimization_problem_over_C} by removing $\ln$. 
Here, $S \subseteq [b], \ [b] = \{1,\dots,b\}$. 
The set $S$ corresponds to the indices of $G$ and $G'$ elements which have coefficient $e^\varepsilon$ (i.e., privacy budget $=\varepsilon$), and the remaining (i.e., $[b] \setminus S$) have coefficient $1$. 
\begin{equation} \label{eqn:optimization_problem_over_C}
    \begin{split}
        \underset{S}{\textbf{maximize}} & \quad F(S,G,G', \varepsilon) = \frac{\sum_{i \in [b]\setminus S}G_i+e^{\varepsilon}\sum_{i \in S}G_i}{\sum_{i \in [b]\setminus S}G'_i+e^{\varepsilon}\sum_{i \in S}G'_i}   \\
        \textbf{subject to} 
        &\quad S \subseteq [b].
    \end{split}
\end{equation} 
The optimal solution of~\eqref{eqn:optimization_problem_over_C} indicates the `maximum CPL' caused due to the release of $\hat{X}$ about $X_k$ for given privacy budget $\varepsilon$ when the ground-truth distribution is available. 
Therefore, we refer to this `maximum CPL' as ``$\operatorname{CPL}^\star$'' throughout the remainder of the paper, unless stated otherwise. 
Two existing algorithms---proposed in \cite{Quantifying_DP} and \cite{paper_1}---solve this class of optimization problems with time-complexities of $O(a^2b^2)$ and $O(a^2b\log b)$, respectively. 
Furthermore, \cite{Quantifying_DP} shows that the solution can be modeled as a piece-wise function, which can be solved with $O(\log ab)$, which requires a pre-computation to derive the piece-wise function. 
A simplified version of the algorithm proposed by~\cite{paper_1} is available in Appendix~\ref{section:The Optimal Additional Privacy Leakage Algorithm}. 
Intuitively, this algorithm allocates the indexes $i$ to $S$ which have a larger ratio value for $\frac{G_i}{G'_i}$. 
Moreover, according to this algorithm, $\operatorname{CPL}^\star $ value depends on only two realizations of $X_k$ attribute for a given $\varepsilon$. 
Additionally, we can compute the maximum $L_{\Hat{X}\rightarrow X_k}$ as $\varepsilon$ varies as shown in Corollary~\ref{corollary:max_privacy_leakage_privacy_budget}.

\begin{corollary}\label{corollary:max_privacy_leakage_privacy_budget}
    Let $X_k$ and $\Hat{X}$ be two correlated attributes. 
    Suppose $\Hat{X}$ is perturbed by $\varepsilon$-LDP mechanism $\Hat{M}$. 
    Then, the CPL caused by $\Hat{X}$ about $X_k$ saturates as $\varepsilon$ increases, as shown in \eqref{eqn:corollary:max_privacy_leakage_privacy_budget}. 
    Moreover, this is the maximum $L_{\Hat{X}\rightarrow X_k}$ as $\varepsilon$ varies (i.e., $\max_{\varepsilon} L_{\Hat{X}\rightarrow X_k}$),
    \begin{equation}\label{eqn:corollary:max_privacy_leakage_privacy_budget}
        \lim_{\varepsilon \rightarrow \infty} L_{\Hat{X}\rightarrow X_k} =  \ln \max_{\Hat{x}\in \mathcal{\Hat{X}},\ x,x'\in\mathcal{X}_k} \frac{p(\Hat{x}|x)}{p(\Hat{x}|x')}.
    \end{equation}
    Here, $\mathcal{\Hat{X}}$ and $\mathcal{X}_k$ are alphabets of  $\Hat{X}$ and $X_k$, respectively. 
\end{corollary}
The proof of Corollary~\ref{corollary:max_privacy_leakage_privacy_budget} is available in~\cite{paper_1}.

\subsection{Inaccuracies in Distributional Information}\label{section:Inaccuracies in Distributional Information}

In LDP deployments, a prior (learned or public) rarely matches the ground-truth distribution of the target population. 
Mismatch can arise from \textit{demographic differences}, \textit{sampling bias}, \textit{temporal drift} or \textit{stale knowledge}. 
When the prior does not reflect the actual dependencies among attributes, CPL computed from the prior can be shifted—either \emph{underestimated} or \emph{overestimated}.

Let $P_A (\in \mathbb{R}^{a\times b})$ and $P_K (\in \mathbb{R}^{a\times b})$ be the actual and known conditional distribution between $X_k$ and $\hat{X}$ (i.e., $P(\hat{X}|X_k)$) respectively.
Let $\operatorname{CPL} ^\star$ denote the estimated CPL of the $P_A$ and $\operatorname{CPL}' $ denote the estimated CPL of $P_K$. 
We say CPL is \emph{underestimated} if $\operatorname{CPL}' <\operatorname{CPL} ^\star$ (calibration may grant more budget than is safe), and \emph{overestimated} if $\operatorname{CPL}'>\operatorname{CPL}^\star$ (calibration is conservative, possibly reducing utility). 
Since the sign of misspecification is generally unknown, mitigating the \emph{underestimation} risk guarantees the privacy. 

We formulate the distributional mismatch as an uncertainty set around the known distribution $P_K (\in \mathbb{R}^{a\times b})$. 
Let $\Delta (\in \mathbb{R}^{a\times b})$ be the individual entry-wise difference between $P_A$ and $P_K$ (i.e., $\Delta \geq |P_A-P_K|$). 
This provides fine-grained variations for $P_A$ from $P_K$. 
For example, we can define relative uncertainty with respect to the known distribution $\Delta = P_K\times 10\%$.
Further details are available in Section~\ref{section:Defining_Distribution_Uncertainty}.

\section{Motivation}\label{section:motivation}

In this section, we extend the CPL analysis optimization problem~\eqref{eqn:optimization_problem_over_C} with distribution uncertainty. 
Next, we show that solving this problem is worst-case exponential time. 
This motivates the development of a heuristic metric to quantify the dependency between attributes.

\subsection{Updated Optimization Problem}
We can update the CPL analysis optimization problem~\eqref{eqn:optimization_problem_over_C} as 
\begin{equation} \label{eqn:optimization_problem_over_C_G_G_prime}
    \begin{split}
        \underset{S, G, G'}{\textbf{maximize}} & \quad F(S,G,G', \varepsilon) = \frac{\sum_{i \in [b]\setminus S}G_i+e^{\varepsilon}\sum_{i \in S}G_i}{\sum_{i \in [b]\setminus S}G'_i+e^{\varepsilon}\sum_{i \in S}G'_i}   \\
        \textbf{subject to} 
        &\quad S \subseteq [b],\\
        &\quad h_i-\Delta_{\hat{x}_i,x} \leq G_i\leq h_i+\Delta_{\hat{x}_i,x}; \forall i \in [b],\\
        &\quad h'_i-\Delta_{\hat{x}_i,x'} \leq G'_i\leq h'_i+\Delta_{\hat{x}_i, x'}; \forall i \in [b],\\
        &\quad\sum_{i\in [b]} G_i = 1,\ \sum_{i\in [b]} G'_i = 1. 
    \end{split}
\end{equation} 
Here, $G = \left(G_i\right)^b_{i=1}, G' = \left(G'_i\right)^b_{i=1}, h_i = P_K(\hat{X} =\hat{x}_i|X_k = x), h'_i = P_K(\hat{X} =\hat{x}_i|X_k = x')$, $\Delta_{\hat{x}_i,x}$ denotes the uncertainty of $P_K(\hat{X}=\hat{x}_i|X_k=x)$, and $P_K$ denotes the known distribution, respectively.

\subsection{Time Complexity Analysis}

This optimization problem is inherently non-convex, as the objective function involves a bilinear fractional form. Consequently, computing the optimal CPL becomes substantially more expensive than in the fixed $(G, G')$ setting (i.e., without distributional uncertainty), where the problem reduces to a linear fractional program (LFP).

A potential approach is to decompose the original problem into quasi-convex sub-problems by fixing the subset $S$ and solving each sub-problem independently.
The global optimum can then be obtained by taking the maximum across all $S$. For a fixed $S$, maximizing $F(\cdot)$ requires increasing the numerator term $\sum_{i \in S} G_i$ while simultaneously minimizing the denominator term $\sum_{i \in [b] \setminus S} G'_i$. 
However, the number of possible subsets $S$ grows exponentially with $b$, leading to a worst-case complexity of $O(2^b)$. 
Similarly, if $G$ and $G'$ are fixed, an optimal solution can be computed efficiently using standard methods for linear fractional programs (LFPs), yet exploring all possible combinations of their bounds still incurs $O(2^b)$ complexity.

Since no straightforward convex relaxation exists, the standard approach to handle bilinear fractional programs (BFPs) is to linearise the bilinear terms—typically via McCormick relaxations or related techniques—yielding a mixed-integer linear program (MILP). 
While this transformation enables the use of off-the-shelf MILP solvers, it remains computationally expensive, exhibiting exponential complexity in the worst case~\cite{mccormick1976computability}.

\subsection{Our Approach}

Heuristic approaches are usually adopted to solve these types of problems. 
In this study, we focus on extracting CPL-related information from distribution and utilize it to compute the CPL in time efficiently, similar to conventional correlation metrics, while mitigating the privacy risks associated with them.

\section{Methodology}\label{section:methodology}

This section explains the design steps and theoretical proofs of ``\name'': Section~\ref{section:Novel_Metric_to_Represent_the Dependency_Between_Attributes} analyzes and defines the requirements that the metric must fulfil. 
Next, Section~\ref{section:Approximating the Distribution} derives a compact representation for the probability distribution. 
Section~\ref{section:Handling Inaccuracies in Prior Distribution} adjusts the metric to handle inaccuracies in distributions.

\subsection{Analyzing and Defining the Requirements of the Metric} \label{section:Novel_Metric_to_Represent_the Dependency_Between_Attributes}

First, we analyze the requirements that need to be satisfied by the metric to address \textbf{P1} and \textbf{P2} (See Section~\ref{section:introduction}). 
Then, we define four objectives that must be achieved by the metric based on this analysis. 

For \textbf{P1-Limited to specific LDP mechanisms or scalability limitations}: A key challenge is to compress the CPL-related information in the distribution into a smaller set of parameters while incurring minimal information loss. Furthermore, these parameters should represent the dependencies between attributes, allowing an algorithm to estimate CPL in a more time-efficient manner. Ensuring a privacy leakage guarantee is another essential requirement, as the conventional correlation metrics (e.g., MI) lack such guarantees~\cite{paper_1}. Therefore, we define three decomposed objectives that should be achieved by the metric to address P1:
\begin{enumerate}
    \item [$\bullet$] \textbf{OB1 -} The metric should provide a compressed representation of probability distribution using `\emph{a fixed}' and `\emph{a smaller}' set of parameters,
    \item [$\bullet$] \textbf{OB2 -} The estimated CPL of the metric should be `\emph{a tight upper bound}' for the $\operatorname{CPL}^\star$ to maximize the accuracy and avoid any privacy risk, 
    \item [$\bullet$] \textbf{OB3 -} CPL computation should be time efficient with the parameters.
\end{enumerate}

For \textbf{P2-Requires ground truth probability distributions
between the attributes}: The metric should be capable of tolerating the inaccuracies in the prior known distributions. Therefore, as the fourth objective, 
\begin{enumerate}
    \item [$\bullet$] \textbf{OB4 -} The metric should be capable of tolerating the inaccuracies in the prior known distribution.
\end{enumerate}

Next, the following sections outline the steps taken to develop the novel metric, achieving each objective (i.e., OB1, OB2, OB3, and OB4).

\subsection{Approximating the Distribution}\label{section:Approximating the Distribution}

\begin{figure}[t]
  \centering
  \includegraphics[trim={0.8cm 0.3cm 0.7cm 0.2cm},width=0.95\linewidth]{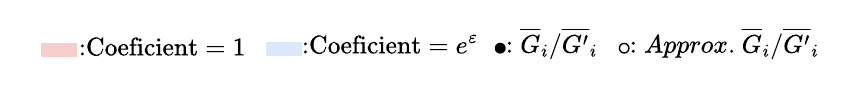}\par
  \makebox[\linewidth][c]{%
    \begin{subfigure}[t]{0.33\columnwidth}
      \centering
      \includegraphics[trim={1.75cm 0.8cm 1.2cm 0.8cm},clip,width=1\textwidth]{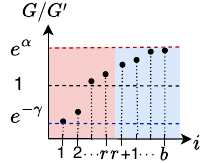}
      \vspace{-1.5em}
      \caption{}\label{fig:g_g'_coefficients_a}
    \end{subfigure}
    \hfill
    \begin{subfigure}[t]{0.33\columnwidth}
      \centering
      \includegraphics[trim={1.8cm 0.8cm 1.2cm 0.8cm},clip,width=1\textwidth]{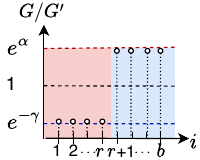}
      \vspace{-1.5em}
      \caption{}\label{fig:g_g'_coefficients_b}
    \end{subfigure}
    \hfill
    \begin{subfigure}[t]{0.33\columnwidth}
      \centering
      \includegraphics[trim={1.8cm 0.8cm 1.2cm 0.8cm},clip,width=1.01\textwidth]{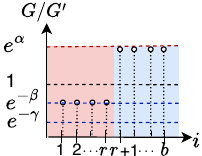}
      \vspace{-1.5em}
      \caption{}\label{fig:g_g'_coefficients_c}
    \end{subfigure}%
  }

  \caption{(a) An example for sorted (ascending order) ratio vector $Q$ for a given $\overline{G},\overline{G'}$. Here, $\overline{G} = (\overline{G_1},\dots,\overline{G_b})$ and $\overline{G'} = (\overline{G'_1},\dots,\overline{G'_b})$. (b) Approximated $(\overline{G}/\overline{G'})$ values into $e^{\alpha}$ and $e^{-\gamma}$, where $\alpha=\ln \max_{i\in [b]}\left(\overline{G_i}/\overline{G'_i}\right)$ and $\gamma=-\ln \min_{i\in [b]}\left(\overline{G_i}/\overline{G'_i}\right)$.
  (c) Approximated $(\overline{G_i}/\overline{G'_i})$ values into $e^{\alpha}$ and $e^{-\beta}$, where $\beta$ is the calibrated parameter using Theorem~\ref{thm:beta}.}
  \label{fig:g_g'_coefficients}
\end{figure}

In the metric development, we consider the same two-attribute scenario. Let $X_k$ and $\Hat{X}$ be two correlated attributes, and suppose we aim to estimate CPL of $X_k$ caused by the release $\Hat{Y}$ of $\Hat{X}$ (i.e., $L_{\hat{X}\rightarrow X_k}$). Here, $\Hat{Y}$ is the LDP version of $\Hat{X}$, which is obtained with $\varepsilon$-LDP mechanism $\hat{\mathcal{M}}$. Let $\overline{G} = (p(\hat{x}_1|x)\text{, }\dots\text{, }p(\hat{x}_b|x))$, $ \overline{G'}= (p(\hat{x}_1|x')\text{, }\dots\text{, }\linebreak p(\hat{x}_b|x'))$, and let $\mathcal{\Hat{X}} = \{\hat{x}_1, \dots, \hat{x}_b\}$ denotes the alphabet of $\Hat{X}$, with $|\mathcal{\Hat{X}}|=b$, and $x,x' \in \mathcal{X}_k$, where $\mathcal{X}_k$ is the alphabet of $X_k$, with $| \mathcal{X}_k|=a$. Then, $L_{\hat{X}\rightarrow X_k}$ is given by~\eqref{eqn:optimization_problem_over_C}. As the first step (\textbf{OB1}), we need to determine the most effective way of approximating the probability distribution using a smaller number of parameters while ensuring the privacy guarantee (\textbf{OB2}). Here, we consider two factors: 

(i) \textit{"How do the existing algorithms compute the CPL?" } According to Algorithm~\ref{algo:related_algorithm_1}, we can converge to the optimal solution of the optimization problem~\eqref{eqn:optimization_problem_over_C} by using the element-wise division between $\overline{G}$ and $\overline{G'}$~\cite{Quantifying_DP,paper_1}. Moreover, the CPL for a given privacy budget depends only on two realizations of the $X_k$ attribute.

(ii) \textit{"How does CPL behave at the boundary cases?"}. Corollary~\ref{corollary:max_privacy_leakage_privacy_budget} shows that CPL converges to $\ln \max_{\Hat{x}\in \mathcal{\Hat{X}},\ x,x'\in\mathcal{X}_k} \frac{p(\Hat{x}|x)}{p(\Hat{x}|x')}$ as privacy budget (i.e., $\varepsilon$) goes to infinity. Therefore, $\max_{\Hat{x}\in \mathcal{\Hat{X}},\ x,x'\in\mathcal{X}_k} \frac{p(\Hat{x}|x)}{p(\Hat{x}|x')}$ is an essential information that should be preserved within the metric to characterize CPL at higher privacy budgets.

These two factors imply that we can use $Q = \left(\frac{\overline{G_1}}{\overline{G'_1}},\dots,\frac{\overline{G_b}}{\overline{G'_b}}\right)$ ratios\footnote{We assume that ratios are well-defined, i.e., $\overline{G'_i} \neq 0 \ \forall i \in [b]$. We extend the analysis to the cases where $\overline{G'_i} = 0$ for some $i$ in Section~\ref{section:Handling Sparse Distribution}.} instead of individual $\bigl(\overline{G_i}\bigr)_{i=1}^b$ and $\bigl(\overline{G'_i}\bigr)_{i=1}^b$ values to estimate CPL with \emph{no privacy risks} (i.e., upper bound for $\operatorname{CPL}^\star $). However, there could be many $G,G'$ vectors that share the same $Q$. Here, $G$ and $G'$ are probability distribution vectors of size $b$. 

For example, suppose $b=3$ and $Q$ is $\{0.5, 3, 2\}$ for given $\overline{G},\overline{G'}$. There are many other vector pairs give these same ratio values e.g., $\left\{G=\left(\frac{3}{8}, \frac{3}{8}, \frac{1}{4}\right), G'=\left(\frac{3}{4}, \frac{1}{8}, \frac{1}{8}\right)\right\}$, $\left\{G=\left(\frac{19}{50}, \frac{21}{50}, \frac{1}{5}\right), G'=\left(\frac{38}{50}, \frac{7}{50}, \frac{1}{10}\right)\right\}$ etc.  

Therefore, let the set $\mathcal{H}$ denote the set of all pairs of $G,G'$ sharing the ratio $Q$. Then, we need to find the maximum CPL value of all $G,G'\in\mathcal{H}$ to ensure \textbf{OB2}. Theorem~\ref{thm:ratios} states that the maximum CPL of all $G,G'\in\mathcal{H}$ is $\ln \frac{e^{\alpha} - 1+e^{\alpha+\varepsilon}(e^{\gamma} - 1)}{e^{\varepsilon}(e^{\gamma} - 1)+ e^{\gamma}(e^{\alpha}-1)}$ which is similar as CPL given by the boundary values of $Q$ (i.e., $\tilde{Q} \in \{e^\alpha, e^{-\gamma}\}^b$, where $\alpha = \ln \max_{q \in Q} (q)$, and $ \gamma = -\ln \min_{q \in Q} (q)$). Since the maximum CPL value of $G,G' \in \mathcal{H}$ does not depend on $b$ and $r$, we only need $\alpha$ and $\gamma$ to compute it. In other words, we can approximate the original $Q$ in Figure~\ref{fig:g_g'_coefficients_a} as its boundary values as shown in Figure~\ref{fig:g_g'_coefficients_b}. Therefore, this approximation reduces the parameter count from $ab$ to \textbf{two} ($\alpha$ and $\gamma$) that are required to characterize the dependency between two attributes. Next, we discuss $\alpha$ and $\gamma$ characterization of probability distributions. 

\begin{theorem}\label{thm:ratios}
Let $X_k$ and $\hat X$ be discrete, correlated attributes with alphabets $\mathcal X_k$ and $\hat{\mathcal X}=\{\hat{x}_1,\dots,\hat{x}_b\}$, respectively. For $x,x'\in\mathcal X_k$ define the conditional probability vectors $\overline{G}=(\overline{G_i})_{i=1}^{b}, \quad \overline{G'}=(\overline{G_i'})_{i=1}^{b}, \text{ with }\overline{G_i}=p(\hat{x}_i|x),\overline{G_i'} =p(\hat{x}_i|x')$. Let $Q$ denotes the element‑ wise ratio vector $Q = \overline{G}\oslash \overline{G'} := \Bigl(\frac{\overline{G_1}}{\overline{G_1'}},\dots,\frac{\overline{G_b}}{\overline{G_b'}}\Bigr)$. Let $\alpha = \ln \max_{q \in Q} (q)$, $ \gamma = -\ln \min_{q \in Q} (q)$, and let $\tilde Q$ be any vector in $\{e^\alpha,e^{-\gamma}\}^{b}$. Let $J$ be the set that contains all the probability distributions of size $b$. Define $\mathcal{H}:=\bigl\{(G, G'): G, G' \in J,\ G\oslash G'=Q\bigr\},$ the set of all pairs of probability vectors sharing the ratio $Q$, and define $\mathcal{\tilde{H}}:=\bigl\{(G, G'): G, G' \in J,\ G\oslash G'=\tilde{Q}\bigr\},$ the set of all pairs of positive probability vectors sharing the ratio $\tilde{Q}$. Then the \emph{correlation-induced privacy leakage (CPL)} satisfies $\max_{(G,G')\in\mathcal{H}}\operatorname{CPL}(G,G')=\max_{(G,G')\in\mathcal{\tilde{H}}}\operatorname{CPL}(G,G')= \ln \frac{e^{\alpha} - 1+e^{\alpha+\varepsilon}(e^{\gamma} - 1)}{e^{\varepsilon}(e^{\gamma} - 1)+ e^{\gamma}(e^{\alpha}-1)}$. Here, $\operatorname{CPL}(G,G')$ denotes the CPL of $G$ and $G'$.
\end{theorem}

\noindent\emph{Proof sketch.}
We first relax the feasible likelihood ratios to the interval $[e^{-\gamma},e^\alpha]$ and solve the resulting linear-fractional optimization problem. 
Monotonicity shows that the optimum occurs at the extreme ratios $e^\alpha$ and $e^{-\gamma}$, yielding the stated closed-form bound independently of the choice of $S$. 
Finally, we show that this bound is tight for the original problem by reducing it to an equivalent two-point instance supported on the maximum and minimum ratios. The complete proof is provided in Appendix~\ref{proof:thm:ratios}.

\subsubsection{Characterization Using \texorpdfstring{$(\alpha, \gamma)$}{(alpha, gamma)}}\label{section:Metric Using alpha gamma}

$(\alpha,\gamma)$ is a simplified representation of a conditional probability distribution between two attributes, which can be used to compute an upper bound for $\operatorname{CPL}^\star $ by, 
\[\ln\frac{e^{\alpha} - 1+e^{\alpha+\varepsilon}(e^{\gamma} - 1)}{e^{\varepsilon}(e^{\gamma} - 1)+ e^{\gamma}(e^{\alpha}-1)}\]
based on Theorem~\ref{thm:ratios}. Here, $\alpha$ should be larger than $\gamma$ when selecting the parameters according to Corollary~\ref{corollary:alpha_greater_than_beta}. Next, we explore the methods to handle sparse distribution in Section~\ref{section:Handling Sparse Distribution}.

\begin{corollary}\label{corollary:alpha_greater_than_beta}
    Let $\alpha_1,\gamma_1$ and $\alpha_2,\gamma_2$ be two ($\alpha,\gamma$) characterization values for $G,G'$ such that $\alpha_1 = \ln max_{q \in \overline{G}\oslash \overline{G'}} (q),\ \gamma_1 =\linebreak -\ln min_{q \in \overline{G}\oslash \overline{G'}} (q)$, and $\alpha_2 = \ln max_{q \in \overline{G'}\oslash \overline{G}} (q),\ \gamma_2 =\linebreak -\ln min_{q \in \overline{G'}\oslash \overline{G}} (q)$. Therefore, $\alpha_2 = \gamma_1$ and $\gamma_2 = \alpha_1$. Suppose $U_1$ corresponds to the CPL of $\alpha_1,\gamma_1$ and $U_2$ corresponds to the CPL of $\alpha_2,\gamma_2$. Then, $\alpha_1 \geq \gamma_1 \iff U_1 \geq U_2$.
\end{corollary}

The proof of Corollary~\ref{corollary:alpha_greater_than_beta} is available in Appendix~\ref{proof:corollary:alpha_greater_than_beta}.

\subsubsection{Handling Sparse Distribution}\label{section:Handling Sparse Distribution}

Sparse probability distributions are common as the alphabet size increases~\cite{SPM,dss}. Such distributions can lead to an undefined $\alpha$ and $\gamma$ parameter when the denominator becomes zero. To handle such scenarios, we introduce two new vectors $(\tilde{\delta}_1,\dots,\tilde{\delta}_t)$ and $(\hat{\delta}_1,\dots,\hat{\delta}_t)$ as
\begin{equation}\label{eqn:delta_constraints}
\begin{split}
        \overline{G_i} &\leq e^{\alpha}\overline{G'_i} + \tilde{\delta}_i;\ \forall \overline{G_i} \in \overline{G},\ \overline{G'_i} \in \overline{G'}, \ \alpha>0,\\
        \overline{G'_i} &\leq e^{\gamma}\overline{G_i} + \hat{\delta}_i;\ \forall \overline{G_i} \in \overline{G},\ \overline{G'_i} \in \overline{G'}, \ \gamma>0.
\end{split}
\end{equation}
Here, 
\begin{equation}\label{eqn:define_delta}
    \tilde{\delta}_i = \begin{cases}
        \overline{G_i}  &; \text{ if}\quad \overline{G'_i} = 0,\\
        0  &; \text{ otherwise,}
    \end{cases},\text{ and }
    \hat{\delta}_i = \begin{cases}
        \overline{G'_i}  &; \text{ if}\quad \overline{G_i} = 0,\\
        0  &; \text{ otherwise.}
    \end{cases}
\end{equation}
Then, we can compute $\alpha$ and $\gamma$ as 
\begin{equation}\label{eqn:gamm_and_alpha_with_delta}
    \alpha = \ln \max_{\forall i \in [b], \overline{G'_i}\neq 0} \frac{\overline{G_i}}{\overline{G'_i}},\quad \gamma = \ln \max_{\forall i \in [b], \overline{G_i}\neq 0} \frac{\overline{G'_i}}{\overline{G_i}}.
\end{equation}
For the simplicity, we take the maximum value $\delta =\linebreak \max\left\{\sum_{i \in [b]}\tilde{\delta}_i,\ \sum_{i \in [b]}\hat{\delta}_i\right\}$, instead of two values. Then we can estimate the CPL using $(\alpha,\gamma,\delta)$ as stated in Theorem~\ref{thm:alpha_beta_delta_apl}.

\begin{theorem}\label{thm:alpha_beta_delta_apl}
Let $X_k$ and $\hat X$ be discrete, correlated attributes. Let $\alpha,\gamma$ be characterization parameters for $L_{\hat{X}\rightarrow X_k}$. Let $(\tilde{\delta}_1,\dots,\tilde{\delta}_b)$ and $(\hat{\delta}_1,\dots,\hat{\delta}_b)$ be parameters introduced to handle sparsity with $\delta = \max\left\{\sum_{i \in [b]}\tilde{\delta}_i,\sum_{i \in [b]}\hat{\delta}_i\right\}$. Then the CPL characterized by $(\alpha,\gamma,\delta)$ can be computes as
\begin{equation}\label{eqn:apl_cal_delta_available}
    \operatorname{CPL}=
    \begin{cases} 
     \ln \bigl(1+\delta(e^\varepsilon-1)\bigr)
    \text{; if }  e^\alpha-1-\delta(e^\varepsilon-1) < 0, \\
     \ln \left(\frac{e^{\alpha} + \delta - e^{\alpha}\delta - \delta e^{\varepsilon}
      + e^{\alpha+\varepsilon}(-1 + e^{\gamma} + \delta)-1} {\delta(e^{\varepsilon}-1) - e^{\varepsilon}
      + e^{\gamma}(-1 + e^{\alpha} + \delta + e^{\varepsilon} - \delta e^{\varepsilon})}\right)\text{; otherwise.}
    \end{cases}
\end{equation}
\end{theorem}

\noindent\emph{Proof sketch.}
For a fixed subset $S$, we aggregate the variables as $g=\sum_{i\in S}G_i$ and $g'=\sum_{i\in S}G'_i$, reducing the original multidimensional problem to a one-dimensional linear-fractional optimization. 
The constraints provide two affine upper bounds on $g$ in terms of $g'$, and monotonicity shows that the optimum occurs at a boundary, their intersection, or the saturation point $g=1$, resulting in a piecewise solution. We then construct feasible $G$ and $G'$ attaining this solution and optimize over $S$ by collecting all $\tilde{\delta}_i$ terms inside $S$ and all $\hat{\delta}_i$ terms outside $S$. 
Finally, bounding both aggregate slack terms by $\delta=\max\{\tilde{\delta},\hat{\delta}\}$ eliminates the infeasible cases and gives the stated CPL bound. 
The complete proof is provided in Appendix~\ref{proof:thm:alpha_beta_delta_apl}.

\subsubsection{Limitations of \texorpdfstring{$(\alpha, \gamma, \delta)$}{(alpha, gamma, delta)} characterization}\label{section:Limitations of_alpha,gamma_delta}
The main challenge with ($\alpha, \gamma, \delta$) characterization is the reduction of estimation accuracy of $\operatorname{CPL} ^\star$. To illustrate this, let us consider the following toy example.

\textbf{Toy Example - } Let $X_k$ and $\hat{X}$ be two correlated attributes, and let $\mathcal{X}_k=\{x,x'\}$ and  $\mathcal{\hat{X}}=\{\hat{x}_1,\hat{x}_2,\hat{x}_3,\hat{x}_4\}$ be their alphabets respectively. Then, suppose we need to compute CPL of $X_k$ caused by the release of perturbed $\hat{X}$ (i.e., $L_{\Hat{X}\rightarrow X_k}$) for two scenarios: 

\textit{Scenario 1} - $\overline{G} = (0.1, 0.1, 0.15, 0.2, 0.45), \overline{G'} = (0.4,\linebreak 0.3, 0.1, 0.1, 0.1)$, where $\overline{G}$ and $\overline{G'}$ are the probability distribution $P(\Hat{X}|X_k=x)$ and $P(\Hat{X}|X_k=x')$, respectively. Then $Q=\left( \frac{\overline{G}_i}{\overline{G'}_i}\right)^4_{i=1}=(0.25, 0.33, 1.5, 2.0, 4.5)$.

\textit{Scenario 2} - $\overline{G} = (0.0025, 0.2975, 0.48, 0.13, 0.09), \overline{G'} = (0.01, 0.39,\linebreak 0.47, 0.11, 0.02)$ where $\overline{G}$ and $\overline{G'}$ are the probability distribution $P(\Hat{X}|X_k=x)$ and $P(\Hat{X}|X_k=x')$ of Scenario 2, respectively. Then $Q=\left( \frac{\overline{G_i}}{\overline{G'_i}}\right)^4_{i=1}=(0.25, 0.58, 1.02,  2.09, 4.5)$.

Note that both scenarios have the \emph{same} $\alpha=1.50,\gamma=1.38, \delta = 0$. However, $\operatorname{CPL}^\star $ is significantly different in two scenarios as shown in Figure~\ref{fig:toy_example_a} and Figure~\ref{fig:toy_example_a}. The main reason for this is the information loss due to representing the distribution with ratios $\left(\frac{\overline{G}_i}{\overline{G'}_i}\right)^b_{i=0}$ as it \emph{only} conveys the information of \emph{extremes}. Thus, we propose a calibration step to address this issue in the next section.

\begin{figure}[t]
  \centering
  \includegraphics[trim={0.8cm 0.3cm 0.7cm 0.2cm},width=0.55\linewidth]{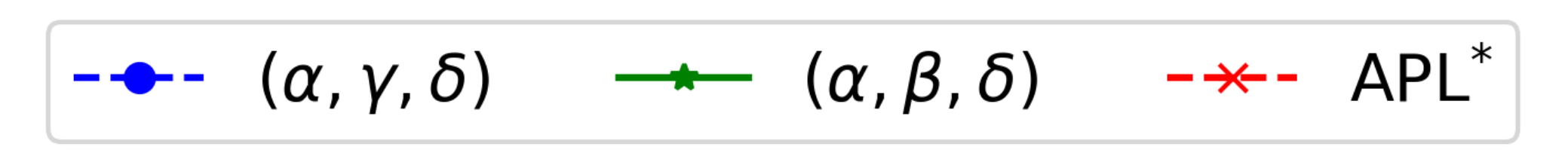}\par
  \makebox[\linewidth][c]{%
    \hfill
    \begin{subfigure}[t]{0.45\columnwidth}
      \centering
      \includegraphics[trim={1.3cm 0.8cm 1.2cm 0.8cm},clip,width=1\textwidth]{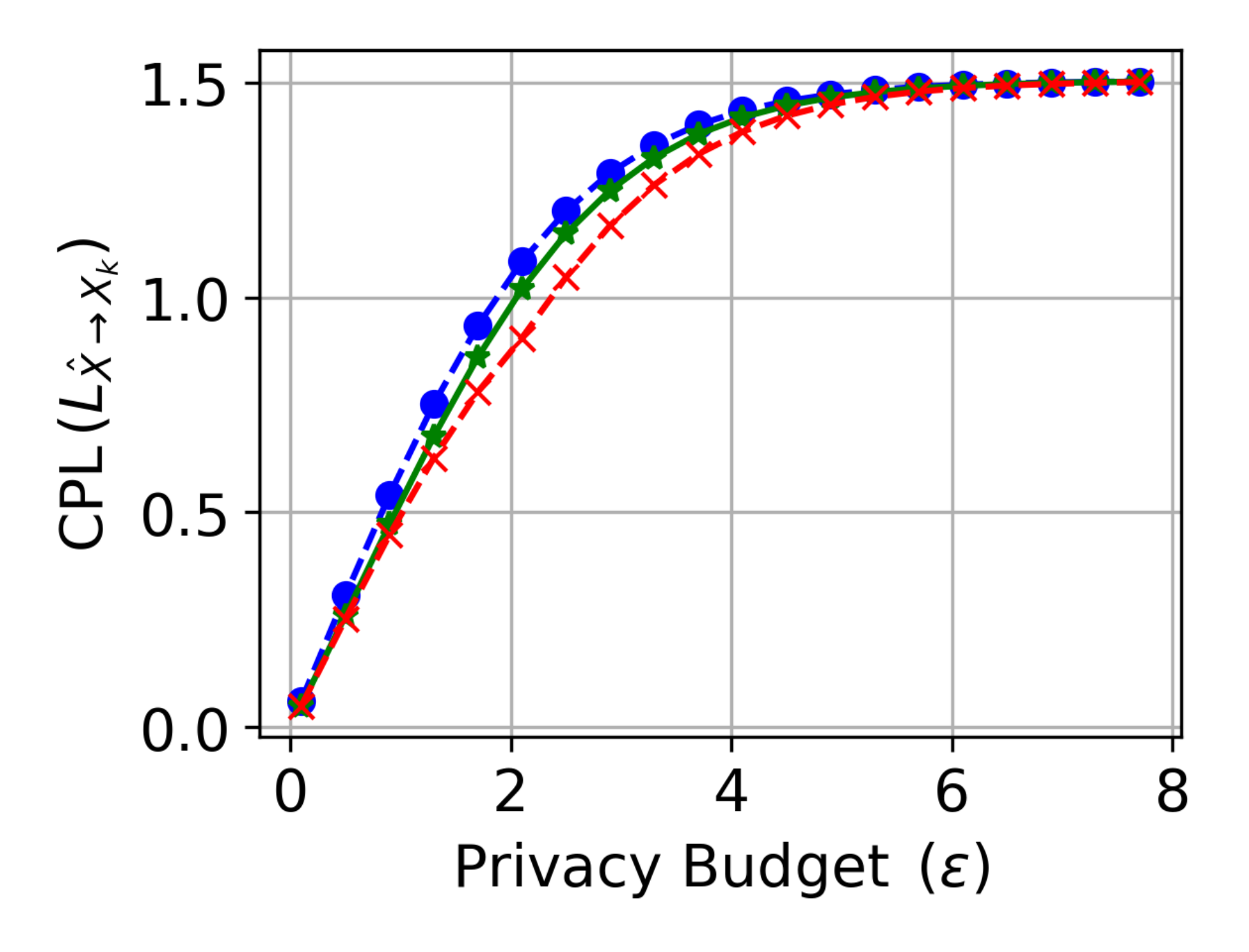}
      \vspace{-1.5em}
      \caption{}\label{fig:toy_example_a}
    \end{subfigure}
    \hfill
    \begin{subfigure}[t]{0.45\columnwidth}
      \centering
      \includegraphics[trim={1.3cm 0.8cm 1.2cm 0.8cm},clip,width=1\textwidth]{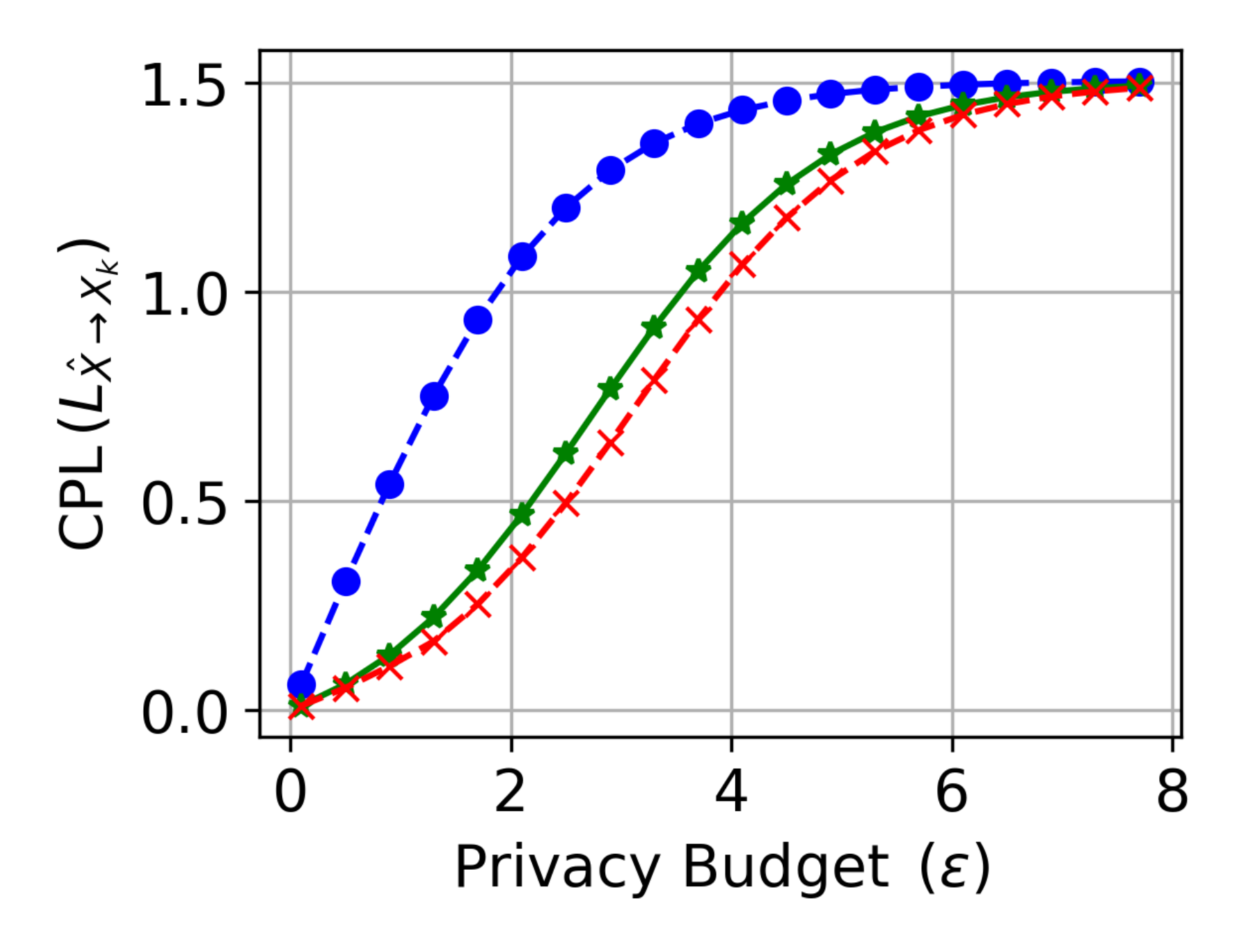}
      \vspace{-1.5em}
      \caption{}\label{fig:toy_example_b}
    \end{subfigure}%
    \hfill
  }
  \caption{Results of the toy example. (a) $L_{\hat{X}\rightarrow X_k}$ vs $\varepsilon$ for Scenario 1. (b) $L_{\hat{X}\rightarrow X_k}$ vs $\varepsilon$ for Scenario 2.
  }\label{fig:toy_example}
\end{figure}

\subsubsection{Calibrated Characterization Using \texorpdfstring{$(\alpha,\beta,\delta)$}{(alpha, beta, delta)}}\label{section:Calibrated Characterisation Using}

First, let us revisit the approximation approach of the $(\alpha,\gamma,\delta)$ method.

``$\alpha$'' is crucial to represent the privacy leakage at higher $\varepsilon$ according to the Corollary~\ref{corollary:max_privacy_leakage_privacy_budget}. Therefore, modifying $\alpha$ affects the privacy guarantee at higher privacy budgets.

"$\gamma$" also contributes to representing the strength of the correlation between two attributes. However, it is not a crucial parameter, as $\alpha$. According to Corollary~\ref{corollary:apl_monotonially_increases_with_beta}, privacy leakage monotonically increases with $\gamma$. As discussed in Section~\ref{section:Limitations of_alpha,gamma_delta}, the main issue with $(\alpha, \gamma, \delta)$ characterization is the over-estimate of the correlation. Adjusting the lower bound (i.e., $e^{-\gamma}$) such that it represents the average impact of the $Q=\left(\frac{\overline{G}_i}{\overline{G'}_i}\right)_{i=1}^{b}$ ratios could possibly improve the accuracy.

"$\delta$" corresponds to the edge cases that do not directly cause the significant estimation error.

In our analysis, we notice that calibrating the minimum level of $Q$, $e^{-\gamma}$ with respect to the $\operatorname{CPL}^\star $ (taking the calibrated minimum level as $e^{-\beta}$) at an $\varepsilon=\varepsilon_0$, improves the accuracy of CPL estimation while holding the privacy guarantee for $\varepsilon \geq \varepsilon_0$ as stated in Theorem~\ref{thm:beta}. Intuitively, we compute intermediate level $e^{-\beta}$ in between $e^{-\gamma}$ and $1$ such that $(\alpha,\beta,\delta)$ estimate CPL more accurately as shown in Figure~\ref{fig:g_g'_coefficients_c}. Here, we can estimate the CPL of $(\alpha,\beta, \delta)$ using Corollary~\ref{corollary:alpha_beta_delta_apl} similar to $(\alpha,\gamma, \delta)$. 

Using Corollary~\ref{corollary:beta_upper_bound}, we can compute $\beta$ at $\varepsilon \rightarrow 0$, and it is always less than or equal to $\gamma$. This is an important result, as the estimated CPL guarantees privacy for $\varepsilon \in (0,\infty)$.
Since $\beta$ is always less than or equal to $\gamma$, it is ensured that $(\alpha,\beta, \delta)$ provides better estimation (Note that estimated CPL is monotonically increasing with $\gamma$ as stated in Corollary~\ref{corollary:apl_monotonially_increases_with_beta}).
Next, let us estimate CPL using the $(\alpha,\beta, \delta)$ method for the running toy example (Section~\ref{section:Limitations of_alpha,gamma_delta}). 

\textbf{Toy Example (Continue) - }
Let us compute $\beta$ for both Scenario 1 and 2 at $\varepsilon_0 = 0$ using Corollary~\ref{corollary:beta_upper_bound}. As shown in Figure~\ref{fig:toy_example_a} and Figure~\ref{fig:toy_example_b}, the estimated CPL of $(\alpha, \beta, \delta)$ method (using Corollary~\ref{corollary:alpha_beta_delta_apl}) is a tighter upper bound and adapts to the changes of $Q$. \textit{Therefore, the calibration step mitigates the uncontrolled overestimation issue of the $(\alpha,\gamma, \delta)$ estimation method}.

\begin{theorem}\label{thm:beta}
Let the metric be calibrated at privacy budget $\varepsilon_0$ and let $l$ be the $\operatorname{CPL}^\star $ at $\varepsilon_0$. Then, the calibrated lower level $\beta$ is given by, $\beta =\ln \frac{-1 + e^{\alpha} + \delta + e^{l} \delta - e^{\alpha} \delta - e^{\varepsilon_0} \left( (e^{l} - e^{\alpha})(-1 + \delta) + \delta \right)}{-e^{\varepsilon_0} \left( e^{\alpha} + e^{l} (-1 + \delta) \right) + e^{l} (-1 + e^{\alpha} + \delta)}$.
The estimated CPL of ($\alpha, \beta, \delta$) method is an upper bound for the $\operatorname{CPL}^\star $ for $\varepsilon \geq \varepsilon_0$.
\end{theorem}

\noindent\emph{Proof sketch.}
We first calibrate $\beta$ by equating the analytical CPL expression to the exact leakage $\operatorname{CPL}^{\star}$ at the calibration budget $\varepsilon_0$ and solving for $\beta$. 
We then express the calibrated CPL through a scalar parameter $g'$ and show that, within each segment where the selected aggregates $A$ and $B$ remain fixed, $g'$ decreases as the calibration budget increases. 
Since the CPL expression is monotonically increasing in $g'$ under the considered condition, calibration at $\varepsilon_0$ upper-bounds calibration at any larger budget within the same segment. 
Extending this ordering across adjacent segment boundaries proves that the calibrated estimate upper-bounds $\operatorname{CPL}^{\star}$ for every $\varepsilon>\varepsilon_0$.
The complete proof is provided in Appendix~\ref{proof:thm:beta}.

\begin{corollary}\label{corollary:beta_upper_bound}
Let $e^{-\beta}$ be the calibrated lower bound calibrated at $\varepsilon=\varepsilon_0$. Then,
\begin{equation}
    \lim_{\varepsilon_0 \rightarrow 0}\beta = \ln\frac{1 + \tilde{A} - \tilde{B} + e^{\alpha}(-1 + \delta) - 2\delta}{1 + (-1 + \tilde{A} - \tilde{B}) e^{\alpha} - \delta}\leq \gamma.
\end{equation} Here, $\tilde{A} = \sum_{\overline{G}_i \in \overline{G},\ \overline{G'}_i\in \overline{G'},\ \frac{\overline{G}_i}{\overline{G'}_i}>1} \overline{G}_i$ and $\tilde{B} = \sum_{\overline{G}_i \in \overline{G},\ \overline{G'}_i\in \overline{G'},\ \frac{\overline{G}_i}{\overline{G'}_i}>1} \overline{G'}_i$. 
\end{corollary}

The proof of Corollary~\ref{corollary:beta_upper_bound} is available in Appendix~\ref{proof:corollary:beta_upper_bound}.

\begin{corollary}\label{corollary:alpha_beta_delta_apl}
Let $X_k$ and $\hat X$ be discrete, correlated attributes. Let $(\alpha,\beta,\delta)$ be characterization parameters for $L_{\hat{X}\rightarrow X_k}$. Then the CPL characterized by $(\alpha,\beta,\delta)$ can be computes as
\begin{equation}\label{eqn:db_apl_cal_delta_available}
    \operatorname{CPL}=
    \begin{cases} 
     \ln \bigl(1+\delta(e^\varepsilon-1)\bigr)
    \text{; if }  e^\alpha-1-\delta(e^\varepsilon-1) < 0, \\
     \ln \left(\frac{e^{\alpha} + \delta - e^{\alpha}\delta - \delta e^{\varepsilon}
      + e^{\alpha+\varepsilon}(-1 + e^{\beta} + \delta)-1} {\delta(e^{\varepsilon}-1) - e^{\varepsilon}
      + e^{\beta}(-1 + e^{\alpha} + \delta + e^{\varepsilon} - \delta e^{\varepsilon})}\right)\text{; otherwise.}
    \end{cases}
\end{equation}
\end{corollary}

The proof of Corollary~\ref{corollary:alpha_beta_delta_apl} is available in Appendix~\ref{proof:corollary:alpha_beta_delta_apl}.

\begin{corollary}\label{corollary:apl_monotonially_increases_with_beta}
    Correlation-induced privacy leakage monotonically increases with $\gamma$.
\end{corollary}

The proof of Corollary~\ref{corollary:apl_monotonially_increases_with_beta} is available in Appendix~\ref{proof:corollary:apl_monotonially_increases_with_beta}.

Next, we explore the methods to handle inaccuracies in the prior distribution in Section~\ref{section:Handling Inaccuracies in Prior Distribution}.

\subsection{Tuning \texorpdfstring{$(\alpha,\beta,\delta)$}{(alpha, beta, delta)} for Inaccuracies in Prior Distribution}\label{section:Handling Inaccuracies in Prior Distribution}

So far, we have focused on approximating the probability distribution between two correlated attributes using a reduced number of parameters (i.e., $\alpha$, $\beta$, $\delta$). In this section, we incorporate distributional uncertainty into the metric ($\textbf{OB4}$).

Estimated CPL of \acronym\ is implicitly robust to the distributional uncertainties to some extent, as it provides a tight upper bound for $\operatorname{CPL}^\star$. However, it is not controllable precisely. Therefore, next, we improve the $(\alpha,\beta, \delta)$ method to incorporate distributional uncertainties.

Let $P_A (\in \mathbb{R}^{a\times b})$ and $P_K (\in \mathbb{R}^{a\times b})$ be the actual and known conditional distribution between $X_k$ and $\hat{X}$ (i.e., $P(\hat{X}|X_k)$) respectively. 
Let $\Delta (\in \mathbb{R}^{a\times b})$ denote the entry-wise uncertainty matrix between $P_A$ and $P_K$ (i.e., $\Delta \geq |P_A-P_K|$, where the inequality is interpreted entry-wise).
Then, we can compute the updated $(\hat{\alpha},\hat{\beta},\hat{\delta})$ such that the estimated CPL of $(\hat{\alpha},\hat{\beta},\hat{\delta})$ provides an upper bound for $\operatorname{CPL}^\star$ as stated in Theorem~\ref{thm:upperbound_dist_bias}. 
Here, $\operatorname{CPL}^\star$ denotes the optimal CPL of $P_A$.

\begin{theorem}\label{thm:upperbound_dist_bias}
    Let $\Delta$ denote the entry-wise uncertainty matrix satisfying $\Delta \geq |P_A-P_K|$, where the inequality is interpreted entry-wise.
    Then, we can compute the updated $\hat{\alpha},\hat{\beta},\hat{\delta}$ such that CPL characterized by $(\hat{\alpha}, \hat{\beta},\hat{\delta})$ provides an upper bound for the optimal CPL of $P_A$ (i.e., $\operatorname{CPL}^\star$) as
\begin{align*}
    \hat{\alpha} &= \max_{\hat{x}\in \mathcal{\hat{X}},\ x,x'\in\mathcal{X}_k,\ p(\hat{x}|x')-\Delta_{\hat{x},x'}> 0} \left\{\frac{\min\{1, p(\hat{x}|x) + \Delta_{\hat{x},x}\}}{ p(\hat{x}|x') - \Delta_{\hat{x},x'}}\right\}, \\
    \hat{\beta} &= \ln\frac{1 + \hat{\phi} + e^{\hat{\alpha}}(-1 + \hat{\delta}) - 2\hat{\delta}}{1 + (-1 + \hat{\phi}) e^{\hat{\alpha}} - \hat{\delta}}, \\
    \hat{\delta} &= \max_{x,x' \in \mathcal{X}_k} \sum_{\hat{x} \in \mathcal{\hat{X}},\ p(\hat{x}|x')-\Delta_{\hat{x},x'} \leq 0} p(\hat{x}|x)+ \Delta_{\hat{x},x}.
\end{align*}
Here, $\phi = \max_{S,G,G'} \bigl(\sum_{i\in S} G_i-\sum_{i\in S}G'_i\bigl)$, subject to $S \subseteq [b]$,\ $\ p(\hat{x}_i|x)-\Delta_{\hat{x}_i,x}\leq G_i \leq p(\hat{x}_i|x)+\Delta_{\hat{x}_i,x},\ p(\hat{x}_i|x')-\Delta_{\hat{x}_i,x'}\leq G'_i \leq p(\hat{x}_i|x')+\Delta_{\hat{x}_i,x'},\ \Delta_{\hat{x}_i,x},\Delta_{\hat{x}_i,x'}\in \Delta,\ x,x' \in \mathcal{X}_k, \hat{\mathcal{X}}=\{\hat{x}_1,\dots,\hat{x}_b\},i\in [b]$.
\end{theorem}

The proof of Theorem~\ref{thm:upperbound_dist_bias} is available in Appendix~\ref{proof:thm:upperbound_dist_bias}.

\section{\name}\label{section:Dependency Budget}

In the previous section, we described the design principles of a novel metric for quantifying dependencies between attributes.
In this section, we formally define \emph{\name{}} and present the associated algorithms.

\subsection{Definition of \name}\label{section:calibrated Dependency Budget}

The proposed metric quantifies the dependencies between attributes and indicates the privacy leakage arising from their interdependencies. 
To reflect this, we name it as ``\textbf{\name}''\ (\acronym). Next, we formally defined \acronym\ in Definition~\ref{defn:calibrated_dependency_budget}.

\begin{definition}[\name\ ($\acronym$)] \label{defn:calibrated_dependency_budget} 
Let $X_k$ and $\hat{X}$ be two dependent random variables with alphabets $\mathcal{X}_k$ and $\mathcal{\hat{X}}$, respectively. Then, $(\alpha, \beta, \delta)$-\name\ from $\hat{X}$ to $X_k$  ($(\alpha, \beta, \delta)-\operatorname{\acronym_{ \hat{X}\rightarrow X_k}}$) is defined as
\begin{align}
    \alpha &= \max_{\hat{x}\in \mathcal{\hat{X}},\ x,x'\in\mathcal{X}_k,\ p(\hat{x}|x')-\Delta_{\hat{x},x'}> 0} \left\{\frac{\min\{1, p(\hat{x}|x) + \Delta_{\hat{x},x}\}}{ p(\hat{x}|x') - \Delta_{\hat{x},x'}}\right\},\label{eqn:db_alpha} \\
    \beta &= \ln\frac{1 + \phi + e^{\alpha}(-1 + \delta) - 2\delta}{1 + (-1 + \phi) e^{\alpha} - \delta}, \label{eqn:db_beta} \\
    \delta &= \max_{x,x' \in \mathcal{X}_k} \sum_{\hat{x} \in \mathcal{\hat{X}},\ p(\hat{x}|x')-\Delta_{\hat{x},x'} \leq 0} p(\hat{x}|x)+ \Delta_{\hat{x},x}.\label{eqn:db_delta} 
\end{align}
Here, $\Delta$ denotes the distribution uncertainty, and $\phi = \max_{S,G,G'} \bigl(\sum_{i\in S} G_i-\sum_{i\in S}G'_i\bigl)$, subject to $S \subseteq [b]$,\ $\ p(\hat{x}_i|x)-\Delta_{\hat{x}_i,x}\leq G_i \leq p(\hat{x}_i|x)+\Delta_{\hat{x}_i,x},\ p(\hat{x}_i|x')-\Delta_{\hat{x}_i,x'}\leq G'_i \leq p(\hat{x}_i|x')+\Delta_{\hat{x}_i,x'},\ \Delta_{\hat{x}_i,x},\Delta_{\hat{x}_i,x'}\in \Delta,\ x,x' \in \mathcal{X}_k, \hat{\mathcal{X}}=\{\hat{x}_1,\dots,\hat{x}_b\},i\in [b]$.

\end{definition}

\subsection{Computing \texorpdfstring{$\alpha,\beta,\delta$}{alpha, beta, delta} Parameters}

Algorithm~\ref{algorithm:parameter_computation} computes $(\alpha, \beta,\delta)-\operatorname{\acronym_{\hat{X}\rightarrow X_k}}$ for two dependent attributes $\Hat{X}$ and $X_k$. 
Here, we traverse all possible $\overline{G}$ and $\overline{G'}$ values over the $\mathcal{X}_k$ alphabet as in line~3. 
Next, we compute the possible lower (i.e., $L, L'$) and upper (i.e., $U, U'$) values of the actual distribution based on the known distribution $P_K$ and uncertainty $\Delta$ in lines~4-5. 
Next, we traverse through the alphabet of $\hat{X}$, $\hat{x}\in \hat{\mathcal{X}}$ using for-loop at line~8. 
Next, we assign $P_K(\hat{X}=\hat{x}|X_k=x)$ and $P_K(\hat{X}=\hat{x}|X_k=x')$ into $g$ and $g'$, respectively at lines~9-10. 
Then, we can compute the maximum value for $\alpha$ as at line~12 and the maximum value for $\delta$ as at line~16. To compute $\beta$ for selected $\overline{G}, \overline{G'}$, we solve the relaxed linear program (LP) defined in line~18, which provides a polynomial-time upper bound to the exact MILP formulation. Then we compute the maximum $\phi$ at line~20, and $\beta$ is computed at line~22. Suppose alphabet sizes of $X_k$ and $\hat{X}$ are $|\mathcal{X}_k|=a$ and $|\hat{\mathcal{X}}|=b$, respectively. 
If we use the interior point method to solve the LP, then we have $b$ variables and $O(b)$ constraints. 
Therefore, the time complexity is approximately $O(b^{3.5})$ for fixed numerical precision~\cite{Boyd_convex_optimization}. 
Then, Algorithm~\ref{algorithm:parameter_computation} has time complexity of $O(a^2b^{3.5})$ as it traverses through all possible pairs $x,x'\in \mathcal{X}_k$.

\begin{algorithm}[h]
\caption{Compute ($\alpha, \beta,\delta$) parameters of $\operatorname{\acronym_{\hat{X}\rightarrow X_k}}$.}\label{algorithm:parameter_computation}
\begin{algorithmic}[1]
    \Statex \textbf{Input:}
    \Statex \hspace{1em} $P_K(\Hat{X}|X_k)$ \Comment{\textit{Known distribution.}}
    \Statex \hspace{1em} $\Delta$ \Comment{\textit{Distributional uncertainty.}} 
    \State $\alpha \gets 0, \quad \delta \gets 0,\quad \phi \gets 0$ \Comment{\textit{Initialisation.}}
    \ForAll{$x,x' \in \mathcal{X}_k$}\Comment{\textit{All possible values pairs in $\mathcal{X}_k$.}}
        \State $\overline{G} \gets P_K(\cdot|X_k=x),\ \overline{G'} \gets P_K(\cdot|X_k=x')$
        \State $L \gets \overline{G} - \Delta_{x},\ U \gets \overline{G} + \Delta_{x}$ \Comment{\textit{Here, $\Delta_x$ is uncertainties of $P_K(\cdot|X_k=x)$.}}
        \State $L' \gets \overline{G'} - \Delta_{x'},\ U' \gets \overline{G'} + \Delta_{x'}$ \Comment{\textit{Here, $\Delta_{x'}$ is uncertainties of $P_K(\cdot|X_k=x')$.}}
        \State $M^+_i=\max(0,U_i-L'_i),\ M^-_i =\max(0,U'_i-L_i),\ \forall i \in [b]$
        \State $\overline{\delta} \gets 0$
        \ForAll{$\hat{x} \in \hat{\mathcal{X}}$}
            \State $g \gets P_K(\hat{X}=\hat{x}|X_k=x)$
            \State $g' \gets P_K(\hat{X}=\hat{x}|X_k=x')$
            \If{$g'-\Delta_{\hat{x},x'} > 0$} \Comment{\textit{Here, $\Delta_{\hat{x},x'}$ is the uncertainty of $P_K(\hat{X}=\hat{x}|X_k=x)$.}}
            \State $\alpha \gets \max(\alpha, \ln \frac{\min(1, g+\Delta_{\hat{x},x})}{g'-\Delta_{\hat{x},x'}})$ \Comment{\textit{From~\eqref{eqn:db_alpha}.}}
            \Else
                \State $\overline{\delta} \gets \overline{\delta}+g+\Delta_{\hat{x},x}$ \Comment{\textit{From~\eqref{eqn:db_delta}.}}
            \EndIf
            \State $\delta \gets \min(1,\max(\delta, \overline{\delta}))$
                
            \EndFor
            \State \textbf{Linear program:}
\[
        \begin{array}{ll}
\underset{Z,G,G'}{\text{maximize}} & \sum_{z_i \in Z, \ i\in [b]} z_i \\[0.4em]
\text{subject to} 
& 0\leq z_i \leq \frac{M^+_i}{M^+_i + M^-_i}( G_i - G'_i +M^-_i);\ \forall i\in [b],\\
& L \leq G \le U,\quad \mathbf{1}^\top G=1,\\
& L' \leq G' \le U',\quad \mathbf{1}^\top G'=1.
\end{array}\]

\State $\overline{\phi}\gets$Solve the LP with any polynomial-time method (e.g., interior-point)
        \State $\phi \gets \max(\phi,\overline{\phi})$  
    \EndFor
    \State $\beta = \ln\frac{1 + \phi + e^{\alpha}(-1 + \delta) - 2\delta}{1 + (-1 + \phi) e^{\alpha} - \delta}$ \Comment{\textit{From~\eqref{eqn:db_beta}.}}
    \State \Return $\alpha, \beta, \delta$
\end{algorithmic}
\end{algorithm}


\subsection{Estimating CPL Using \name}\label{secion:Calculate Privacy Leakage Using Dependency Budget}

As illustrated in Algorithm~\ref{algo:APL_cal_using_db}, CPL can be estimated in constant time $O(1)$ using $(\alpha,\beta,\delta)-\operatorname{\acronym}_{ \hat{X}\rightarrow X_k}$ and~\eqref{eqn:db_apl_cal_delta_available}.

\begin{algorithm}[bth]
\caption{Compute CPL (i.e., $L_{\Hat{X} \rightarrow X_k}$) using $(\alpha,\beta,\delta)-\operatorname{\acronym}_{\hat{X}\rightarrow X_k}$.}\label{algo:APL_cal_using_db}
\begin{algorithmic}[1]
    \Statex \textbf{Input:}
    \Statex \hspace{1em} $\alpha, \beta, \delta$ \Comment{\textit{Parameters of $\operatorname{\acronym_{\hat{X}\rightarrow X_k}}$.}}
    \Statex \hspace{1em} $\varepsilon$ \Comment{\textit{Privacy budget of the mechanism of $\hat{X}$.}}
    \If{$e^\alpha-1-\delta(e^\varepsilon-1) < 0$} \Comment{\textit{From~\eqref{eqn:db_apl_cal_delta_available}.}}
        \State $L_{\hat{X} \rightarrow X_k} \gets \ln \bigl(1+\delta(e^\varepsilon-1)\bigr)$
    \Else
        \State $L_{\hat{X} \rightarrow X_k} \gets \ln \left(\frac{e^{\alpha} + \delta - e^{\alpha}\delta - \delta e^{\varepsilon}
      + e^{\alpha+\varepsilon}(-1 + e^{\beta} + \delta)-1} {\delta(e^{\varepsilon}-1) - e^{\varepsilon}
      + e^{\beta}(-1 + e^{\alpha} + \delta + e^{\varepsilon} - \delta e^{\varepsilon})}\right)$  
    \EndIf
    \State \Return $L_{\Hat{X} \rightarrow X_k}$
\end{algorithmic}
\end{algorithm}

\subsection{Defining Distribution Uncertainty}\label{section:Defining_Distribution_Uncertainty}

We consider two main approaches to modeling uncertainty in the known distributions: 
(i) a naive approach and (ii) a statistically estimated approach.

\subsubsection{Naive Approach}

If we have only one source of known distribution (i.e., $P_K$), then defining the uncertainty in the known distribution can be tricky. In such scenarios, we adopt a naive method of adding a maximum error margin for $P_K$.

Let $P_K$ be the known distribution and \textbf{\textit{e}} be the percentage error. Then the uncertainty $\Delta$ is given by,
\begin{equation}
    \Delta = eP_K
\end{equation}
In our experiments, we observe that adding $10\%-20\%$ error margin (i.e., \textbf{\textit{e}}) provides robust CPL estimation in most of the time in real data. See Section~\ref{section:Evaluating Distribution Uncertainty Defining Methods} for the experimental results.

\subsubsection{Statistically Estimated Approach}\label{section:Statistically Estimated Uncertainty}

If we have multiple sources to learn the distribution, then uncertainty can be statistically computed.
We have identified that the range provides a privacy robust uncertainty while the standard deviation provides a utility-privacy balanced uncertainty level.

\textbf{Range}
Let $P_{K1}, P_{K2},\dots, P_{KN}$ be the $N$ number of known distributions. Here, $P_{Ki}\in \mathbb{R}^{a\times b}$. Then the uncertainty $\Delta \in \mathbb{R}^{a\times b}$ is given by,
\begin{equation}\label{eqn:range}
    \Delta = \max_{i \in [N]} \{P_{Ki}\}-\min_{i \in [N]} \{P_{Ki}\}.
\end{equation}
Here, $\max_{i \in [N]} \{P_{Ki}\}$ and $\min_{i \in [N]} \{P_{Ki}\}$ compute the maximum and minimum of each entry over all known distributions, respectively.

\textbf{Standard Deviation}
Let $P_{K1}, P_{K2},\dots, P_{KN}$ be the $N$ known distributions. Here, $P_{Ki}\in \mathbb{R}^{a\times b}$. Then the uncertainty $\Delta \in \mathbb{R}^{a\times b}$ is given by,
\begin{equation}\label{eqn:std}
    \Delta = \sqrt{\frac{\sum_{i=1}^N \bigl(P_{Ki}-\overline{P_K}\bigr)^2}{N-1}}.
\end{equation}
Here, $\overline{P_K} \in \mathbb{R}^{a\times b}$ is the mean of $P_{K1}, P_{K2},\dots, P_{KN}$.

\section{Evaluation}\label{section:evaluation}

In this section, we experimentally evaluate and validate the proposed \emph{\name{}} (\acronym) metric.
We first describe the experimental setup in Section~\ref{section:Experimental Setup}.
We then compare \acronym{} with competing metrics in Section~\ref{section:Benchmark}.
Section~\ref{section:Validating_DT_with_CPL} validates the CPL estimated by \acronym{} against $\operatorname{CPL}^{\star}$, while Section~\ref{section:Tolerating Distribution Uncertainties} evaluates its robustness under distributional uncertainty.
Section~\ref{section:Evaluating_dt_different_structural_regimes} evaluates \acronym{} under different structural regimes, demonstrating that it is not limited to a single dependency pattern.
Section~\ref{section:Privacy Budget Calibration} shows how \acronym{} can be used to calibrate privacy budgets in data collection settings that incorporate prior knowledge.
Section~\ref{section:inference_attacks} conducts attribute inference attacks on data collected under privacy budgets calibrated using \acronym{}.
Finally, Section~\ref{section:random_sampling} compares \acronym{} with random sampling.

\subsection{Experimental Setup}\label{section:Experimental Setup}

\subsubsection{Environment}
We implement our experiments using Python 3.10 and run them on an AMD Ryzen Threadripper PRO 3955WX processor, 128GB of RAM, and 4TB of storage.

\subsubsection{Datasets}\label{section:datasets}
We use five real datasets. In these datasets, we have converted continuous numerical attributes into discrete attributes by binning, as this work focuses on privacy leakages related to discrete/categorical data.

\textbf{SPM}~\cite{SPM} dataset is a community survey conducted in the USA on the Supplemental Poverty Measure (SPM), which contains 38 attributes with more than 3M samples.

\textbf{CelebA}~\cite{celebA_attributes} dataset contains 40 binary attributes of human faces with more than 200K samples. 

\textbf{Adult}~\cite{adult_dataset} is an income survey dataset containing 14 features and over 40K samples.

\textbf{CVD}~\cite{sulianova2021cardiovascular} dataset consists of 11 features of 70K cardiovascular patient records.

\textbf{DSS}~\cite{dss} dataset is released by the Department of Social Services (DSS) in Australia, and the selected dataset contains 27 features, with more than 38K samples.

\textbf{Synthetic Datasets.} We generate synthetic datasets with two nonlinearly dependent attributes, $X_k$ and $\hat{X}$, using $x_k = \sin{(2\pi \hat{x})}+n$.
Here, $\hat{x} \in \mathcal{\hat{X}},\quad x_k \in \mathcal{X}_k$, $n\sim\mathbf{N}(0,2)$, where $\mathcal{\hat{X}}$ and $\mathcal{X}_k$ are the alphabets of $\hat{X}$ and $X_k$, and $\mathbf{N}$ is normal distribution.
We create multiple datasets with varying alphabet sizes from 2 to 1000.
For example, Figure~\ref{fig:synthetic_distribution} depicts $P(\hat{X}|X_k)$ for alphabet size $20$.

\begin{figure}[t] 
  \centering
  \includegraphics[trim={0cm 0.15cm .2cm 0.cm},clip,width=0.5\linewidth]{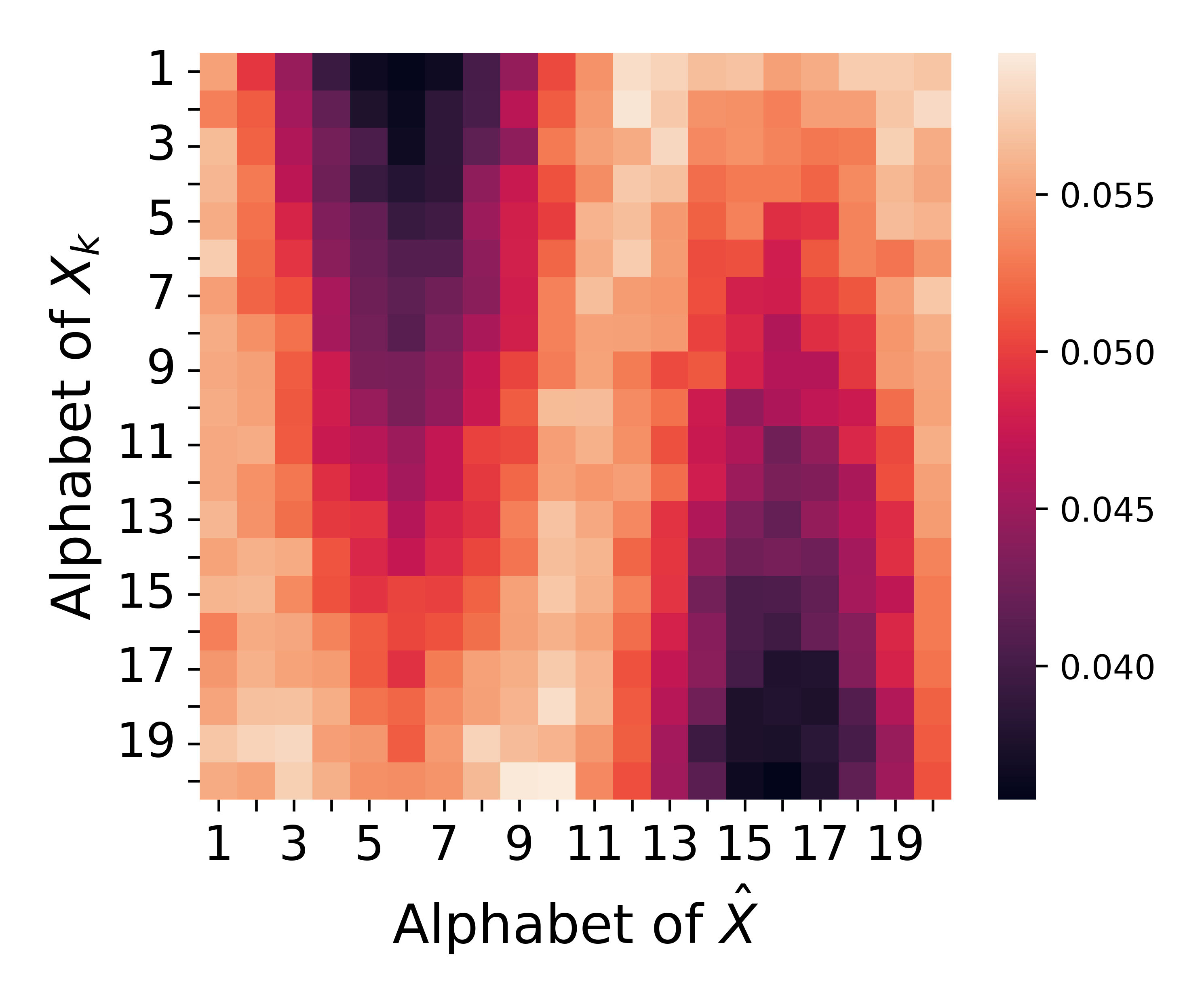}
  \caption{Heatmap of the conditional distribution $P(\hat{X}\mid X_k)$ for the synthetic dataset, where both $X_k$ and $\hat{X}$ have an alphabet size of 20.}\label{fig:synthetic_distribution}
\end{figure}

\begin{table*}[h]
\centering
\caption{Comparison between different CPL computation algorithms (CPL caused by one attribute about another) with our proposed Dependency Triad (\acronym). Here, $a$ and $b$ are alphabet sizes of the two attributes.}
\begin{tabular}{lcccccc}
\hline
Algorithm/ & \multicolumn{2}{c}{Time Complexity} & Space Complexity & Distribution Shift & Supporting & Nature of the \\ 
\cline{2-3}
Metric & Precomputation & Privacy Calculation &  & Incorporability& Mechanisms&Estimation \\\hline
ITS~\cite{data_correlation_location_similar_LDP}  & -  & $O(ab)$ & $O(ab)$ & Not Support & Two-level LDP & $\operatorname{CPL}^\star$ \\
CBP~\cite{secon_Collecting_High-Dimensional_Correlation_LDP}  & - & $O(2^a)$ & $O(ab)$ & Not Support & GRR & N/A\\
HCC-1~\cite{paper_1} & - & $O(a^2b^2)$ & $O(ab)$ & Not support & Simple LDP & Actual\\
HCC-2~\cite{paper_1} & - & $O(a^2b\log b)$ & $O(ab)$ & Not support & Any LDP & $\operatorname{CPL}^\star$\\
LTM~\cite{Quantifying_DP} & $O(a^2b+ab\log b)$ & $O(\log ab)$ & $O(ab)$ & Not Support & Any LDP & $\operatorname{CPL}^\star$\\
MI~\cite{covers_mutual_information}/PCC~\cite{PCC_new} & $O(ab)$ & $O(1)$ & $O(1)$ & Not Support & N/A & Rough Estimation\\
\textbf{\acronym\ (Ours)}  & ${O(a^2b^{3.5})}$  & $\mathbf{O(1)}$ & $\mathbf{O(1)}$ & \textbf{Support} & \textbf{Any LDP} & Upperbound \\
\hline
\end{tabular}
\label{table:comparison_CPL_algos}
\end{table*}

\subsubsection{Competitors}
We have considered five state-of-the-art CPL computation algorithms to compare and validate \acronym: ITS~\cite{data_correlation_location_similar_LDP}, CBP~\cite{secon_Collecting_High-Dimensional_Correlation_LDP}, HCC-1~\cite{paper_1}, HCC-2~\cite{paper_1} and LTM~\cite{Quantifying_DP}.
We use the CPL of HCC-2 and LTM as the reference to evaluate the estimated CPL of \acronym, because they provide the \emph{optimal CPL (i.e., $\operatorname{CPL}^\star$) for generic LDP mechanisms} when the ground-truth distributions are known (See Table~\ref{table:comparison_CPL_algos} for comparison details).
Additionally, we use random sampling (RS)~\cite{ldp_Frequency_Estimation} as a baseline to evaluate DT in high-dimensional statistical learning (i.e., k-marginal).

\subsubsection{Evaluation Metrics}\label{section:Evaluation Metrics}
We mainly use four metrics to evaluate \acronym: (i) TCPL, (ii) CPL estimation error, (iii) NMSE-CPL, and (iv) total variation distance.

\textbf{TCPL} quantifies the total CPL of all the attributes in the dataset as 
\begin{equation}\label{eqn:TCPL}
    \operatorname{TCPL} = \sum_{X_k \in X}\ \sum_{x \in X\setminus \{X_k\}} L_{x \rightarrow X_k}.
\end{equation}

\textbf{CPL Estimation Error} quantifies the algebraic error between estimated CPL ($\tilde{L}_{x \rightarrow X_k}$) and $\operatorname{CPL}^\star$ ($L^\star_{x \rightarrow X_k}$) by, 
\begin{equation}\label{eqn:cpl_estimation_error}
    \text{CPL Estimation Error} =  \tilde{L}_{x \rightarrow X_k} - L^\star_{x \rightarrow X_k}.
\end{equation}
Here, we can observe the sign of the estimation error, which is essential to identify whether the estimation is an upper bound or a lower bound for the $\operatorname{CPL}^\star$.

\textbf{NMSE-CPL} quantifies the normalized mean square error between $\operatorname{CPL}^\star$ and the estimated CPL for all the attributes in the dataset as
\begin{equation}\label{eqn:NMSE}
    \operatorname{NMSE} = \frac{\sum_{X_k \in X} \sum_{x \in X\setminus \{X_k\}} \bigl(\tilde{L}_{x \rightarrow X_k} - L^\star_{x \rightarrow X_k}\bigr)^2}{\sum_{X_k \in X} \sum_{x \in X\setminus \{X_k\}}( L^\star_{x \rightarrow X_k})^2}.
\end{equation}


\textbf{Total Variation Distance (TVD)} quantifies the discrepancy between the estimated distribution and the true distribution by,
\begin{equation}\label{eqn:tvd}
    \operatorname{TVD}(P,\hat{P}) = \frac{1}{2}\sum_{x \in \mathcal{X}} |P(x)-\hat{P}(x)|.
\end{equation}
Here, $P(x)$ denotes the true probability of value $x$, and $\hat{P}(x)$ denotes the estimated probability after reconstruction from perturbed data. A smaller TVD indicates that the estimated distribution is closer to the true distribution.

\subsubsection{Uncertainty Metrics}\label{section:Uncertainty Metrics}
We have considered four commonly used variability measurement metrics to measure the uncertainty of the known distribution. 

\textbf{Range} is a measurement of the difference between the maximum and minimum value and can be computed using~\eqref{eqn:range}.

\textbf{Standard Deviation} is a measure of how dispersed the data is in relation to the mean as defined in~\eqref{eqn:std}.

\textbf{Variance} is the square of standard deviation.

\textbf{Mean Absolute Deviation (MAD)} is a measure of variability that indicates the average distance between observations and their mean $\bigl(\overline{P_K}\bigr)$, which is given by,
\begin{equation}
    \operatorname{MAD} = \frac{\sum_{i=1}^N|P_{Ki}-\overline{P_K}|}{N}.
\end{equation}

\subsubsection{Attribute-Inference Attack}\label{sec:attack_models}

We perform attribute-inference attacks, where the adversary attempts to infer the \emph{actual} value of a target attribute from a perturbed record. 
In our threat model, we assume a \emph{worst-case} adversary that knows the perturbation mechanism and the underlying data distribution, but does not observe the exact ground-truth values of individual records.

Let $X=(X_1,\dots,X_n)$ be actual user record, and let $Y=(Y_1,\dots,Y_n)$ be the perturbed output.
For each target attribute $k\in[n]$, we train an attacker
$g_k:Y_{-k}\mapsto X_k,$
where $Y_{-k}$ denotes all released attributes except $Y_k$. Here, inference is primarily driven by inter-attribute dependencies.
During training, we split the datasets into a \emph{disjoint} training set and an evaluation set.
We apply the mechanism to the training records to obtain perturbed reports $Y$, and train $g_k$ on pairs $(Y_{-k}, X_k)$. 
We instantiate $g_k$ using \emph{a neural TabTransformer-based attacker}. 
See Appendix~\ref{appendix:section:attacks} for more details.

\textbf{Metric - }
We measure attribute-inference success (AIA) using \emph{Macro-F1}~\cite{opitz2021macrof1macrof1}.
For each attribute $X_k,\ k\in[n]$, we compute the Macro-F1 score between the attacker's prediction
$g_k(Y_{-k})$ and the ground truth $X_k$, and then average across attributes.
\begin{equation*}\label{eqn:inference_accuracy}
\text{Inference Accuracy} = \text{AIA}=\frac{1}{n}\sum_{k \in [n]}\mathrm{MacroF1}\big(g_k(Y_{-k}),\, X_k\big). 
\end{equation*}
\textbf{Lower} `Inference Accuracy' indicates \textbf{stronger resistance} to attribute inference.

\subsubsection{Privacy Budget Calibration Algorithms}
Privacy budget calibration is the process of adjusting the privacy budgets of mechanisms to trade off privacy and utility based on the requirements. 
We have proposed an algorithm (see Appendix~\ref{appendix:Unequal Privacy Budgets Across Attributes}) to calibrate the privacy budgets of each attribute when there is a targeted overall privacy for all attributes. 

\begin{figure*}[t]
  \centering
  \includegraphics[trim={0.8cm 0.cm 0.7cm 0.2cm},width=0.8\linewidth]{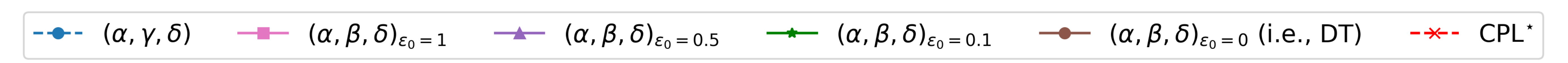}\par
  \makebox[0.9\linewidth][c]{%
  \hfill
    \begin{subfigure}[t]{0.25\linewidth}
      \centering
      \includegraphics[trim={1.45cm 0.8cm 1.2cm 0.8cm},clip,width=0.98\textwidth]{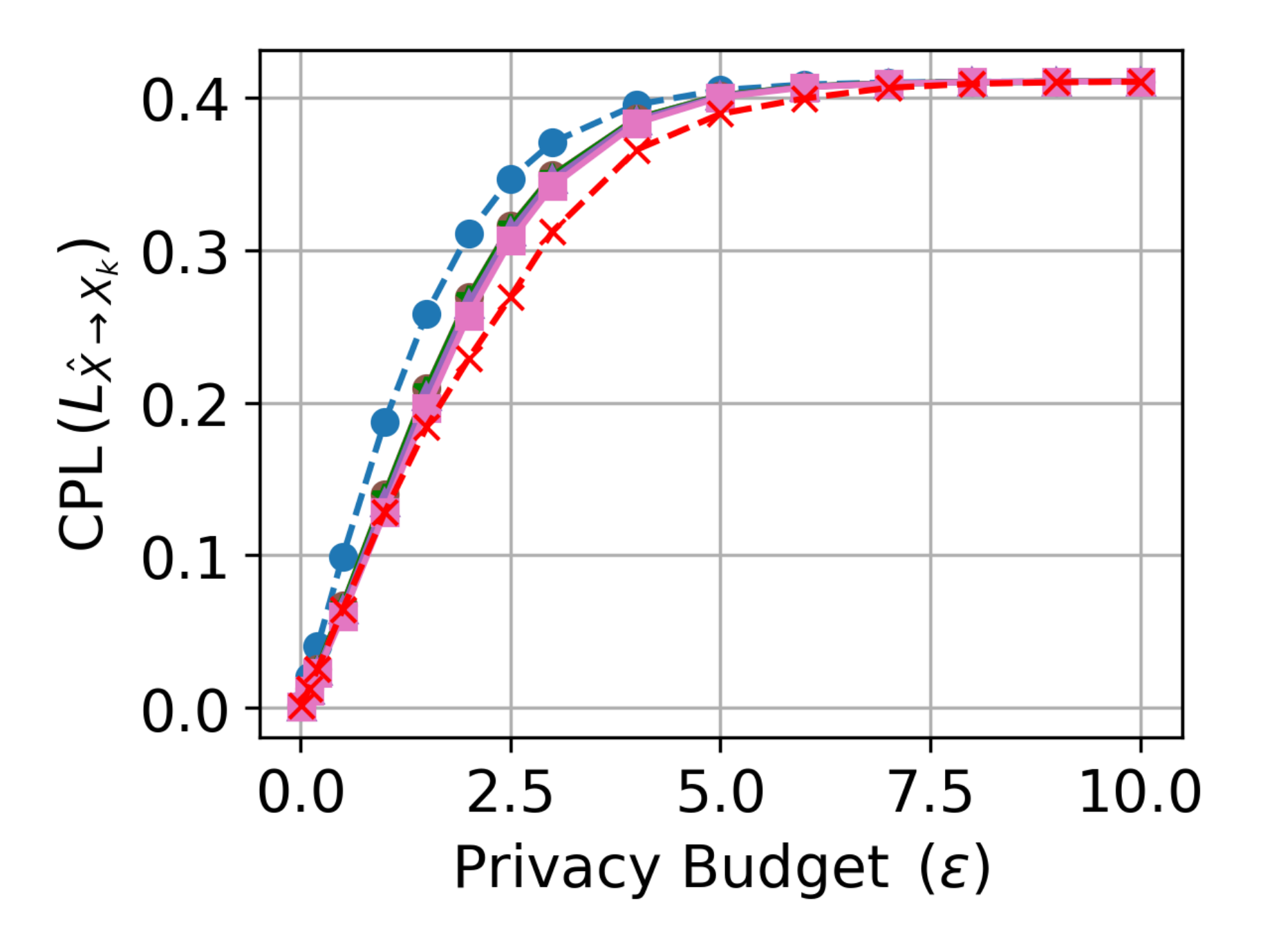}
      \caption{}\label{fig:benchmark_a}
    \end{subfigure}
    \hfill
    \begin{subfigure}[t]{0.25\linewidth}
      \centering
      \includegraphics[trim={1.45cm 0.8cm 1.2cm 0.8cm},clip,width=1\textwidth]{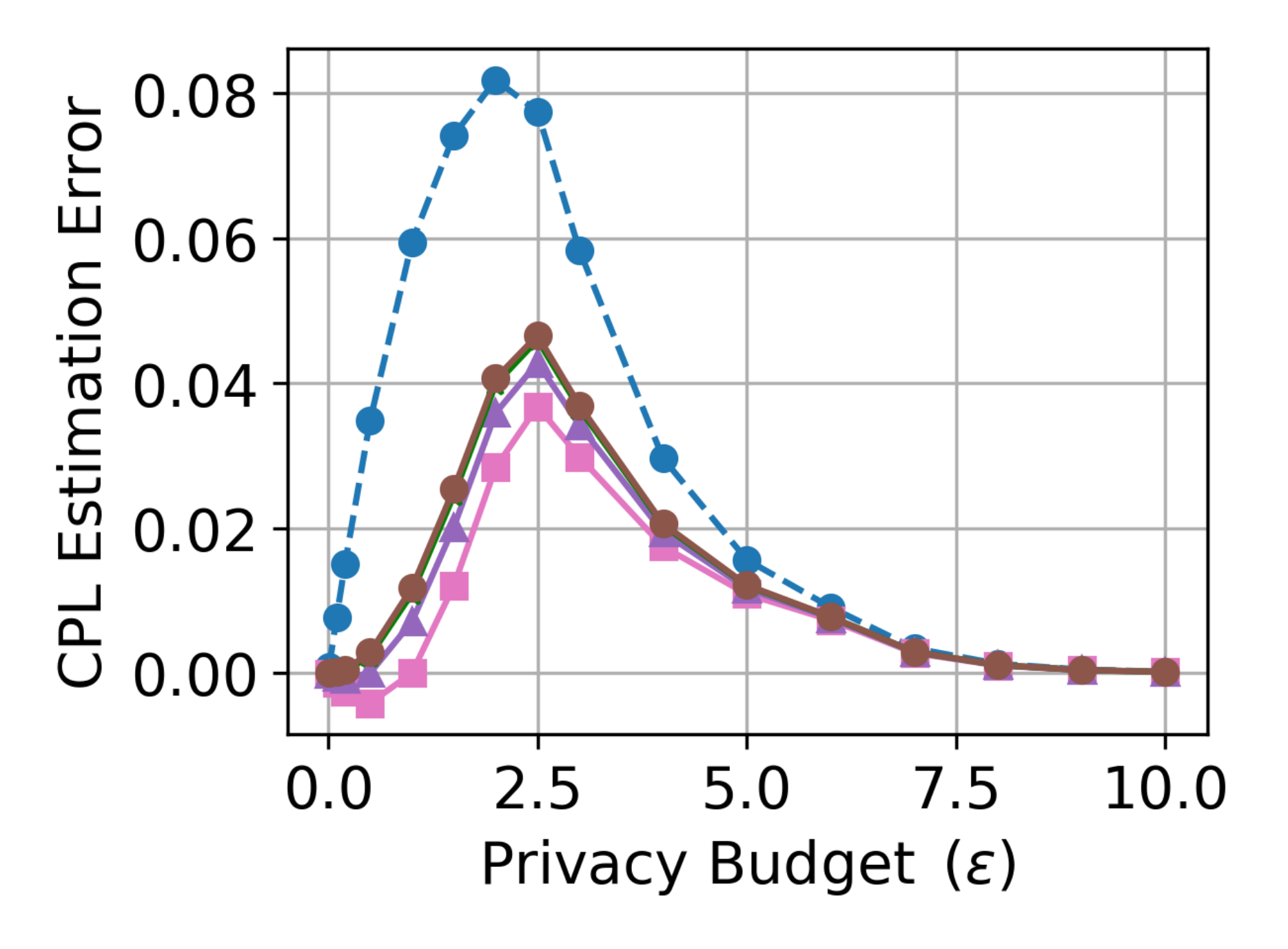}
      \caption{}\label{fig:benchmark_b}
    \end{subfigure}
    \hfill
    \begin{subfigure}[t]{0.25\linewidth}
      \centering
      \includegraphics[trim={1.45cm 0.8cm 1.2cm 0.8cm},clip,width=0.97\textwidth]{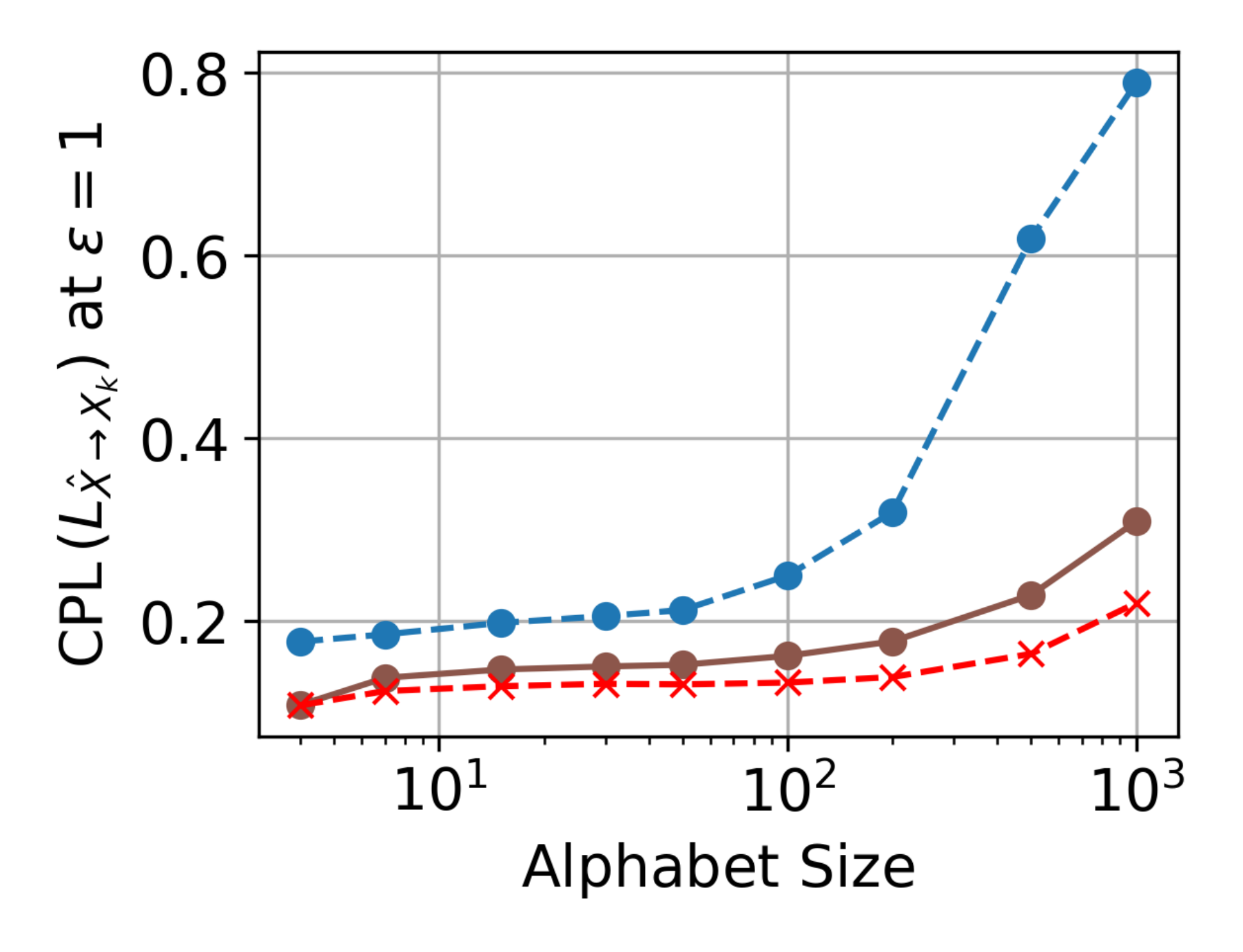}
      \caption{}\label{fig:benchmark_c}
    \end{subfigure}%
    \hfill
  }
  \caption{CPL analysis on two attribute synthetic datasets. Here, $\varepsilon_0$ refers to the calibrated parameters at $\varepsilon_0$. Further, $\varepsilon_0 = 0$ refers to the calibrated parameters limit $\varepsilon_0 \rightarrow 0$. (a) The estimated CPL of different methods varies with privacy budget $\varepsilon$. Here, the alphabet size of the dataset is 10. (b) CPL estimation error varies with $\varepsilon$. Here, the alphabet size of the dataset is 10. (c) The estimated CPLs of different methods vary with the alphabet size of the two attributes.}
  \label{fig:benchmark}
\end{figure*}

\subsection{Benchmark}\label{section:Benchmark}

We use state-of-the-art six benchmarks: IST~\cite{data_correlation_location_similar_LDP}, CBP~\cite{secon_Collecting_High-Dimensional_Correlation_LDP}, HCC-1~\cite{paper_1}, HCC-2~\cite{paper_1}, LTM~\cite{Quantifying_DP} and MI~\cite{covers_mutual_information}/PCC~\cite{PCC_new} to compare \acronym\ as summarized in Table~\ref{table:comparison_CPL_algos}. 
Unlike existing approaches, \acronym\ is the first metric to show that three parameters can characterize the dependency to quantify the potential CPL under LDP. 
The polynomial-time precomputation of \acronym\ is a one-time offline calibration step. After calibration, CPL can be evaluated in constant time for any privacy budget. 
By contrast, baselines such as HCC-1, HCC-2, and LTM require repeated full computation for each candidate privacy budget.
This distinction matters in practice, since privacy-budget calibration often requires evaluating many candidate privacy budgets to explore the privacy–utility trade-off, especially for large-cardinality attributes.
Moreover, \acronym\ has addressed the long-lasting distribution uncertainty challenge in CPL calculation as \acronym\ is adjustable to distribution uncertainties and computationally efficient. 
Appendix~\ref{appendix:Experimental Results for Time and Space Analysis} empirically validates the time complexity and space complexity of \acronym.

\subsection{Validating \acronym\ with \texorpdfstring{$\operatorname{CPL}^\star$}{CPL*}}\label{section:Validating_DT_with_CPL}

In this section, we benchmark the CPL estimations of \acronym\ against $\operatorname{CPL}^\star$. First, we evaluate three aspects of the $(\alpha,\beta,\delta)$ estimation method using synthetic distributions: (i) estimation accuracy (Figure~\ref{fig:benchmark_a}), (ii) upper bound guarantee (Figure~\ref{fig:benchmark_b}), and (iii) robustness to large alphabet sizes (Figure~\ref{fig:benchmark_c}). Four calibration levels are used in these experiments ($\varepsilon_0 = 1, 0.5, 0.1$, and $0$).
Note that \acronym\ corresponds to the $(\alpha,\beta,\delta)$ method calibrated at $\varepsilon_0=0$, i.e., $(\alpha,\beta,\delta)_{\varepsilon_0=0}$. Here, $\varepsilon_0 = 0$ refers to the calibrated parameters limit $\varepsilon_0 \rightarrow 0$ (See Section~\ref{section:Calibrated Characterisation Using} for more details).

\acronym\ (i.e., $(\alpha,\beta,\delta)_{\varepsilon_0=0}$) remains consistently close to $\operatorname{CPL}^\star$, even for alphabets containing $1,000$ symbols as shown in Figure~\ref{fig:benchmark}(c). In contrast, the estimation error of the uncalibrated $(\alpha,\gamma,\delta)$ method increases considerably with alphabet size.
This highlights the importance of calibration (Section~\ref{section:Calibrated Characterisation Using}).

Finally, we evaluate \acronym\ on five real-world datasets (Table~\ref{table:realworld_summary} and Figure~\ref{fig:nmse_real_datasets}). \acronym\ achieves very low NMSE-CPL error ($<0.025$) across all datasets for privacy budget $\varepsilon \in [0.01, 10]$, confirming the tightness of its estimates relative to $\operatorname{CPL}^\star$ on real data. NMSE-CPL error is further low ($< 0.003$) in high privacy regime $\varepsilon \in (0,1]$, which is commonly used in practice~\cite{Local_Differential_Privacy_in_Practice}.

\begin{table}[h]
\centering
\caption{The summary of optimal TCPL ($\operatorname{TCPL}^\star$) and estimated TCPL by \acronym\ for real datasets.}
\begin{tabular}{ccccc}
\hline
Dataset & $\varepsilon$ & $\operatorname{TCPL^\star}$ & \multicolumn{2}{c}{$\operatorname{\acronym}$}  \\ 
\cline{4-5}
  &  &  & Estimated TCPL & NMSE-CPL \\\hline
   & 0.5 & 291.45 & 294.33 & 0.00041  \\ 
\cline{2-5} 
SPM & 1 & 601.56 & 613.68 & 0.0017 \\ 
\cline{2-5} 
\hline
 & 0.5 & 101.04 & 101.04 & $1.03 e{-26}$ \\ 
\cline{2-5} 
CelebA & 1 & 221.53 & 221.53 & $8.73e{-27}$ \\ 
\cline{2-5} 
\hline
 & 0.5 &  40.33 & 41.00 & 0.00074  \\ 
\cline{2-5} 
Adult & 1 & 84.15 & 86.81 & 0.0025  \\ 
\cline{2-5} 
\hline
 & 0.5 & 17.88 & 18.13 & 0.00051   \\ 
\cline{2-5} 
CVD & 1 & 37.53 & 38.73 & 0.0025 \\ 
\cline{2-5} 
\hline
 & 0.5 & 215.44 & 218.99 & 0.00063  \\ 
\cline{2-5} 
DSS & 1 & 437.39 & 450.60 & 0.0022 \\ 
\cline{2-5} 
\hline
\end{tabular}
\label{table:realworld_summary}
\end{table}

\begin{figure}[htbp] 
  \centering
  \includegraphics[trim={0.1cm 0.15cm .2cm 0.cm},clip,width=.62\linewidth]{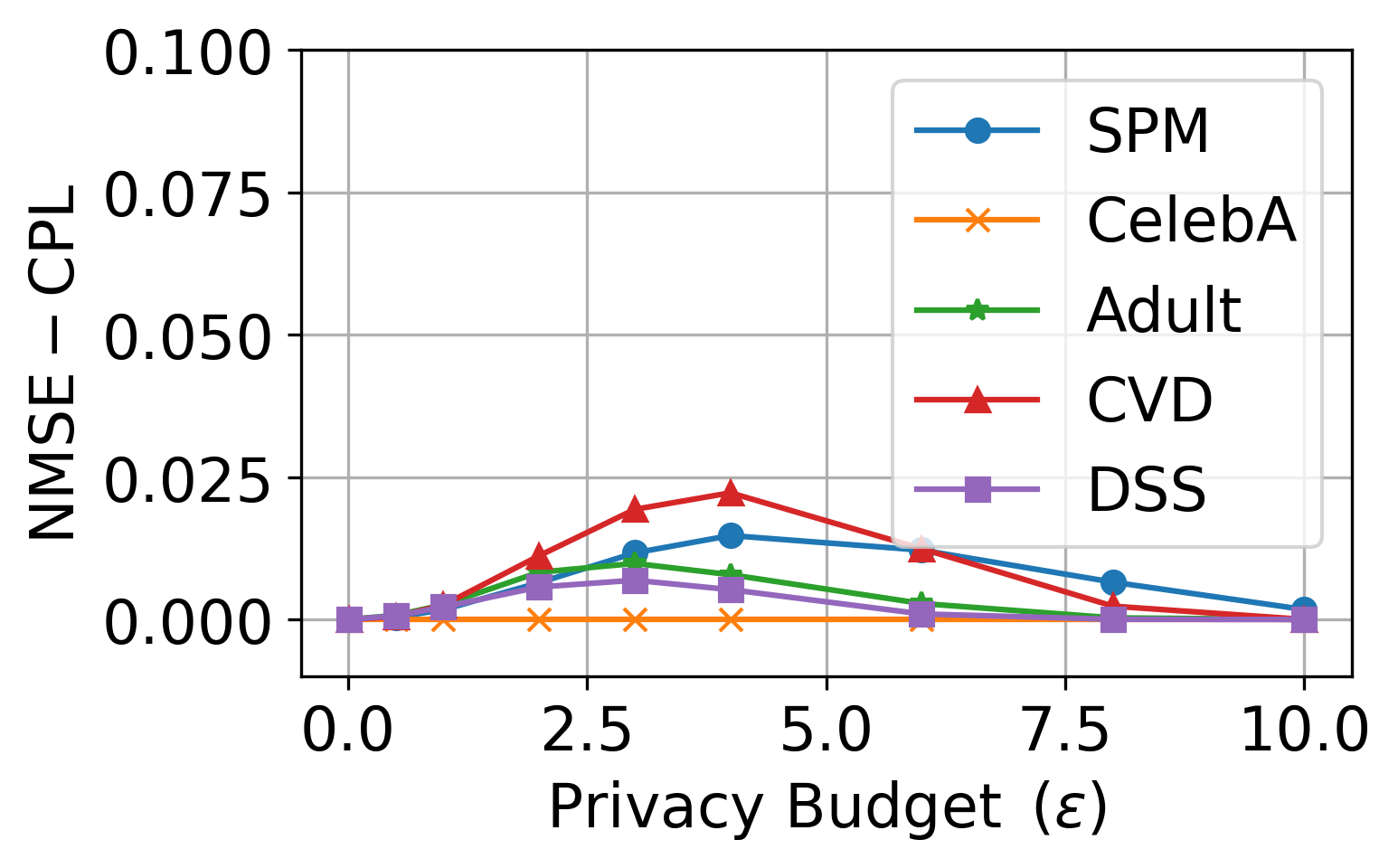}
  \caption{NMSE-CPL between $\operatorname{CPL}^\star$ and \acronym\ for the real datasets.}\label{fig:nmse_real_datasets}
\end{figure}

\subsection{Interpreting CPL Through DT Parameters}\label{section:Interpreting CPL Through DT Parameters}

Although DT characterizes CPL through three parameters $(\alpha,\beta,\delta)$, each parameter captures a distinct aspect of the dependency structure. 
In this section, we empirically examine how these parameters relate to the actual CPL observed across real datasets. Specifically, we compute CPL and the corresponding DT parameters for every pair of attributes in the five real datasets.

Figure~\ref{fig:beta_vs_cpl} shows that $\beta$ provides a useful single-parameter interpretation of CPL in high-privacy regimes, i.e., when the privacy budget $\varepsilon$ is small. In particular, the results show an approximately linear relationship between $\ln(\beta)$ and $\ln(\mathrm{CPL})$.
For $\varepsilon=1$, the observed CPL is also upper-bounded by the privacy budget, with $\max(\mathrm{CPL})=1$, as expected.

We further quantify this relationship using the Pearson correlation coefficient (PCC). Figure~\ref{fig:pcc_alpha_beta_vs_cpl} reports the PCC between $\ln(\alpha)$, and $\ln(\mathrm{CPL})$, and $\ln(\beta)$ and $\ln(\mathrm{CPL})$ across different privacy budgets. 
The results show that
$PCC(\ln(\beta), \ln(\mathrm{CPL}))$ is high in the strong-privacy regime, decreasing from $0.97$ to $0.69$ as $\varepsilon$ increases. 
This suggests that $\beta$ is the dominant explanatory parameter when the privacy budget is small. In contrast, $PCC(\ln(\alpha), \ln(\mathrm{CPL}))$ increases from $0.67$ to $0.97$ as $\varepsilon$ grows, indicating that $\alpha$ becomes more informative in lower-privacy regimes where weaker perturbation allows stronger dependency effects to be expressed.
Since $\delta$ implicitly captures edge cases where the distribution is sparse, we have not specifically evaluated.

Overall, these results provide an interpretable view of the DT parameters: $\beta$ characterizes CPL under strong privacy protection, while $\alpha$ becomes increasingly important as the privacy budget increases.

\begin{figure}[h]
  \centering
\hfill
  \begin{subfigure}[h]{0.5\columnwidth}
    \centering
    \includegraphics[width=1\textwidth]{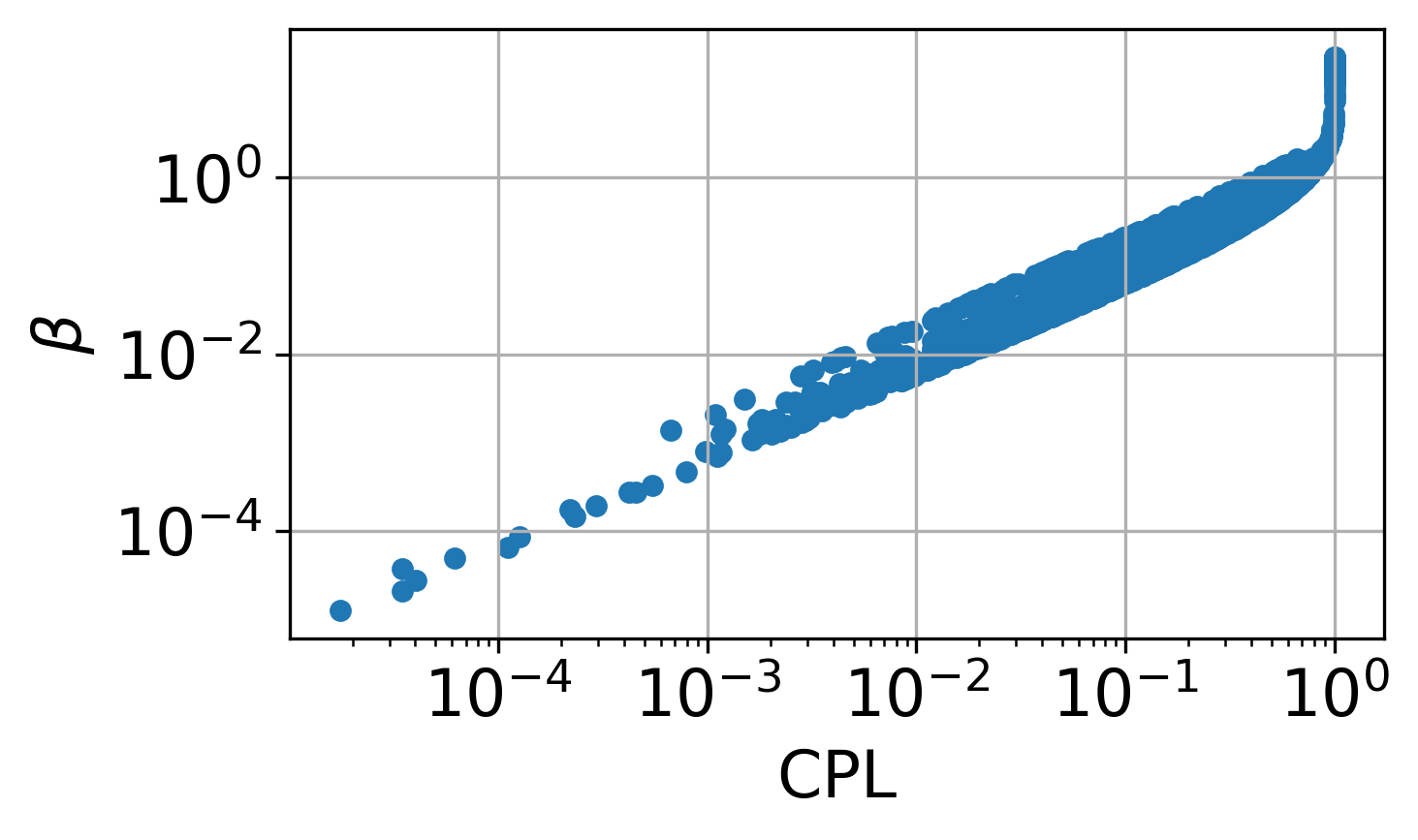}
    \caption{}\label{fig:beta_vs_cpl}
  \end{subfigure}\hfill
  \begin{subfigure}[h]{0.5\columnwidth}
    \centering
    \includegraphics[width=1\textwidth]{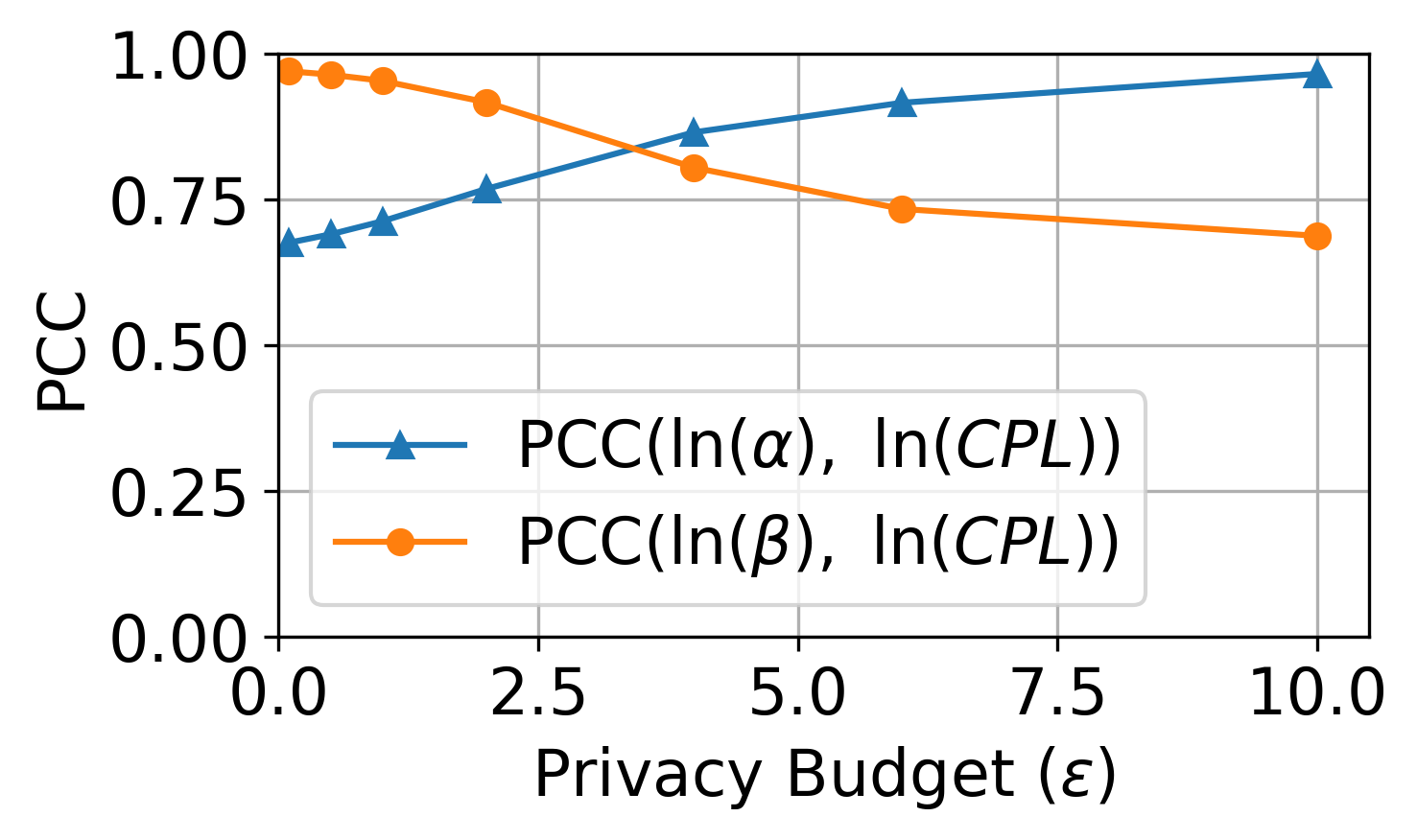}
    \caption{}\label{fig:pcc_alpha_beta_vs_cpl}
  \end{subfigure}
  \hfill
  \caption{Interpreting CPL Through DT Parameters using all five real datasets. (a) $\beta$ vs CPL at $\varepsilon = 1$. (b) Pearson correlation coefficient (PCC) between $\ln (\alpha)  (\text{or } \ln (\beta))$ and $\ln (\operatorname{CPL})$, over privacy budgets.}
  \label{fig:synthetic_alpha}
\end{figure}

\subsection{Tolerating Distribution Uncertainties}\label{section:Tolerating Distribution Uncertainties}

\subsubsection{Evaluating Distribution Uncertainty Defining Methods}\label{section:Evaluating Distribution Uncertainty Defining Methods}

Next, we evaluate the impact of estimation accuracy and privacy guarantee in different uncertainty-defining methods: \emph{statistically estimated uncertainty} and \emph{naively defined uncertainty}. Due to the limited space, we include the experimental results in Appendix~\ref{appendix:Experimental Results for Distribution Uncertainty Defining Methods}.

\textbf{Statistically Estimated Uncertainty -} Experimental results show that \textit {range} and \textit{standard deviation} (std) based uncertainty computations offer greater privacy tolerance, with std-based uncertainty being particularly robust across different prior knowledge settings.

\textbf{Naively Defined Uncertainty -} The experiments demonstrate that defining uncertainty within the range $10\% \leq \mathbf{e} \leq 20\%$ provides the best balance between preserving utility and ensuring privacy.

\begin{figure}[h]
  \centering
\hfill
  \begin{subfigure}[h]{0.5\columnwidth}
    \centering
    \includegraphics[width=1.04\textwidth]{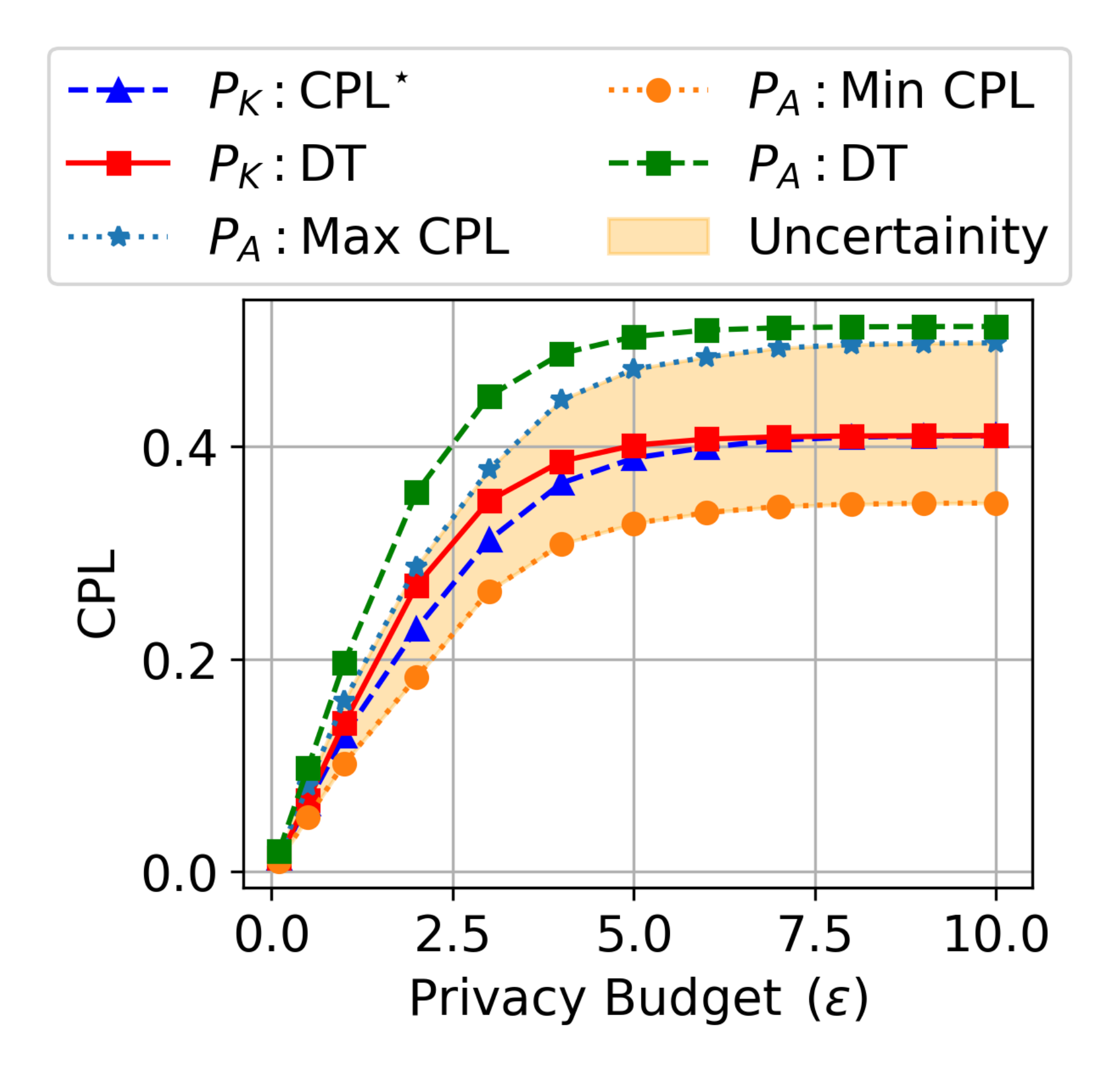}
    \caption{}\label{fig:dist_bias_a}
  \end{subfigure}\hfill
  \begin{subfigure}[h]{0.5\columnwidth}
    \centering
    \includegraphics[width=1\textwidth]{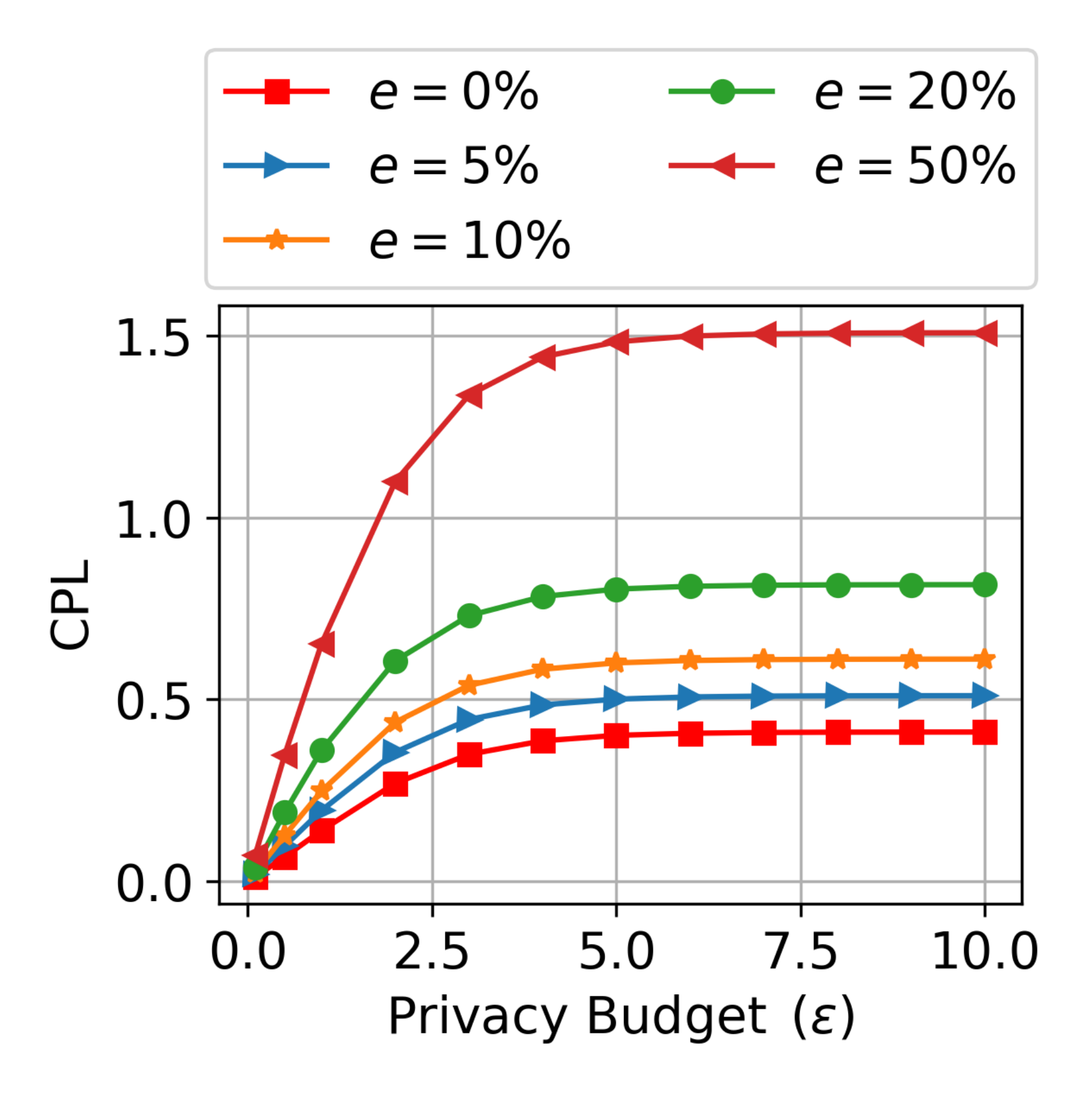}
    \caption{}\label{fig:dist_bias_b}
  \end{subfigure}
  \hfill
  \caption{(a) Evaluation of \acronym\ under distributional uncertainty using synthetic dataset. (b) Estimated CPL vs uncertainty level.}\label{fig:synthetic_distribution_uncertianty}
\end{figure}

\subsubsection{Validating \acronym\ Over Distribution Uncertainty}

First, we evaluate the CPL variation on the synthetic distribution $D_1$ with alphabet size 10, as shown in Figure~\ref{fig:synthetic_distribution_uncertianty}(a). Here, $P_K: \operatorname{CPL}^\star$ denotes the optimal CPL for the known distribution $P_K$.  As expected, the CPL estimated by \acronym\ is a tight upper bound for $P_K: \operatorname{CPL}^\star$.
Next, suppose the actual distribution $P_A$ deviates from $P_K$ by up to $\pm 10\%$ (i.e., $|P_A - P_K| \leq 10\%$). Under this uncertainty, the actual $P_A:\operatorname{CPL}^\star$ can range between the `$P_A:$ Max CPL' and `$P_A:$ Min CPL' lines. The uncertainty-aware \acronym\ estimate is depicted by `$P_A:\acronym$', which provides a tight upper bound on the maximum CPL `$P_A:$ Max CPL'. This demonstrates \acronym's ability to effectively incorporate uncertainty.

Next, we analyse how the estimated CPL of \acronym\ changes with varying error margins, as shown in Figure~\ref{fig:synthetic_distribution_uncertianty}(b). The estimated CPL grows with the error parameter $\mathbf{e}$ to accommodate a safety margin for greater uncertainty in $P_K$. The case $\mathbf{e}=0\%$ corresponds to the ideal scenario where the known distribution perfectly matches the actual.

\subsection{Evaluating \acronym\ in Different Structural Regimes}\label{section:Evaluating_dt_different_structural_regimes}

We further evaluate \acronym{} under six representative structural regimes of pairwise dependency: (i) low-entropy conditionals, (ii) sparse concentrated distributions, (iii) many-to-one mappings, (iv) one-to-many mappings, (v) near-monotone likelihood-ratio structure, and (vi) limited-tail-mass distributions. 
These regimes capture different dependency shapes, including concentrated, asymmetric, ordered, and tail-sensitive conditional distributions, which are commonly studied in information theory, sparse distribution analysis, stochastic ordering, and tail-mass estimation~\cite{covers_mutual_information,karlin1956theory}.

For each regime, we generate $100$ synthetic pairwise distributions with an alphabet size of $100$ and compare the \acronym{}-based CPL estimate with the exact value $\operatorname{CPL}^{\star}$ across different privacy budgets. 
We report the \emph{CPL estimation error} where smaller values indicate tighter estimation. Since \acronym\ is designed as a conservative estimator, a small positive gap is desirable.

\begin{table}[t]
\centering
\caption{Summary of \emph{CPL estimation errors} of \acronym\ across different structural regimes. The table reports the mean error and standard deviation; lower values indicate better estimation accuracy.}
\begin{tabular}{lcccc}
\hline
Structural & \multicolumn{4}{c}{Privacy Budget $(\varepsilon)$}  \\
\cline{2-5}
 Regime & 0.1 & 1 & 5 & 10 \\
  \hline
  Low entropy  & $9.4e{-9}\pm $ & $9.2e{-8}\pm $ & $9.3e{-8}\pm $ & $5.3e{-9}\pm $  \\ 
conditionals & $4.3e{-9}$ & $4.4e{-8}$ & $4.3e{-8}$ & $8.4e{-10}$  \\ 
\hline
 Sparse  & $3.4e{-5} \pm $ & $2.8e{-2} \pm $ & $2.4e{-1} \pm $ & $4.2e{-2} \pm $ \\ 
concentrated & $3.0e{-5} $ & $2.1e{-2} $ & $1.2e{-1}  $ & $5.0e{-2} $ \\ 
\hline
Many to one & $1.8e{-8} \pm$ &  $1.7e{-7}\pm $ & $1.6e{-7} \pm$ & $6.0e{-9} \pm $  \\  
 & $7.8e{-9}$ &  $9.1e{-8} $ & $7.6e{-8} $ & $3.0e{-9} $  \\  
\hline
One to many & $4.4e{-5} \pm$ & $2.0e{-2}\pm $ & $1.1e{-1} \pm$ & $7.0e{-3} \pm$   \\ 
 & $2.5e{-5} $ & $1.3e{-2} $ & $7.4e{-2} $ & $8.8e{-3} $   \\ 
\hline
Near monotone & $7.9e{-4} \pm$ &  $2.5e{-1} \pm$ & $7.9e{-1} \pm$ & $1.5e{-1}\pm $  \\  
likelihood ratio & $2.5e{-1} $ &  $8.6e{-2} $ &
$2.6e{-1} $& $1.8e{-1} $  \\  
\hline
Limited & $1.1e{-4} \pm $ & $3.8e{-2} \pm $ & $1.9e{-1}\pm$ & $5.9e{-2} \pm $  \\ 
tail mass & $6.9e{-5}$ & $2.2e{-2} $ & $5.3e{-2}$ & $5.0e{-2} $\\
\hline
\end{tabular}
\label{table:summary_structural_regimes}
\end{table}

Table~\ref{table:summary_structural_regimes} shows that \acronym\ provides comparable CPL estimates across diverse dependency structures. 
The results can be grouped into three cases. 
First, for concentrated or near-deterministic dependencies, such as low-entropy conditionals and many-to-one mappings, \acronym\ is extremely tight across all privacy budgets because the dominant conditional mass largely determines the worst-case leakage. 
Second, for sparse or asymmetric dependencies, such as sparse concentrated and one-to-many regimes, \acronym\ remains accurate but shows moderate looseness, especially at intermediate privacy budgets, where leakage depends on both dominant and sparse support regions. 
Third, for order-sensitive and tail-sensitive dependencies, such as near-monotone likelihood-ratio and limited-tail-mass regimes, \acronym\ becomes more conservative because small tail probabilities or gradual changes in the likelihood ratio can affect the worst-case CPL.

Overall, these results show that \acronym\ is not limited to a single dependency pattern. 
It is tightest when leakage is driven by concentrated conditional mass, while remaining conservative for sparse, asymmetric, ordered, and tail-sensitive structures.

\subsection{Privacy Budget Calibration}\label{section:Privacy Budget Calibration}

Next, we demonstrate how \acronym\ can be used to calibrate the privacy budget of LDP mechanisms, leveraging prior knowledge (with distributional uncertainty) to maximise data utility while satisfying privacy requirements (i.e., achieving a favorable privacy–utility trade-off) in a real-world scenario.

\begin{figure}[hbt] 
  \centering
  \includegraphics[trim={0.2cm 2.5cm .2cm 0.5cm},clip,width=0.9\linewidth]{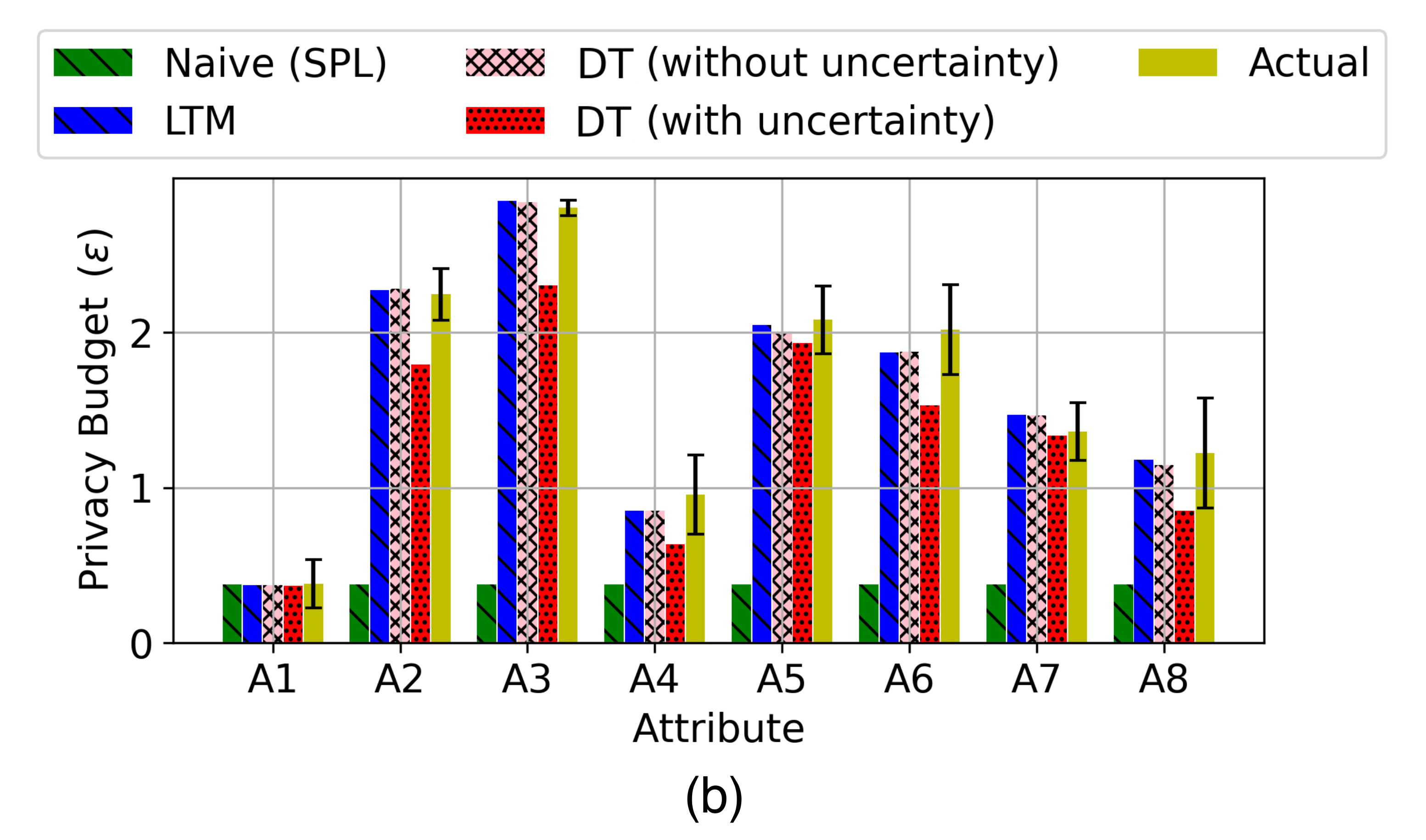}
  \caption{Calibrated privacy budgets of various algorithms at $\overline{\varepsilon}=3$. Here, `Actual' represents the mean and standard deviation of calibrated privacy budgets ($\operatorname{CPL}^\star$) for each state, except the initially collected five states.}\label{fig:privacy_budget_calibration_results}
\end{figure}

We revisit the running motivational example of state-wise data collection (Section~\ref{section:introduction}) and instantiate it on the real \textit{SPM} dataset~\cite{SPM}. 
We focus on an 8-attribute subset and assume we have prior knowledge from $5$ states (NY, CA, MD, IL, TX), yielding empirical distributions $\{P^{(s)}\}_{s=1}^5$. 
Our goal is to leverage this prior information to calibrate a per-attribute privacy-budget vector $\boldsymbol{\varepsilon}=(\varepsilon_1,\ldots,\varepsilon_8)$ that \emph{maximises utility} in data collections of \textit{remaining states}, subject to a global privacy requirement $\overline{\varepsilon}$.
Since utility typically increases with larger per-attribute budgets, calibration amounts to choosing the largest feasible $\boldsymbol{\varepsilon}$ while ensuring that the (correlation-aware) privacy leakage does not exceed $\overline{\varepsilon}$. 
We solve this constrained optimization problem using gradient descent, as detailed in Appendix~\ref{appendix:Unequal Privacy Budgets Across Attributes}.

\textbf{Nominal calibration (ignoring uncertainty).}
We first form a nominal prior by averaging the $5$ known distributions (element-wise) and calibrate $\boldsymbol{\varepsilon}$ under the \emph{no-uncertainty} assumption. Figure~\ref{fig:privacy_budget_calibration_results} shows the calibrated results for $\overline{\varepsilon}=3$. 
In this setting, \acronym\ produces calibrated budgets that closely match the theoretical optimum $CPL^\star$, which is obtained by the LTM baseline. 
This indicates that \acronym\ is \emph{tight} in the nominal regime. 
For context, the ``Actual'' calibrated budgets in Figure~\ref{fig:privacy_budget_calibration_results} report the mean and standard deviation of budgets obtained by running LTM on \emph{other} held-out states. 
The observed variation reflects the dependency shifts induced by demographic differences across states, and therefore highlights a key risk: calibrating solely on the averaged prior can be overly optimistic, and may violate the intended privacy level due to distributional mismatch.

\textbf{Uncertainty-aware calibration.}
We next quantify distributional uncertainty via the variance-based estimator presented in Section~\ref{section:Statistically Estimated Uncertainty} and calibrate $\boldsymbol{\varepsilon}$ using uncertainty-aware \acronym. 
As expected, accounting for uncertainty systematically yields \emph{smaller} per-attribute budgets, i.e., a deliberate \emph{safety margin} that hedges against plausible shifts in the dependency structure. 
This behavior is evident in Figure~\ref{fig:privacy_budget_calibration_results}(b): DT (with uncertainty) lies below LTM and closely tracks the lower envelope of the observed ``Actual'' privacy risk across attributes, thereby preventing underestimation of leakage when the data distribution changes. 
In contrast, LTM can become optimistic under mismatch (underestimating privacy risks), whereas \acronym\ remains reliable by conservatively overestimating CPL and preserving a robust privacy guarantee.

Finally, \acronym\ calibrates the full budget vector in $7$s for this example, compared to $61$s for LTM, demonstrating an order-of-magnitude speedup in \acronym.

\subsection{Attribute Inference Attacks}\label{section:inference_attacks}

Next, we evaluate \acronym{} and SPL under practical attribute inference attacks. 
We consider nine attributes with mixed dependency strengths, including both highly correlated and weakly correlated attributes as measured by mutual information. 
The data are collected using the GRR mechanism for 3-marginal estimation, and an attribute inference attack is performed separately for each target attribute.

We compare the attack accuracy of \acronym{} and SPL at the same average total variation distance (TVD), so that both methods are evaluated under comparable utility levels. 
Table~\ref{table:attribute_inference} reports the results for $\operatorname{TVD}=0.4$ and $\operatorname{TVD}=0.2$. 
As expected, SPL requires a larger overall privacy budget to achieve the same utility level as \acronym{}. 
Despite this, \acronym{} shows comparable robustness against attribute-inference attacks, even though it allocates larger per-attribute privacy budgets to selected attributes through CPL-aware budget allocation.
Additionally, we observe that weakly correlated attributes are unlikely to be exploitable via attribute-inference attacks, as inference accuracy remains largely constant across different utility levels.

These results suggest that CPL-based budget allocation improves utility while preserving comparable privacy guarantees. 
In particular, \acronym{} achieves the same TVD with a smaller overall privacy budget than SPL, while maintaining comparable attack robustness across the evaluated attributes.

\begin{table}[h]
\centering
\caption{Summary of the attribute inference attack on \acronym\ and SPL. The table reports the inference accuracy (AIA) and corresponding privacy budget for $TVD=0.2$ and $0.4$ (estimation error of 3-marginal). Here, $\varepsilon$ denotes the overall privacy budget.}
\begin{tabular}{lccccc}
\hline
Attribute & \multicolumn{1}{|c|}{$\operatorname{TVD}$} & \multicolumn{2}{c|}{$\operatorname{\acronym}$} & \multicolumn{2}{c}{$\operatorname{SPL}$}  \\ 
\cline{3-6}
  & \multicolumn{1}{|c|}{} &  \multicolumn{1}{|c|}{$\operatorname{\varepsilon}$} & \multicolumn{1}{c|}{$\operatorname{AIA}$} & $\varepsilon$ & \multicolumn{1}{|c}{$\operatorname{AIA}$} \\\hline
Marital  & 0.2 & 2.94 & 0.200  &  9.0 & 0.201 \\ 
 status & 0.4 & 1.72 & 0.187  & 6.17 & 0.190 \\ 
\hline
 Gender & 0.2 & 2.94 & 0.445 & 9.0 & 0.465 \\ 
   & 0.4 & 1.72 & 0.405 & 6.17 & 0.410 \\ 
\hline
Person's  & 0.2 & 2.94 & 0.156 & 9.0 & 0.156 \\
Medicare & 0.4 & 1.72 & 0.156 & 6.17 & 0.156 \\
\hline
Race & 0.2 & 2.94 & 0.193 & 9.0 & 0.203 \\ 
 & 0.4 & 1.72 & 0.180 & 6.17 & 0.167 \\
\hline 
Hispanic & 0.2 & 2.94 & 0.489 & 9.0 & 0.488  \\ 
& 0.4 & 1.72 & 0.489 & 6.17 & 0.489 \\
\hline
 Poverty  & 0.2 & 2.94 & 0.463 & 9.0 & 0.463  \\ 
status & 0.4 & 1.72 & 0.463 & 6.17 & 0.463 \\
\hline
Unit has   & 0.2 & 2.94 & 0.498 & 9.0 & 0.498 \\ 
a minor & 0.4 & 1.72 & 0.498 & 6.17 & 0.498 \\ 
\hline
 Unit has & 0.2 & 2.94 & 0.468 & 9.0 & 0.468 \\ 
 couples  
 & 0.4 &  1.72 & 0.468 & 6.17 & 0.468 \\ 
\hline
Child care  & 0.2 & 2.94  & 0.330 & 9.0 & 0.330 \\
expenses & 0.4 & 1.72 & 0.330 & 6.17  & 0.330 \\
\cline{2-5} 
\hline
\end{tabular}
\label{table:attribute_inference}
\end{table}

\begin{figure*}[t]
  \centering

  \includegraphics[width=0.65\textwidth]{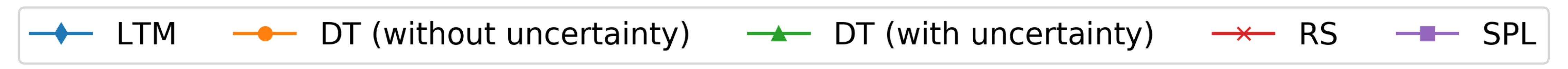}\par

  \makebox[\textwidth][c]{%
  \hfill
    \begin{subfigure}[t]{0.49\columnwidth}
      \centering
      \includegraphics[width=\linewidth]{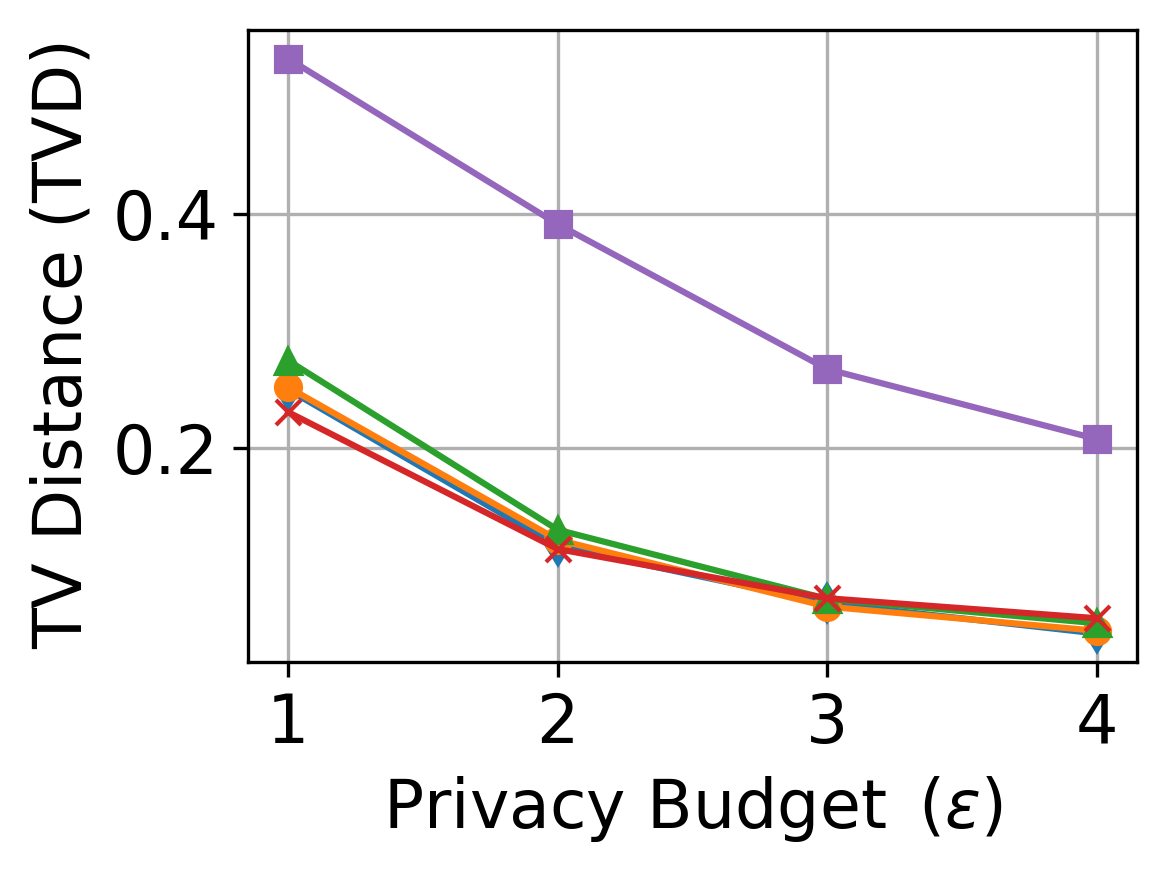}
      \caption{$2$-Marginal}
    \end{subfigure}
    \hfill
    \begin{subfigure}[t]{0.49\columnwidth}
      \centering
      \includegraphics[width=\linewidth]{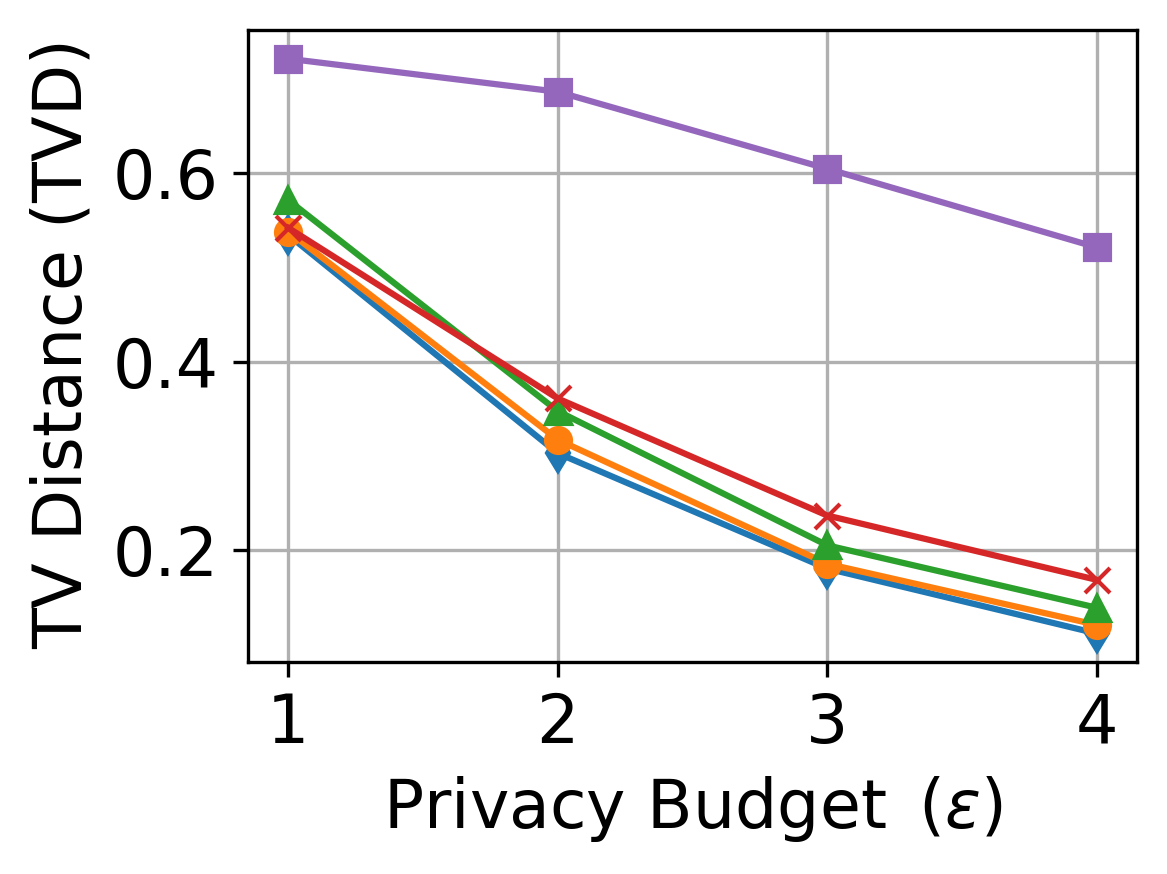}
      \caption{$3$-Marginal}
    \end{subfigure}
    \hfill
    \begin{subfigure}[t]{0.49\columnwidth}
      \centering
      \includegraphics[width=\linewidth]{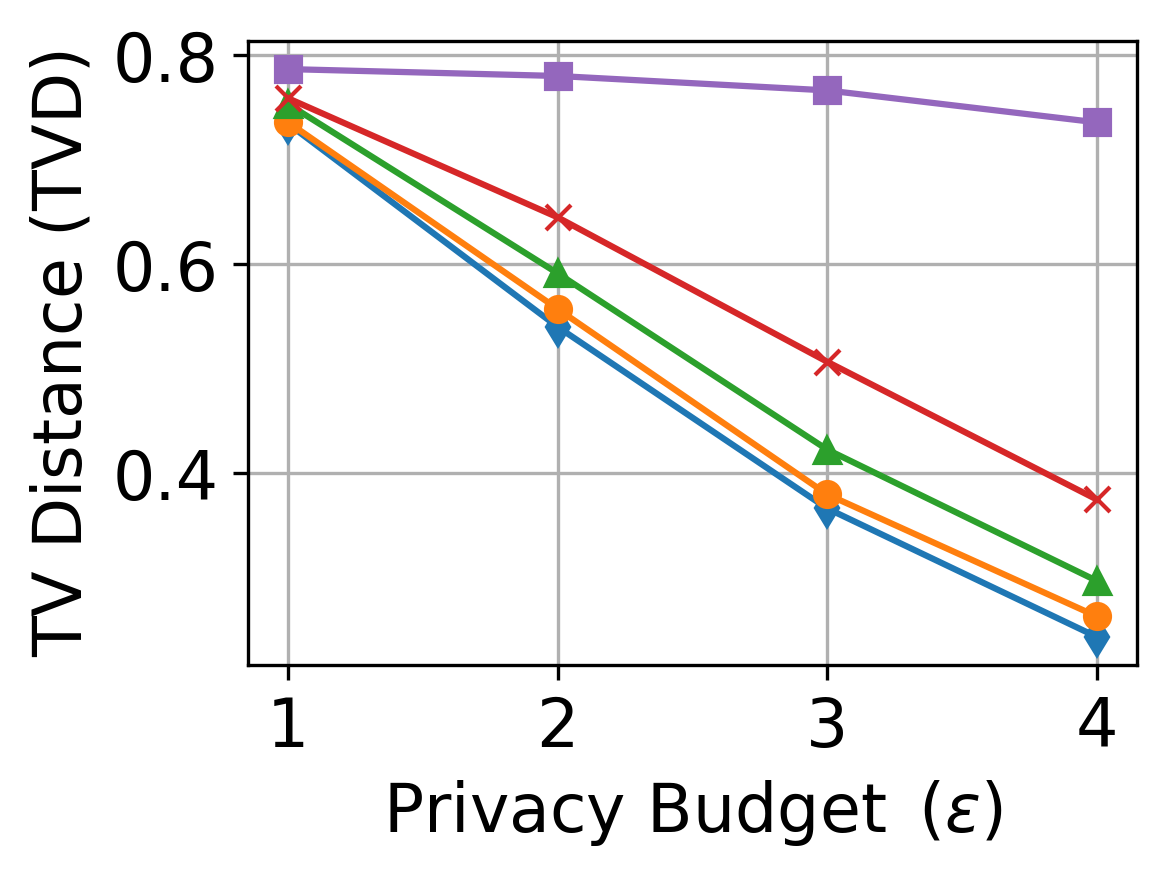}
      \caption{$4$-Marginal}
    \end{subfigure}%
    \hfill
  }

  \caption{TV distance (TVD) versus privacy budget $\varepsilon$ for $k$-marginal estimation with $k=2,3,4$. Lower TVD indicates better estimation utility. CPL-aware budget-tuning methods achieve substantially lower TVD than SPL and remain competitive with RS, with clearer advantages for higher-order marginals.}
  \label{fig:TV distance vs privacy budget}
\end{figure*}

\subsection{Comparison with Random Sampling}\label{section:random_sampling}

To evaluate whether the proposed dependency-aware budget allocation provides utility benefits beyond standard sampling-based multidimensional LDP, we include random sampling (RS) as a baseline for $k$-marginal estimation. 
In RS, each user reports only one randomly selected $k$-way attribute subset and perturbs the corresponding joint value using the full privacy budget. 
This makes RS attractive because it avoids splitting the privacy budget among multiple attributes or marginals. 
However, for joint frequency estimation, RS suffers from sample fragmentation: each $k$-way marginal is estimated only from the users who selected that specific view. 
As $k$ increases, the number of possible views and the joint domain size both grow, resulting in fewer effective samples per marginal and consequently higher estimation error.

Figure~\ref{fig:TV distance vs privacy budget} reports the TV distance for 2-, 3-, and 4-marginal estimation under different privacy budgets. 
Here, we select the same nine attributes used in the attribute inference attack.
For 2-marginals, RS is competitive with CPL-aware budget-tuning methods, since the number of views and the joint domain size remain relatively small. 
However, the advantage of \acronym{} becomes clearer for higher-order marginals. 
For 3- and 4-marginals, both variants of \acronym{} achieve lower TV distance than RS across most privacy budgets, while SPL performs substantially worse because the privacy budget is uniformly divided across attributes. 
These results demonstrate that although RS can be effective for low-dimensional marginal estimation, it becomes less reliable for higher-order joint frequency estimation. 
By exploiting the dependency structure among attributes, \acronym{} allocates the privacy budget more effectively and preserves more useful joint-distribution information, leading to better utility in correlation-sensitive estimation tasks.

\section{Related Work}\label{section:related_works}

We start this section by exploring the concept of prior knowledge and its implications for privacy risks. We then review existing CPL analysis algorithms under Differential privacy (DP) and Local differential privacy (LDP).

\subsection{CPL and Prior Knowledge}

Prior knowledge can manifest in various forms depending on the data context, ranging from statistical distributions~\cite{correlated_DP, impact_of_prior_knowledge} to temporal dependencies in time-series data~\cite{L-SRR_location_based_LBS, data_correlation_location_similar_LDP}, and spatial inferences such as estimating the current location from previously observed ones~\cite{data_correlation_location_similar_LDP, location_based_survey_LBS}. Studies have shown that adversaries may acquire such knowledge through multiple channels, including public datasets~\cite{adult_dataset, celebA_attributes}, social norms~\cite{impact_of_prior_knowledge}, or even during the data collection process itself~\cite{AAA, CALM}. These sources of prior knowledge significantly heighten privacy risks, and (L)DP mechanisms require explicitly accounting for such knowledge to safeguard privacy~\cite{impact_of_prior_knowledge, correlated_DP, CALM}.


\subsection{CPL Computation in DP and LDP}

DP~\cite{dwork2006differential} and LDP~\cite{LDP_duchi} have emerged as standard techniques for safeguarding privacy, primarily due to the limitations of earlier approaches such as \textit{k}-anonymity~\cite{k-anonymity} and \textit{l}-diversity~\cite{l-diversity}. LDP is widely employed when the data collector is considered untrusted~\cite{location_based_survey_LBS, apple2017privacy, microsoft_ldp, Local_Differential_Privacy_in_Practice}.
Study~\cite{no_free_lunch} is a pioneering work that highlights the privacy risks with correlated data under the central setting of DP. 
Then, studies propose different methods to address the privacy risks in correlated data. Group DP is a simple method to adjust the global sensitivity based on the number of correlated records~\cite{Algorithmic_Foundations_dwork}. Privacy budget splitting and sampling methods are commonly adopted methods to apply LDP in correlated data~\cite{ldp_Frequency_Estimation}.
These methods primarily assume no prior knowledge.
When prior knowledge is available, studies propose different algorithms to measure the CPL using prior knowledge. We can group them into two main categories as discussed in the following sections.

\subsubsection{Conventional Correlation Metrics Based CPL Analysis and Their Problems}
Early studies tend to adopt conventional correlation metrics (e.g., mutual information (MI), Pearson correlation coefficient (PCC), and covariance) to measure the CPL. 
Study~\cite{correlated_DP} formulates the sensitivity parameter using a correlation degree matrix. 
Another study~\cite{correlated_data_dp_conventional_metrics_MI} uses distance correlation to measure the correlation. 
In~\cite{dynamic_sensitivity_change_dp}, a periodic sensitivity calculation algorithm is proposed to address the periodic correlation changes in temporal data using autocorrelation. 
Study~\cite{secon_Collecting_High-Dimensional_Correlation_LDP} proposes a covariance-based algorithm to quantify CPL, and~\cite{LDP-PCC-dependency-graph-IoT} calibrates privacy budgets using PCC to trade-off privacy and utility. 
However, these conventional metrics are \textbf{unable} to correctly characterize some critical properties of CPL in (L)DP---e.g., asymmetric behavior and identifying the worst-case scenario~\cite{paper_1}.

\subsubsection{Bayesian Theorem-Based Algorithms to Compute CPL}

Study~\cite{Quantifying_DP} proposes an algorithm to quantify the CPL in DP. 
A recent study~\cite{data_correlation_location_similar_LDP} has analyzed privacy leakage related to spatial-temporal data, a specific family of LDP mechanisms which have two levels of perturbation probabilities (i.e., GRR~\cite{LDP_survey_composition_theorem}, Unary encoding~\cite{LDP_Frequent_Itemset_Mining}, Hash encoding~\cite{LDP_Frequent_Itemset_Mining}, etc.). 
Study~\cite{paper_1} proposes two algorithms to compute CPL in any LDP mechanism. However, all these algorithms primarily assumed that ground truth probability distributions are known, which is usually not true, especially for locally perturbing scenarios like LDP.
Furthermore, these algorithms have either polynomial or logarithmic time complexity over alphabet size, which can be challenging when CPL is iteratively evaluated at different privacy budgets to calibrate utility and privacy.

In summary, privacy leakage analysis in correlated environments is crucial to guarantee the privacy of mechanisms as well as to maximize the utility by tuning the privacy budgets.
However, the CPL analysis is currently performed using computationally expensive algorithms that assume the \emph{ground-truth} distributions are known. \textit{Therefore, currently, there is \textbf{no} method to time efficiently estimate CPL under inaccurate distributional information}.



\section{Conclusion}\label{section:conclusion}

This work demonstrates that three parameters derived from the joint distribution of two attributes are sufficient to construct a conservative surrogate for pairwise correlation-induced privacy leakage (CPL), even under distributional uncertainty. We formalise these parameters as a new metric, the \emph{Dependency Triad} (\acronym), and provide a theoretical foundation for its use. \acronym\ is not intended to fully characterise the underlying joint distribution; rather, we prove that it yields a tight, uncertainty-aware upper bound on pairwise CPL.
Further, we develop algorithms that enable practical CPL evaluation and privacy-budget calibration. 
Since exact multidimensional privacy leakage analysis is often intractable, this pairwise formulation provides a tractable building block for assessing total privacy leakage across attributes. Collectively, these results show that \acronym\ is more than an empirical heuristic: the proofs certify when and why \acronym\ remains conservative enough to mitigate privacy risks, while staying tight in the nominal regime.

Moreover, \acronym\ overcomes two key limitations of prior CPL computation methods by (i) explicitly incorporating configurable dependence through uncertainty in distributional information and (ii) enabling constant-time CPL estimation using the computed parameters of \acronym. Our empirical results corroborate the theory: \acronym\ accurately tracks CPL, particularly in high-privacy regimes (e.g., $\varepsilon \le 1$), and scales effectively to high-dimensional settings.

\noindent\textbf{Limitations and Future work - }
This work focuses on estimating CPL in pure LDP mechanisms using the proposed metric \acronym. We will extend this work to CPL estimation in approximated LDP mechanisms in our future work.

\section*{Ethical Considerations}
We use five publicly available datasets—SPM, CelebA, Adult, Cardiovascular (CVD),  DSS—under their research licenses and within their stated consent and usage scopes.




\bibliographystyle{IEEEtran}
%

\bibliography{paper/refs}

\pagebreak

\appendix


\begin{appendices}

\section*{Open Science}
Artifacts are available in \url{https://github.com/SandaruJayawardana/dependency-triad-ccs26}.

\section{The \texorpdfstring{$\operatorname{CPL}^\star$}{CPL*} Algorithm}\label{section:The Optimal Additional Privacy Leakage Algorithm}

\begin{algorithm}[h]
\caption{Compute $\operatorname{CPL}^\star$ (HCC method~\cite{paper_1}).}\label{algo:related_algorithm_1}
\begin{algorithmic}[1]
    \Statex \textbf{Input:}
    \Statex \hspace{1em} $P(\Hat{X}|X_k)$ \Comment{\textit{Conditional probability distribution.}}
    \Statex \hspace{1em} $\varepsilon$\Comment{\textit{Privacy budget of $\Hat{M}$.}}
    \State $L \gets 0$ \Comment{\textit{Variable to keep maximum leakage.}}
    \ForAll{$x,x' \in X_k$} 
        \State$ G \gets P(\Hat{X}|X_k=x)$ 
        \State$ G' \gets P(\Hat{X}|X_k=x')$ 
        \State $Q \gets \text{sort}(G \oslash G')$ \Comment{\textit{Descending order element-wise ratios.}}
        \State $A, B \gets 0$ \Comment{\textit{Variables to keep the summation.}}
        \ForAll{$Q_i \in Q$}
                \If{$Q_i \geq \frac{1+A(exp(\varepsilon)-1)}{1+B(exp(\varepsilon)-1)}$}
                    \State $A \gets A + g$ \Comment{\textit{$g$ is the element of $G$ corresponds to $Q_i$.}}
                    \State $B \gets B + g'$ \Comment{\textit{$g'$ is the element of $G'$ corresponds to $Q_i$.}}
                \EndIf
            \EndFor
                \State $L \gets \max \left(L,\frac{1+A(exp(\varepsilon)-1)}{1+B(exp(\varepsilon)-1)}\right)$ \Comment{\textit{The maximum.}}

    \EndFor
    \State \Return $L_{\hat{X} \rightarrow X_k} = \ln L$
\end{algorithmic}
\end{algorithm}

\section{ Monotonicity of Linear Fractional Functions}\label{proof:lemma:monotone_fn}
Lemma~\ref{lemma:monotone_fn} will be used in mathematical proofs.
\begin{lemma}\label{lemma:monotone_fn}
    The function \( f(x) \) defined as $f(x) = \frac{a_1 x + b_1}{a_2 x + b_2}$ is monotonic. Specifically, $f(x)$ is \textbf{increasing} with respect to $x$ if $a_1 b_2 - a_2 b_1 \geq 0$, and $f(x)$ is \textbf{strictly decreasing} with respect to $x$ if $a_1 b_2 - a_2 b_1 < 0$.
\end{lemma}

\begin{proof}
    We have $\frac{d f(x)}{d x} = \frac{a_1b_2 - a_2b_1}{(a_2x+b_2)^2}$.
    \begin{equation*}
        \text{Therefore,} \quad\quad \frac{d f(x)}{d x} \text{  } 
        \begin{cases} 
        \geq 0 & \text{; if }  a_1b_2 - a_2b_1 \geq 0, \\
        < 0 & \text{; otherwise.}
        \end{cases}
    \end{equation*}
This completes the proof.
\end{proof}

\section{Proof Theorem~\ref{thm:ratios}}\label{proof:thm:ratios}

First, let us formulate the optimization problem based on~\eqref{eqn:optimization_problem_over_C} as

\begin{equation} \label{eqn:exact_optimisation_alpha_beta}
\begin{split}
    \underset{G,G',S}{\textbf{maximize }} & \quad F(S,G,G', \varepsilon) = \frac{\sum_{i \in [b]\setminus S}G_i+e^{\varepsilon}\sum_{i \in S}G_i}{\sum_{i \in [b]\setminus S}G'_i+e^{\varepsilon}\sum_{i \in S}G'_i} \\
    \textbf{subject to }
    & \begin{split}
        \quad G_i = q_i G'_i; \quad  & \forall i \in [b],\ G_i \in G,\ G'_i \in G',\\[-0.5em]
        & q_i \in Q,\\
    \end{split}\\
    & \quad 0 < G_i,G'_i \leq 1; \quad \forall i \in[b],\\
    & \quad \sum_{i \in [b]} G_i = 1,\sum_{i \in [b]} G'_i = 1, S \subseteq [b].\\
\end{split}
\end{equation} 
Here, we need to maximize $F(\cdot)$ with respect to all $S,G,G'$. Next, we consider an alternative optimization problem~\eqref{eqn:alternative_optimization_problem_with_GG'_constraints} with relaxing the strict constraints $G_i = q_i G'_i$, $\forall i \in [b]$ in~\eqref{eqn:exact_optimisation_alpha_beta}. Also, let us assume that $S$ is fixed for the simplicity and $|S| = b-r$ (i.e., $|[b] \setminus S| = r$). Later, we show that the optimal result is independent of $S$.
\begin{equation} \label{eqn:alternative_optimization_problem_with_GG'_constraints}
\begin{split}
    \underset{G,G'}{\textbf{maximize }} & \quad F(S,G,G', \varepsilon) = \frac{\sum_{i \in [b]\setminus S}G_i+e^{\varepsilon}\sum_{i \in S}G_i}{\sum_{i \in [b]\setminus S}G'_i+e^{\varepsilon}\sum_{i \in S}G'_i} \\
    \textbf{subject to }
    & \quad G_i \leq e^{\alpha}G'_i;\quad \forall i \in [b], \alpha = \ln \max_{q \in Q} q,\\
    & \quad G'_i \leq e^{\gamma}G_i;\quad \forall i \in [b], \gamma =-\ln \min_{q \in Q} q,\\
    & \quad 0 <  G_i,G'_i \leq 1;\quad \forall i\in [b],\\
    & \quad \sum_{i \in [b]} G_i = 1,\quad\sum_{i \in [b]} G'_i = 1.
\end{split}
\end{equation} 

Next, we establish intermediate results: Proposition~\ref{prop:g_g'_boundaries} and Corollary~\ref{corollary:optimisation_problem_similar_t_equals_2} to prove Theorem~\ref{thm:ratios}.

\begin{proposition}\label{prop:g_g'_boundaries}
Using the same notation as above, $\linebreak\max_{G,G',\ G \oslash G' \in \tilde{Q}} F(S,G,G',\varepsilon) = \frac{e^{\alpha} - 1+e^{\alpha+\varepsilon}(e^{\gamma} - 1)}{e^{\varepsilon}(e^{\gamma} - 1)+ e^{\gamma}(e^{\alpha}-1)}$.
\end{proposition}

\begin{proof}
Let us take $F^\star$ as the optimal solution for \eqref{eqn:alternative_optimization_problem_with_GG'_constraints}.

Since $\sum_{i \in [b]}G_i= 1$ and $\sum_{i \in [b]} G'_i = 1$, we have
\begin{equation}
    F(S,G,G',\varepsilon) = \frac{
    1-\sum_{j \in S} G_j+e^{\varepsilon}\sum_{j \in S}G_j}{1-\sum_{j \in S} G'_j+e^{\varepsilon}\sum_{j \in S} G'_j}.
\end{equation}
Moreover, we can take $G_i=g$ and $G'_i=g'$ for all $i \in S$ as it does not change the optimal solution for~\eqref{eqn:alternative_optimization_problem_with_GG'_constraints}. Then we can write $F(S,G,G',\varepsilon)$ in a form of $U(g,g',\varepsilon,b,r)$ as
\begin{equation}
    U(g,g',\varepsilon,b,r) = \frac{1+(b-r)(e^\varepsilon-1)g}{1 + (b-r)(e^{\varepsilon}-1)g'}.
\end{equation}
Next, let us take the ratio $\frac{g}{g'} = q$. Here, $q$ should be $q>1$ to maximize $U(\cdot)$. Then we can write $U(g,g',\varepsilon,b,r)$ in a form of $V(q,g',\varepsilon,b,r)$ as
\begin{equation}
    V(q,g',\varepsilon,b,r) = \frac{1+(b-r)(e^\varepsilon-1)qg'}{1 + (b-r)(e^{\varepsilon}-1)g'}.
\end{equation}
Here, we know that $V(\cdot)$ monotonically increases with $g'$ from Lemma~\ref{lemma:monotone_fn} as $(b-r)(e^\varepsilon-1)(q-1) \geq 0$.
We can calculate the feasible $g'$ values based on the constraint $G'_i \leq e^{\gamma}G_i$ for all $i\in [b]$ in~\eqref{eqn:alternative_optimization_problem_with_GG'_constraints}. From this constraint, we have 
\begin{equation}
\sum_{i \in [b]\setminus S} G_i' \leq e^\gamma \sum_{i \in [b]\setminus S} G_i. 
\end{equation}
Here, $\sum_{i \in [b]\setminus S} G_i'$ and $\sum_{i \in [b]\setminus S} G_i$ can be written using terms of $g'$ and $q$ as $\sum_{i \in [t]\setminus S} G_i'=1-(b-r)g'$ and $\sum_{i \in [b]\setminus S} G_i=1-(b-r)qg'$. Therefore, we have 
\begin{equation}
    0 \leq g' \leq \frac{e^{\gamma}-1}{(qe^{\gamma}-1) (b-r)}.
\end{equation}
Note that the upper bound of $g'$ monotonically decreases as $q$ increases. Moreover, $V(\cdot)$ monotonically increases with $q$ as $\linebreak\frac{(b-r)(e^\varepsilon-1)g'}{1 + (b-r)(e^{\varepsilon}-1)g'} >0$. Therefore, to check which $g'$ and $q$ values maximize $V(\cdot)$, let us compute $V(q,g'=\frac{e^{\gamma}-1}{(qe^{\gamma}-1) (b-r)}, \varepsilon,b,r)$ as
\begin{equation}
    V(\cdot) = \frac{-1+ \bigl(1+e^\varepsilon (e^\gamma -1)\bigr)q}{e^\varepsilon (e^\gamma -1)-e^\gamma + e^\gamma q}.
\end{equation}
Here, we know that $V(q,g'=\frac{e^{\gamma}-1}{(qe^{\gamma}-1) (b-r)}, \varepsilon,b,r)$ monotonically increases with $q$ from Lemma~\ref{lemma:monotone_fn} as $ \bigl(1+e^\varepsilon (e^\gamma -1)\bigr)\bigl(e^\varepsilon (e^\gamma - 1) - e^\gamma \bigr) - e^\gamma= e^\varepsilon(e^\varepsilon - 1)(e^\gamma - 1)^2\geq 0$. Since $q \in \{1,e^\alpha\}$, the optimal solution $F^\star$ is given by $q = e^\alpha$ and $g' = \frac{e^{\gamma}-1}{(qe^{\gamma}-1) (b-r)}$,
\begin{equation}\label{eqn:H*_value}
    F^\star = \frac{e^{\alpha} - 1+e^{\alpha+\varepsilon}(e^{\gamma} - 1)}{e^{\varepsilon}(e^{\gamma} - 1)+ e^{\gamma}(e^{\alpha}-1)}.
\end{equation}
Therefore, the optimal solution for~\eqref{eqn:alternative_optimization_problem_with_GG'_constraints} is given by when $\frac{G}{G'}\in \{e^\alpha, e^{-\gamma}\}^b$. Thus $\max_{G,G', G \oslash G' \in \tilde{Q}} F(S,G,G',\varepsilon) = F^\star$. This completes the proof of Proposition~\ref{prop:g_g'_boundaries}.
\end{proof}

Based on Proposition~\ref{prop:g_g'_boundaries}, we can state the following Corollary~\ref{corollary:optimisation_problem_similar_t_equals_2}, which is required to prove Theorem~\ref{thm:ratios}.

\begin{corollary}\label{corollary:optimisation_problem_similar_t_equals_2}
    When $b = 2$, the solution for the optimization problem~\eqref{eqn:exact_optimisation_alpha_beta} can be solved with Proposition~\ref{prop:g_g'_boundaries}.
\end{corollary}

\begin{proof}
    The constraint $G_i = q_i G'_i$ in~\eqref{eqn:exact_optimisation_alpha_beta} becomes $\tilde{Q}$ for $b=2$. Therefore, we can compute the optimal solution of~\eqref{eqn:exact_optimisation_alpha_beta} using the Proposition~\ref{prop:g_g'_boundaries}.
    
    This completes the proof of Corollary~\ref{corollary:optimisation_problem_similar_t_equals_2}.
\end{proof}

\noindent\textbf{Proof of Theorem~\ref{thm:ratios}}

\begin{proof}
Next, let us move to the proof of Theorem~\ref{thm:ratios}. We can find a scenario where all $G_i$ and $G_i'$ values go to zero except the value corresponding to the maximum ratio and minimum ratio (i.e., $e^{-\gamma}$ and $e^\alpha$) for any $G,G'$ without violating the constraints in~\eqref{eqn:exact_optimisation_alpha_beta}. In such a case, the impact on the objective function of~\eqref{eqn:exact_optimisation_alpha_beta} is equivalent to $b=2$. Then, from Corollary~\ref{corollary:optimisation_problem_similar_t_equals_2}, the solution for~\eqref{eqn:exact_optimisation_alpha_beta} equals to~\eqref{eqn:H*_value} (i.e., Proposition~\ref{prop:g_g'_boundaries}). Therefore, the \emph{correlation-induced privacy leakage (CPL)} satisfies $\max_{(G,G')\in\mathcal{H}}\operatorname{CPL}(G,G')=\max_{(G,G')\in\mathcal{\tilde{H}}}\operatorname{CPL}(G,G')= \ln\frac{e^{\alpha} - 1+e^{\alpha+\varepsilon}(e^{\gamma} - 1)}{e^{\varepsilon}(e^{\gamma} - 1)+ e^{\gamma}(e^{\alpha}-1)}$.

This completes the proof of Theorem~\ref{thm:ratios}.
\end{proof}

\section{Proof of Corollary~\ref{corollary:alpha_greater_than_beta}}\label{proof:corollary:alpha_greater_than_beta}

\begin{proof}
    Let $U_1$ and $U_2$ be estimated CPL values corresponds to $(\alpha_1,\gamma_1)$ and $(\alpha_2,\gamma_2)$ respectively.
    \begin{align*}
    U_1 &= \frac{
        e^{\alpha_1} - 1+e^{\alpha_1+\varepsilon}(e^{\gamma_1} - 1)}
        {e^{\varepsilon}(e^{\gamma_1} - 1)+ e^{\gamma_1}(e^{\alpha_1}-1)}\\
    U_2 &= \frac{
        e^{\alpha_2} - 1+e^{\alpha_2+\varepsilon}(e^{\gamma_2} - 1)}
        {e^{\varepsilon}(e^{\gamma_2} - 1)+ e^{\gamma_2}(e^{\alpha_2}-1)}\\
        &\text{Since $\gamma_2=\alpha_1$ and $\alpha_2=\gamma_1$},\\
        &= \frac{
        e^{\gamma_1} - 1+e^{\gamma_1+\varepsilon}(e^{\alpha_1} - 1)}
        {e^{\varepsilon}(e^{\alpha_1} - 1)+ e^{\alpha_1}(e^{\gamma_1}-1)}
    \end{align*}
    Next, let us evaluate $U_1 - U_2$
\begin{equation*}
    \begin{split}
    &= \frac{
        e^{\alpha_1} - 1+e^{\alpha_1+\varepsilon}(e^{\gamma_1} - 1)}
        {\lambda_1} - \frac{e^{\gamma_1} - 1+e^{\gamma_1+\varepsilon}(e^{\alpha_1} - 1)}
        {\lambda_2}\\
    &= \frac{(-1 + e^{\alpha_1})(e^{\alpha_1} - e^{\gamma_1})(-1 + e^{\gamma_1})(-1 + e^{\varepsilon})^2}{\lambda_1 \lambda_2}.
    \end{split}
    \end{equation*}
    Here, $\lambda_1 = e^{\varepsilon}(e^{\gamma_1} - 1)+ e^{\gamma_1}(e^{\alpha_1}-1)$ and $\lambda_2 =e^{\varepsilon}(e^{\alpha_1} - 1)+ e^{\alpha_1}(e^{\gamma_1}-1)$.
    Since $\alpha_1 \geq 1$ and $ \gamma_1 \geq 1$, denominator is positive (i.e., $\lambda_1, \lambda_2 > 0$). Therefore, if and only if $\alpha_1 \geq \gamma_1$ then $U_1 - U_2 \geq 0$. This completes the proof.
\end{proof}

\section{Proof of Theorem~\ref{thm:alpha_beta_delta_apl}}\label{proof:thm:alpha_beta_delta_apl}

\begin{proof}
    Based on new constraints in~\eqref{eqn:delta_constraints}, we can formulate the updated optimization problem as
\begin{equation} \label{eqn:new_reformed_optimization_problem_with_GG'_constraints}
\begin{split}
    \underset{S,G,G'}{\textbf{maximize }} & \quad F(S,G,G', \varepsilon) = \frac{1+\sum_{i \in S}G_i (e^{\varepsilon}-1)}{1+\sum_{i \in S}G'_i(e^{\varepsilon}-1)} \\
    \textbf{subject to }
    & \quad G_i \leq e^{\alpha}G'_i+\tilde{\delta}_i; \forall i\in [b],\\
    & \quad G'_i \leq e^{\gamma}G_i+\hat{\delta}_i; \forall i\in [b],\\
    & \quad 0 < G_i,G'_i \leq 1; \forall i\in [b],\\
    & \quad \sum_{i \in [b]} G_i=1, \sum_{i \in [b]} G'_i=1,S\subseteq [b].\\
\end{split}
\end{equation} 

\subsection{Simplifying the Problem With Fixed \texorpdfstring{$S$}{S}}
First, let us fix $S$ (see~\eqref{eqn:new_reformed_optimization_problem_with_GG'_fixed_s}) and solve the above optimization. Next, we repeat that for all $S\subseteq [b]$ and find the maximum value as the solution for~\eqref{eqn:new_reformed_optimization_problem_with_GG'_constraints}.
\begin{equation} \label{eqn:new_reformed_optimization_problem_with_GG'_fixed_s}
\begin{split}
    \underset{G,G'}{\textbf{maximize }} & \quad F(S,G,G', \varepsilon) = \frac{1+\sum_{i \in S}G_i (e^{\varepsilon}-1)}{1+\sum_{i \in S}G'_i(e^{\varepsilon}-1)} \\
    \textbf{subject to }
    & \quad G_i \leq e^{\alpha}G'_i+\tilde{\delta}_i; \forall i\in [b],\\
    & \quad G'_i \leq e^{\gamma}G_i+\hat{\delta}_i; \forall i\in [b],\\
    & \quad 0 \leq G_i,G'_i \leq 1; \forall i\in [b],\\
    & \quad \sum_{i \in [b]} G_i=1, \sum_{i \in [b]} G'_i=1,S\subseteq [b].\\
\end{split}
\end{equation} 

Let $g=\sum_{i\in S}G_i$ and $g'=\sum_{i\in S}G'_i$. Then, we can write the objective function as 
\begin{equation}
     F(\cdot) = \frac{1+g (e^{\varepsilon}-1)}{1+g'(e^{\varepsilon}-1)}.
\end{equation}
Since the objective function is monotonically increasing with respect to $g$, maximizing $F(\cdot)$ requires maximizing $g$. We therefore consider an upper bound function $V(\cdot)$ of $F(\cdot)$, obtained by replacing $g$ with its upper bound $\overline{g}$. Let $V^\star$ denote the maximum value of $V(\cdot)$. We then show that there always exist attainable vectors $G$ and $G'$ such that $F(\cdot) = V^\star$, which implies the optimality of $V^\star$ for fixed $S$.

From the constraint $G_i \leq e^{\alpha}G'_i+\tilde{\delta}_i,\ \forall G_i \in G, G'_i \in G'$, we have
\begin{equation}
\begin{split}
    \sum_{i\in [S]} G_i &\leq e^{\alpha} \sum_{i \in [S]} G'_i+\sum_{i \in [S]} \tilde{\delta}_i, \\
    g &\leq e^\alpha g' +\tilde{\delta}.
\end{split}
\end{equation}
Here, $\sum_{i\in S}\tilde{\delta}_i =\tilde{\delta}$.

From the constraint $G'_i \leq e^{\gamma}G'_i+\hat{\delta}_i,\ \forall G_i \in G, G'_i \in G'$, we have
\begin{equation}
\begin{split}
    \sum_{i\in [b]\setminus S} G'_i &\leq e^{\gamma} \sum_{i \in [b]\setminus S} G_i+\sum_{i \in [b]\setminus S} \hat{\delta}_i \\
    1-g' &\leq e^\gamma(1-g) +\hat{\delta}\\
    g &\leq \frac{e^\gamma + \hat{\delta} + g' - 1}{e^\gamma}.
\end{split}
\end{equation}
Here, $\sum_{i\in [b]\setminus S}\hat{\delta}_i =\hat{\delta}$.

There are two upper bounds for $g$ (i.e., $\overline{g}$) in terms of $g'$: $g \leq \min \{\frac{e^\gamma + \hat{\delta} + g' - 1}{e^\gamma}, \ e^\alpha g' +\tilde{\delta}\}$.

Then, we can write an upper bound $V(g',\delta,\varepsilon,\alpha)$ for $F(\cdot)$ such that
\begin{equation}
    V(g',\delta,\varepsilon,\alpha) = \frac{1 + \overline{g}(e^\varepsilon-1)}{1+ g'(e^\varepsilon-1)},
\end{equation}
where $\overline{g} = \min \{ \frac{e^\gamma + \hat{\delta} + g' - 1}{e^\gamma}, e^\alpha g' +\tilde{\delta}\}$.
Next, we compute $V^\star$, which is the maximum value of $V(\cdot)$ for $g' \in [0,1]$. 

\subsubsection{Compute $V^\star$}

We check the valid value range of $g'$ and how $V(\cdot)$ behaves with $g'$ for $\overline{g} =\frac{e^\gamma + \hat{\delta} + g' - 1}{e^\gamma}$ and $\overline{g} = e^\alpha g' +\tilde{\delta}$.


\noindent $\bullet$ If $\overline{g} = \frac{e^\gamma + \hat{\delta} + g' - 1}{e^\gamma}$ then 
\begin{equation}
    V(g',\delta,\varepsilon,\alpha) = \frac{1+e^{-\gamma}(e^\gamma + \hat{\delta} + g' - 1)(e^\varepsilon-1)}{1+ g'(e^\varepsilon-1)}.
\end{equation}
Since $e^{-\gamma}(e^\varepsilon-1)-(e^\varepsilon-1)(1+e^{-\gamma}(e^\gamma+\hat{\delta}-1)(e^\varepsilon-1))=-e^{-\gamma}(e^\varepsilon-1)(e^\varepsilon(e^\gamma-1)+\hat{\delta}(e^\varepsilon-1)) < 0$, $V(g',\delta,\varepsilon,\alpha)$ monotonically decreases with $g'$ as per Lemma~\ref{lemma:monotone_fn}. Further, we have
\begin{equation}
\begin{split}
    g \leq \frac{e^\gamma + \hat{\delta} + g' - 1}{e^\gamma} &\leq e^{\alpha}g'+\tilde{\delta}\\
    g' &\geq \frac{e^\gamma-1 - \tilde{\delta}e^\gamma +\hat{\delta}}{e^{\alpha+\gamma}-1}.    
\end{split}
\end{equation}
Therefore, $\frac{e^\gamma + \hat{\delta} + g' - 1}{e^\gamma} \leq e^{\alpha}g'+\tilde{\delta}$ holds for $g' \in [\frac{e^\gamma-1 - \tilde{\delta}e^\gamma +\hat{\delta}}{e^{\alpha+\gamma}-1} ,1]$ and $g'=\max\{0,\frac{e^\gamma-1 - \tilde{\delta}e^\gamma +\hat{\delta}}{e^{\alpha+\gamma}-1}\}$ maximizes $V$.

\noindent $\bullet$ If $\overline{g} = e^\alpha g' +\tilde{\delta}$ then 
\begin{equation}
    V(g',\delta,\varepsilon,\alpha) = \frac{1+(e^\alpha g' +\tilde{\delta})(e^\varepsilon-1)}{1+ g'(e^\varepsilon-1)}.
\end{equation}
Since $e^\alpha(e^\varepsilon-1)-(e^\varepsilon-1)(1+\tilde{\delta}(e^\varepsilon-1))=(e^\varepsilon-1)(e^\alpha-1-\tilde{\delta}(e^\varepsilon-1))$, from Lemma~\ref{lemma:monotone_fn}, we know that $V(\cdot)$ monotonically increases with $g'$ if $e^\alpha-1-\tilde{\delta}(e^\varepsilon-1) \geq 0$. Otherwise $V(\cdot)$ monotonically decreases with $g'$. Further, we have
\begin{equation}
\begin{split}
    g \leq e^{\alpha}g'+\tilde{\delta} &\leq \frac{e^\gamma + \hat{\delta} + g' - 1}{e^\gamma} \\
    g' &\leq \frac{e^\gamma-1 - \tilde{\delta}e^\gamma +\hat{\delta}}{e^{\alpha+\gamma}-1}.    
\end{split}
\end{equation}
Therefore, $e^{\alpha}g'+\tilde{\delta} \leq \frac{e^\gamma + \hat{\delta} + g' - 1}{e^\gamma}$ holds for $g' \in [0,\frac{e^\gamma-1 - \tilde{\delta}e^\gamma +\hat{\delta}}{e^{\alpha+\gamma}-1}]$. If $e^\alpha-1-\tilde{\delta}(e^\varepsilon-1) < 0$, then $g'=0$ gives the maximum $V$. Otherwise $g'=\min\{\frac{1-\tilde{\delta}}{e^\alpha},\frac{e^\gamma-1 - \tilde{\delta}e^\gamma +\hat{\delta}}{e^{\alpha+\gamma}-1}\}$ gives the maximum $V$, as $g=1$ when $g' = \frac{1-\tilde{\delta}}{e^\alpha}$. 

Interestingly, $e^\alpha g' +\tilde{\delta}$ and $\frac{e^\gamma + \hat{\delta} + g' - 1}{e^\gamma}$ become upper bounds for $g$ in the mutually exclusive $g'$ range. 
\begin{enumerate}
    \item If $g' \in [0,\frac{e^\gamma-1 - \tilde{\delta}e^\gamma +\hat{\delta}}{e^{\alpha+\gamma}-1}]$ then $\overline{g} = e^\alpha g' +\tilde{\delta}$.
    \item If $g' \in [\frac{e^\gamma-1 - \tilde{\delta}e^\gamma +\hat{\delta}}{e^{\alpha+\gamma}-1}, 1]$ then $\overline{g} = \frac{e^\gamma + \hat{\delta} + g' - 1}{e^\gamma}$.
\end{enumerate}


Furthermore, $e^\alpha g' +\tilde{\delta}=\frac{e^\gamma + \hat{\delta} + g' - 1}{e^\gamma}$ at $g'=\frac{e^\gamma-1 - \tilde{\delta}e^\gamma +\hat{\delta}}{e^{\alpha+\gamma}-1}$. Therefore, we can compute $V^\star$ based on the value of $\frac{e^\gamma-1 - \tilde{\delta}e^\gamma +\hat{\delta}}{e^{\alpha+\gamma}-1}$ under three cases.  

\textbf{Case 1 - ($\frac{e^\gamma-1 - \tilde{\delta}e^\gamma +\hat{\delta}}{e^{\alpha+\gamma}-1} \leq 0$)}
Since $g' \geq 0$, the feasible upper bound is $\overline{g} = \frac{e^\gamma + \hat{\delta} + g' - 1}{e^\gamma}$. Therefore, if $\frac{e^\gamma-1 - \tilde{\delta}e^\gamma +\hat{\delta}}{e^{\alpha+\gamma}-1} \leq 0$ then $g'=0$ gives the maximum $V(\cdot)$ as
\begin{equation}
    V^\star=1+e^{-\gamma}(e^\gamma + \hat{\delta}  - 1)(e^\varepsilon-1).
\end{equation}

\textbf{Case 2 - ($0 < \frac{e^\gamma-1 - \tilde{\delta}e^\gamma +\hat{\delta}}{e^{\alpha+\gamma}-1} \leq \frac{1-\tilde{\delta}}{e^\alpha}$)}  When $ e^\alpha-1-\tilde{\delta}(e^\varepsilon-1) \geq 0$ both $\overline{g} = e^\alpha g' + \tilde{\delta}$ and $\overline{g} = \frac{e^\gamma + \hat{\delta} + g' - 1}{e^\gamma}$ have the same maximum $V(\cdot)$ at $g'=\frac{e^\gamma-1 - \tilde{\delta}e^\gamma +\hat{\delta}}{e^{\alpha+\gamma}-1}$ as $e^\alpha g' +\tilde{\delta}=\frac{e^\gamma + \hat{\delta} + g' - 1}{e^\gamma}$. If $ e^\alpha-1-\tilde{\delta}(e^\varepsilon-1) < 0$ then  the maximum $V(\cdot)$ of $\overline{g} = e^\alpha g' +\tilde{\delta}$ is $1+\tilde{\delta}(e^\varepsilon-1)$ which is larger than the maximum $V(\cdot)$ given by $\overline{g} = \frac{e^\gamma + \hat{\delta} + g' - 1}{e^\gamma}$.
Therefore, 
\begin{equation}\label{eqn:V_upper_bound2}
    V^\star=
    \begin{cases} 
     1+\tilde{\delta}(e^\varepsilon-1)
    \text{; if }  e^\alpha-1-\tilde{\delta}(e^\varepsilon-1) < 0, \\
     \frac{
 e^{\alpha+\gamma+\varepsilon} - e^{\alpha}(-1+\hat{\delta} ) + e^{\alpha+\varepsilon}(-1+\hat{\delta} )+ \tilde{\delta} - e^{\varepsilon}\tilde{\delta}-1}{-\hat{\delta} + e^{\varepsilon}\left(-1+\hat{\delta} - e^{\gamma}(-1+\tilde{\delta})\right)+ e^{\gamma}\left(-1 + e^{\alpha} + \tilde{\delta}\right)}\text{; Otherwise.}
    \end{cases}
\end{equation}

\textbf{Case 3 - ($\frac{1-\tilde{\delta}}{e^\alpha} < \frac{e^\gamma-1 - \tilde{\delta}e^\gamma +\hat{\delta}}{e^{\alpha+\gamma}-1}$)} Here, both $\overline{g}=e^\alpha g' +\tilde{\delta}$ and $\overline{g}=\frac{e^\gamma + \hat{\delta} + g' - 1}{e^\gamma}$ become $>1$ when $g' > \frac{1-\tilde{\delta}}{e^\alpha}$. However, we know that $\max g = 1$. Therefore, $\overline{g}=e^\alpha g' +\tilde{\delta}$ gives more tight upper bound for $g$ for this case. Therefore,
\begin{equation}\label{eqn:V_upper_bound}
    V^\star=
    \begin{cases} 
     1+\tilde{\delta}(e^\varepsilon-1)
    \text{; if }  e^\alpha-1-\tilde{\delta}(e^\varepsilon-1) < 0, \\
      \frac{e^{\varepsilon+\alpha}}{e^\alpha+ (1-\tilde{\delta})(e^\varepsilon-1)}\text{; Otherwise.}
    \end{cases}
\end{equation}



\subsubsection{Optimality of $V^\star$}
Next, let us evaluate the optimality of $V^\star$. 
Interestingly, each of the above cases has attainable $G$ and $G'$. 



For example, consider the case 
\[
\frac{e^{\gamma} - 1 - \tilde{\delta} e^{\gamma} + \hat{\delta}}
     {e^{\alpha + \gamma} - 1} \le 0.
\]
In this scenario, we select $G'_i$ for all $i \in [b] \setminus S$ as 
\[
G'_i = e^{\gamma} G_i + \hat{\delta}_i,
\]
such that $\sum_{i \in [b] \setminus S} G'_i = 1$ and $\sum_{i \in S} G'_i = 0$. 
Accordingly, for all $i \in S$, we set $G_i \le \tilde{\delta}_i$ to ensure
\[
\sum_{i \in S} G_i = \frac{e^{\gamma} + \hat{\delta} - 1}{e^{\gamma}}
\quad \text{and} \quad
\sum_{i \in [b] \setminus S} G_i = \frac{1 - \hat{\delta}}{e^{\gamma}}.
\]
Since these assignments satisfy the constraints in~\eqref{eqn:new_reformed_optimization_problem_with_GG'_fixed_s}, 
there exists at least one feasible pair $(G_i, G'_i)$ for all $i \in [b]$. 
Substituting these values yields $F(\cdot) = V^\star$. 

A similar construction shows that feasible $(G_i, G'_i)$ exist for the remaining two cases: $0 < \frac{e^{\gamma} - 1 - \tilde{\delta} e^{\gamma} + \hat{\delta}}{e^{\alpha + \gamma} - 1}\leq \frac{1 - \tilde{\delta}}{e^{\alpha}}\text{ and }\frac{1 - \tilde{\delta}}{e^{\alpha}} <\frac{e^{\gamma} - 1 - \tilde{\delta} e^{\gamma} + \hat{\delta}}{e^{\alpha + \gamma} - 1}$.
Therefore, $V^\star$ serves as an attainable upper bound for $F(\cdot)$, and there exist feasible vectors 
$G$ and $G'$ such that $F(\cdot) = V^\star$. 
Hence, $V^\star$ is the optimal value of the problem in~\eqref{eqn:new_reformed_optimization_problem_with_GG'_fixed_s}.

\subsection{Variable \texorpdfstring{$S$}{S}}

As we defined $\tilde{\delta}$ and $\hat{\delta}$ in~\eqref{eqn:define_delta}, we have $\tilde{\delta_i} > 0 \implies \hat{\delta_i}=0$ and $\hat{\delta_i} > 0 \implies \tilde{\delta_i}=0$ for all $i\in[b]$.
Next, let us consider the variable $S$ scenario. Here, we need to obtain the maximum $V(\cdot)$ for all $S\subseteq [b]$. Note that, only $\tilde{\delta}$ and $\hat{\delta}$ depends on $S$. 
 Since $V(\cdot)$ is monotonically increases with $\tilde{\delta}$ and $\hat{\delta}$ as stated in Corollary~\ref{corollary:monotone_V_over_deltas}. Therefore, the maximum $V(\cdot)$ is given by any subset $S \subseteq [b]$ when $\sum_{i\in S} \tilde{\delta} = \sum_{i\in [b]} \tilde{\delta} $ and $\sum_{i\in [b]\setminus S} \hat{\delta} = \sum_{i\in [b]} \hat{\delta}$. Therefore, at the maximum of $V(\cdot)$, we have $\tilde{\delta}=\sum_{i\in[b] }\tilde{\delta}_i $ and $\hat{\delta}=\sum_{i\in[b] }\hat{\delta}_i $. 

For the simplicity let us take a single parameter $\delta$ instead of two parameters $\hat{\delta}$ and $\tilde{\delta}$ such that $\delta = \max\{\hat{\delta},\tilde{\delta}\}$. This $\delta$ holds the constraints in~\eqref{eqn:new_reformed_optimization_problem_with_GG'_constraints}. Furthermore, this simplifies the optimal results as Case 1 and Case 3 are not feasible. 
\begin{align*}
    \hat{\delta} = \tilde{\delta}=\delta &\implies \frac{e^\gamma-1 - \tilde{\delta}e^\gamma +\hat{\delta}}{e^{\alpha+\gamma}-1}  = \frac{(e^\gamma-1) (1-\delta)}{e^{\alpha+\gamma}-1} \geq 0,\\
    \hat{\delta} = \tilde{\delta}=\delta &\implies \frac{1-\tilde{\delta}}{e^\alpha} - \frac{e^\gamma-1 - \tilde{\delta}e^\gamma +\hat{\delta}}{e^{\alpha+\gamma}-1}\\
    &\qquad\ =\frac{(e^\alpha-1)(1-\delta)}{e^\alpha (e^{\alpha+\gamma} - 1)} \geq 0.
\end{align*}
 Therefore, the maximum of $V(\cdot)$ (i.e., $\operatorname{CPL(\alpha,\gamma,\delta)}$) is given by
\begin{equation}
    \operatorname{CPL}=
    \begin{cases} 
     \ln \bigl(1+\delta(e^\varepsilon-1)\bigr)
    \text{; if }  e^\alpha-1-\delta(e^\varepsilon-1) < 0, \\
     \ln \left(\frac{e^{\alpha} + \delta - e^{\alpha}\delta - \delta e^{\varepsilon}
      + e^{\alpha+\varepsilon}(-1 + e^{\gamma} + \delta)-1} {\delta(e^{\varepsilon}-1) - e^{\varepsilon}
      + e^{\gamma}(-1 + e^{\alpha} + \delta + e^{\varepsilon} - \delta e^{\varepsilon})}\right)\text{; otherwise.}
    \end{cases}
\end{equation}

In addition, $\operatorname{CPL(\alpha,\gamma,\delta)}$ can be characterised as 
\begin{equation}\label{eqn:alternative_rep}
    \operatorname{CPL(\alpha,\gamma, \delta)} = \frac{1+(e^\alpha g' + \delta)(e^\varepsilon-1)}{1+g'(e^\varepsilon-1)}.
\end{equation}
This result is used to prove Proposition~\ref{prop:g_prime} in Section~\ref{proof:thm:beta}.

This completes the proof of Theorem~\ref{thm:alpha_beta_delta_apl}.
\end{proof}

\begin{corollary}\label{corollary:monotone_V_over_deltas}
    $V^\star$ monotonically increases with both $\tilde{\delta}$ and $\hat{\delta}$. 
\end{corollary}
\begin{proof}
    First, we analyze how $V^\star$ behaves in each case.
    
    \noindent\textbf{Case 1:} Here, $V^\star$ monotonically increases with $\hat{\delta}$ and does not vary with $\tilde{\delta}$.

    \noindent\textbf{Case 2:} Here, if $e^\alpha-1-\tilde{\delta}(e^\varepsilon-1) <0$ then $V^\star$ monotonically increases with $\tilde{\delta}$ and does not vary with $\hat{\delta}$ as $V^\star = 1 + \tilde{\delta}(e^\varepsilon - 1)$.
    If $e^\alpha-1-\tilde{\delta}(e^\varepsilon-1) \geq 0$ then $V^\star$ monotonically increases with $\tilde{\delta}$ according to Lemma~\ref{lemma:monotone_fn} as $(-1 + e^{\alpha + \gamma})(e^{\varepsilon}-1)(e^{\varepsilon}(e^{\gamma}-1) + \hat{\delta}(e^{\varepsilon}-1))\geq 0$. Similarly, we can show that $V^\star$ monotonically increases with $\hat{\delta}$.

    \noindent\textbf{Case 3:} Similar to Case 2, if $e^\alpha-1-\tilde{\delta}(e^\varepsilon-1) <0$ then $V^\star$ monotonically increases with $\tilde{\delta}$ and does not vary with $\hat{\delta}$.
    If $e^\alpha-1-\tilde{\delta}(e^\varepsilon-1) \geq 0$ then $V^\star$ monotonically increases with $\tilde{\delta}$ and does not vary with $\hat{\delta}$.

    Therefore, in each case, $V^\star$ monotonically increases with $\hat{\delta}$ and $\tilde{\delta}$. Next, let us consider the transition between cases due to increasing the $\hat{\delta}$ or $\tilde{\delta}$. It can be easily proved that $V^\star$ is continuous between each case over $\tilde{\delta}$ and $\hat{\delta}$. 
    
    Therefore, $V^\star$ monotonically increases with both $\tilde{\delta}$ and $\hat{\delta}$. This completes the proof of Corollary~\ref{corollary:monotone_V_over_deltas}.
\end{proof}

\section{Proof of Theorem~\ref{thm:beta}}\label{proof:thm:beta}

\begin{proof}

First, we prove the calibrated level $e^{-\beta}$. According to Thm.~\ref{thm:alpha_beta_delta_apl}, Computed CPL does not depends on $\gamma$ when $e^\alpha-1-\delta(e^\varepsilon-1) < 0$. Therefore, we do not need a calibrated $\beta$ for this case as $\beta$ is a replacement for $\gamma$. From Theorem~\ref{thm:alpha_beta_delta_apl}, we can compute $\operatorname{CPL}$ when $e^\alpha-1-\delta(e^\varepsilon-1) \geq 0$ with,
\begin{equation}
        = \ln \left(\frac{e^{\alpha} + \delta - e^{\alpha}\delta - \delta e^{\varepsilon}
      + e^{\alpha+\varepsilon}(-1 + e^{\gamma} + \delta)-1} {\delta(e^{\varepsilon}-1) - e^{\varepsilon}
      + e^{\gamma}(-1 + e^{\alpha} + \delta + e^{\varepsilon} - \delta e^{\varepsilon})}\right).
\end{equation}

Let $e^{-\beta}$ be the updated lower bound by calibrating at privacy budget $\varepsilon_0$ and let $l$ be $\operatorname{CPL}^\star$ at privacy budget $\varepsilon_0$. Then we have  
$l = \ln \left(\frac{e^{\alpha} + \delta - e^{\alpha}\delta - \delta e^{\varepsilon_0}
      + e^{\alpha+\varepsilon_0}(-1 + e^{\beta} + \delta)-1} {\delta(e^{\varepsilon_0}-1) - e^{\varepsilon_0}
      + e^{\beta}(-1 + e^{\alpha} + \delta + e^{\varepsilon_0} - \delta e^{\varepsilon_0})}\right)$.
Therefore, 
\begin{equation}
    \beta =\ln \frac{ e^{\alpha} + \delta + e^{l} \delta - e^{\alpha} \delta - e^{\varepsilon_0} \bigl ((e^{l} - e^{\alpha})(-1 + \delta) + \delta \bigr)-1}{-e^{\varepsilon_0} \bigl( e^{\alpha} + e^{l} (-1 + \delta) \bigr) + e^{l} (-1 + e^{\alpha} + \delta)}.
\end{equation}

Next, we prove that the estimated CPL of $(\alpha,\beta,\delta)$ method is an upper bound for the $\operatorname{CPL}^\star$ for $\varepsilon>\varepsilon_0$. To prove this, we first establish the following intermediate results: Proposition~\ref{prop:g_prime}, Lemma~\ref{lemma:monotonically_increasing_f} and Proposition~\ref{prop:g_g'_boundaries}.

\begin{proposition}\label{prop:g_prime}
    Estimated CPL (calibrated at privacy budget $\varepsilon_0$) can be characterised using $\ln\frac{1+(e^\alpha g'+\delta)(e^\varepsilon - 1)}{1+ g'(e^\varepsilon - 1)}$. Here, $g' = \linebreak \frac{A - \delta + B(-1 + \delta - e^{\varepsilon_0} \delta)}{-1 + A - A e^{\varepsilon_0} + e^{\alpha} \left(1 + B(-1 + e^{\varepsilon_0})\right)}$, and $A,B$ are sum of selected elements of $\overline{G}$ and $\overline{G'}$ respectively by the Algorithm~\ref{algo:related_algorithm_1} at $\varepsilon_0$.
\end{proposition}

\begin{proof}
From \eqref{eqn:alternative_rep}, we can compute the estimated $\operatorname{CPL}$ as $\operatorname{CPL} =\ln \frac{1+ (e^\alpha g'+\delta) (e^{\varepsilon}-1)}{1+ g' (e^{\varepsilon}-1)}$.

Let $l$ be $\operatorname{CPL}^\star$ at privacy budget $\varepsilon_0$. Here, $A,B$ are sum of selected elements of $\overline{G}$ and $\overline{G'}$ respectively by the Algorithm~\ref{algo:related_algorithm_1} at $\varepsilon_0$. Then we have, $l = \ln\frac{1+ A (e^{\varepsilon_0}-1)}{1+ B(e^{\varepsilon_0}-1)} = \ln\frac{1+ (e^\alpha g' +\delta)(e^{\varepsilon_0}-1)}{1+ g' (e^{\varepsilon_0}-1)}$. Therefore,
\begin{equation}\label{eqn:g_prime}
g' = \frac{A - \delta + B(-1 + \delta - e^{\varepsilon_0} \delta)}
{-1 + A - A e^{\varepsilon_0} + e^{\alpha} \left(1 + B(-1 + e^{\varepsilon_0})\right)}.
\end{equation}
This completes the proof of Theorem~\ref{thm:beta}.
\end{proof}

\begin{lemma}\label{lemma:monotonically_increasing_f}
    Let $\frac{1+(e^\alpha g+\delta)(e^\varepsilon-1)}
     {1+g(e^\varepsilon-1)}$, where $\alpha > 0,\ \varepsilon >0,\ 0 \leq g \leq 1$. 
     If $e^\alpha-1-\delta(e^\varepsilon-1)\geq0$, then $f(g)$ is monotonically non-decreasing in $g$.
\end{lemma}

\begin{proof}
    From Lemma~\ref{lemma:monotone_fn}, the function $f(g)$ is monotonically increasing with $g$ when $(e^\varepsilon-1)\bigl(e^\alpha - \delta(e^\varepsilon - 1) -1\bigl) \geq 0$.
    Therefore, if $e^\alpha-1-\delta(e^\varepsilon-1)\geq0$, then $f(g)$ is monotonically non-decreasing in $g$.

    This completes the proof of Lemma~\ref{lemma:monotonically_increasing_f}.
\end{proof}

\begin{proposition}\label{prop:f1_geq_f2}
    Let $g'_0 = \frac{A - \delta + B(-1 + \delta - e^{\varepsilon_0} \delta)}{-1 + A - A e^{\varepsilon_0} + e^{\alpha} \left(1 + B(-1 + e^{\varepsilon_0})\right)}$ and $g'_1 = \frac{A - \delta + B(-1 + \delta - e^{\varepsilon_1} \delta)}{-1 + A - A e^{\varepsilon_1} + e^{\alpha} \left(1 + B(-1 + e^{\varepsilon_1})\right)}$  where, $0 < A, B < 1$, $\quad 0< \varepsilon_0 < \varepsilon_1$. Then, $g'_0 \geq g'_1$.
\end{proposition}

\begin{proof}
    Let us consider $g'_0 - g'_1$,

    \begin{equation*}
        \begin{split}
           &= \begin{split}
                &\frac{A - \delta + B(-1 + \delta - e^{\varepsilon_0} \delta)}{-1 + A - A e^{\varepsilon_0} + e^{\alpha} \left(1 + B(-1 + e^{\varepsilon_0})\right)} - \\
                &\qquad \qquad \frac{A - \delta + B(-1 + \delta - e^{\varepsilon_1} \delta)}{-1 + A - A e^{\varepsilon_1} + e^{\alpha} \left(1 + B(-1 + e^{\varepsilon_1})\right)}
            \end{split}\\
            &= \begin{split}
                &\frac{(A - \delta + B(-1 + \delta - e^{\varepsilon_0} \delta))\omega_1}{\omega_0 \omega_1} -\\
            &\qquad\qquad\frac{(A - \delta + B(-1 + \delta - e^{\varepsilon_1} \delta))\omega_0}{\omega_0 \omega_1}
            \end{split}\\
            &\text{Here, $\omega_0 = 1 + A (e^{\varepsilon_0} - 1)+ e^\alpha(B - B e^{\varepsilon_0} -1)$, and}\\
            & \text{$\omega_1 = 1 + A (e^{\varepsilon_1} - 1)+ e^\alpha(B - B e^{\varepsilon_1} -1)$.}\\
            &= \frac{(B-A) (e^{\varepsilon_1} - e^{\varepsilon_0})(A-Be^{\alpha}-\delta)}{\omega_0 \omega_1}\\
            & \text{Since, } (B-A) \leq 0, \ (e^{\varepsilon_1} - e^{\varepsilon_0}) > 0, \ (A-Be^{\alpha}-\delta) \leq 0,\\
            &g_0' - g'_1 \geq 0.
        \end{split}
    \end{equation*}

    This completes the proof of Proposition~\ref{prop:f1_geq_f2}.
\end{proof}

    Next, let us move to the proof that the estimated CPL of $(\alpha,\beta,\delta)$ method is an upper bound for the $\operatorname{CPL}^\star$ for $\varepsilon>\varepsilon_0$. Let $\operatorname{CPL_0}=\frac{1+(e^\alpha g'_0+\delta)(e^\varepsilon - 1)}{1+ g'_0(e^\varepsilon - 1)}$ and $\operatorname{CPL_1}=\frac{1+(e^\alpha g'_1+\delta)(e^\varepsilon - 1)}{1+ g'_1(e^\varepsilon - 1)}$ be two estimated CPLs at $\varepsilon$ that  are calibrated at privacy budgets $\varepsilon_0$ and $\varepsilon_1$, respectively using~\eqref{eqn:alternative_rep}. Let $A_0,B_0$ and $A_1,B_1$ be the sum of selected elements of $\overline{G},\overline{G'}$ respectively at $\varepsilon_0$ and $\varepsilon_1$ by Algorithm~\ref{algo:related_algorithm_1}. Moreover, we know that the Algorithm~\ref{algo:related_algorithm_1} selects elements of $\overline{G},\overline{G'}$ in a discrete manner over privacy budget $\varepsilon$. For example, suppose $A,B$ are the selected $\overline{G},\overline{G'}$ respectively at $\varepsilon'$. Then, $A,B$ would not change for a range around $\varepsilon'$, which we call a segment. Therefore, the previously considered $A_0,B_0$ and $A_1,B_1$ can be vary based on their segments. We can analyze them in three different cases.

    \textbf{Case 1 - } (within same segment, i.e., $A_0=A_1$ and $B_0=B_1$) According to Prop.~\ref{prop:g_prime}, $g'_0=\frac{A_0 - \delta + B_0(-1 + \delta - e^{\varepsilon_0} \delta)}{-1 + A_0 - A_0 e^{\varepsilon_0} + e^{\alpha} \left(1 + B_0(-1 + e^{\varepsilon_0})\right)}$ and $g'_1=\frac{A_1 - \delta + B_1(-1 + \delta - e^{\varepsilon_1} \delta)}{-1 + A_1 - A_1 e^{\varepsilon_1} + e^{\alpha} \left(1 + B_1(-1 + e^{\varepsilon_1})\right)}$. From Prop.~\ref{prop:f1_geq_f2}, we have $g'_0 > g'_1$ as $\varepsilon_0 \leq \varepsilon_1$. Therefore, from Lemma~\ref{lemma:monotonically_increasing_f}, the estimated CPL by calibrated parameters at $\varepsilon_0$ is an upper bound for the CPL by calibrated parameters at $\varepsilon_1$. Therefore, the estimated CPL by calibrated parameters at $\varepsilon_0$ is an upper bound for $\operatorname{CPL}^\star$ within the same segment.

    \textbf{Case 2 - } (between neighboring segments, i.e., $A_0 < A_1$ and $B_0 \leq B_1$) Let us take $\operatorname{CPL}_2$ be estimated CPL at $\varepsilon_2$ which is calibrated at $\varepsilon_2$. Here, $\varepsilon_2$ is selected such that the two segments join each other. From previous case 1, $\operatorname{CPL_0} \geq \operatorname{CPL}_2$ and $\operatorname{CPL_2} \geq \operatorname{CPL}_1$. Therefore, $\operatorname{CPL_0} \geq \operatorname{CPL}_1$.

    \textbf{Case 3 - } (between any two segments $\varepsilon > \varepsilon_0$) Using Case 1 and Case 2, we have $\operatorname{CPL_0} \geq \operatorname{CPL}_1$ for any two segments.\\

Therefore, estimated CPL by calibrated parameters at $\varepsilon_0$ provides an upper bound for $\operatorname{CPL}^\star$ for privacy budgets $\varepsilon>\varepsilon_0$.
This completes the proof of Theorem~\ref{proof:thm:beta}.
\end{proof}

\section{Proof of Corollary~\ref{corollary:beta_upper_bound}}\label{proof:corollary:beta_upper_bound}

\begin{proof}

From Theorem~\ref{proof:thm:beta}, $\beta =\ln \frac{-1 + e^{\alpha} + \delta + e^{l} \delta - e^{\alpha} \delta - e^{\varepsilon_0} \bigl (e^{l} - e^{\alpha})(-1 + \delta) + \delta \bigr)}{-e^{\varepsilon_0} \bigl( e^{\alpha} + e^{l} (-1 + \delta) \bigr) + e^{l} (-1 + e^{\alpha} + \delta)}$. Since $\operatorname{CPL}^\star$ at $\varepsilon_0$ can be computed as $l = \ln\frac{1+ A (e^{\varepsilon_0}-1)}{1+ B(e^{\varepsilon_0}-1)}$, we can write $\beta$ in terms of $A$ and $B$ as $\beta = \ln \frac{\lambda_1+\lambda_2}{e^{\varepsilon_0} (B e^{\alpha} + A(-1 + \delta)) - (-1 + A)(-1 + e^{\alpha} + \delta)}$. Here, $\lambda_1=-1 + e^{\alpha} + A e^{\varepsilon_0} (-1 + \delta) + B (-1 + e^{\alpha})(-1 + \delta),\ \lambda_2=2\delta - A\delta - e^{\alpha} \delta + B e^{\varepsilon_0} (e^{\alpha} + \delta - e^{\alpha} \delta)$ and $A,B$ are selected elements from $G,G'$ respectively by the Algorithm~\ref{algo:related_algorithm_1} at privacy budget $\varepsilon_0$. Then, 
\begin{equation}
    \lim_{\varepsilon_0 \rightarrow 0}\beta = \ln\frac{1 + \tilde{A} - \tilde{B} + e^{\alpha}(-1 + \delta) - 2\delta}{1 + (-1 + \tilde{A} - \tilde{B}) e^{\alpha} - \delta}.
\end{equation} 
Here, $\tilde{A} = \sum_{i\in [b], \frac{G_i}{G'_i}>1,} G_i$ and $\tilde{B} = \sum_{i\in [b], \frac{G_i}{G'_i}>1,} G'_i$. 

Furthermore, we know that $\gamma \geq \ln\frac{1-\tilde{B}}{1-\tilde{A}}$ as $\gamma = \ln \sup_{i\in[b],\ G_i\neq 0} \frac{G'_i}{G_i}$ from~\eqref{eqn:gamm_and_alpha_with_delta}. Next, let us evaluate
\begin{equation*}
    \begin{split}
        &= \frac{1 + \tilde{A} - \tilde{B} + e^{\alpha}(-1 + \delta) - 2\delta}{1 + (-1 + \tilde{A} - \tilde{B}) e^{\alpha} - \delta}-\frac{1-\tilde{B}}{1-\tilde{A}}\\
        &= \frac{\tilde{A}\tilde{B}-\tilde{A}^2-\delta+2\tilde{A}\tilde{B}-\delta \tilde{B}+e^\alpha ((\tilde{A}-\tilde{B})\tilde{B}+\delta-\tilde{A}\delta)}{(1 + (-1 + \tilde{A} - \tilde{B}) e^{\alpha} - \delta)(1-\tilde{A})}\\
        &= \frac{(\tilde{A}-\tilde{B})(\tilde{B}e^\alpha+\delta-\tilde{A})+\delta(1-\tilde{A})(e^\alpha-1)}{(1 + (-1 + \tilde{A} - \tilde{B}) e^{\alpha} - \delta)(1-\tilde{A})}\\
        &\text{Since $0\leq \tilde{B} \leq \tilde{A} <1$ and $\tilde{A}\leq e^\alpha \tilde{B} + \delta$},\\
        &\leq 0\\
    \end{split}
\end{equation*}
Therefore,
\begin{equation}
    \lim_{\varepsilon_0 \rightarrow 0}\beta \leq \gamma.
\end{equation}
This completes the proof of Corollary~\ref{corollary:beta_upper_bound}.
\end{proof}

\section{Proof of Theorem~\ref{corollary:alpha_beta_delta_apl}}\label{proof:corollary:alpha_beta_delta_apl}

\begin{proof}
    As per Theorem~\ref{thm:beta}, the calibrated parameter $\beta$ holds the privacy guarantee of the estimated CPL. Therefore, we can compute the estimated CPL of the ($\alpha,\beta,\delta$) method in the same way as in Theorem~\ref{thm:alpha_beta_delta_apl}.
\end{proof}

\section{Proof of Corollary~\ref{corollary:apl_monotonially_increases_with_beta}}\label{proof:corollary:apl_monotonially_increases_with_beta}

\begin{proof}
    From~\eqref{eqn:apl_cal_delta_available}, $\operatorname{CPL}(\alpha,\gamma,\delta)$ does not depends on $\gamma$ when $e^\alpha-1-\delta(e^\varepsilon-1) < 0$.

    When $e^\alpha-1-\delta(e^\varepsilon-1) \geq 0$, let us check the monotonicity by evaluating $a_1b_1 - a_2b_2$. Here, we remove $\ln(x)$ as it is a monotonically increasing function with $x$.
    \begin{equation*}
    \begin{split}
        &a_1b_1 - a_2b_2 \\
        &= e^{\alpha+\varepsilon}(\delta(e^{\varepsilon}-1) - e^{\varepsilon})-(-1 + e^{\alpha} + \delta + e^{\varepsilon} - \delta e^{\varepsilon})\lambda\\
        &\text{Here, } \lambda=(e^{\alpha} + \delta - e^{\alpha}\delta - \delta e^{\varepsilon}
      + e^{\alpha+\varepsilon}(-1  + \delta)-1)\\
        &= - (e^{\alpha}-1)\,( e^{\varepsilon}-1)\,( \delta-1)\,(e^{\alpha} + \delta - e^{\varepsilon}\delta-1) \geq 0.
    \end{split}
    \end{equation*}
    Since $a_1b_1 - a_2b_2 \geq 0$, CPL is monotonically increasing with $\gamma$.
    This completes the proof of Corollary~\ref{corollary:apl_monotonially_increases_with_beta}.
\end{proof}

\begin{figure*}[t]
  \centering
  \includegraphics[trim={0.8cm 0.cm 0.7cm 0.2cm},width=0.3\linewidth]{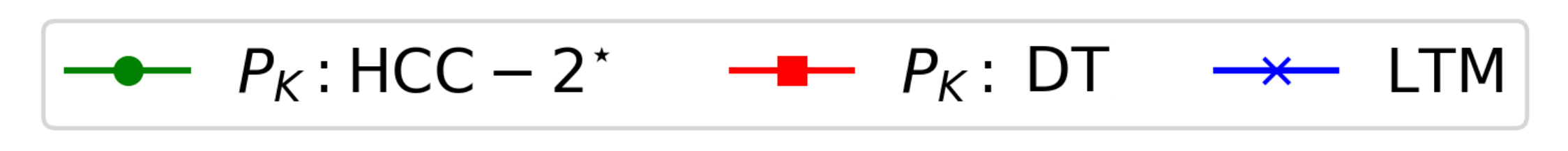}\par
  \makebox[0.9\linewidth][c]{%
    \begin{subfigure}[t]{0.25\linewidth}
      \centering
      \includegraphics[trim={1.45cm 0.8cm 1.2cm 0.8cm},clip,width=0.97\textwidth]{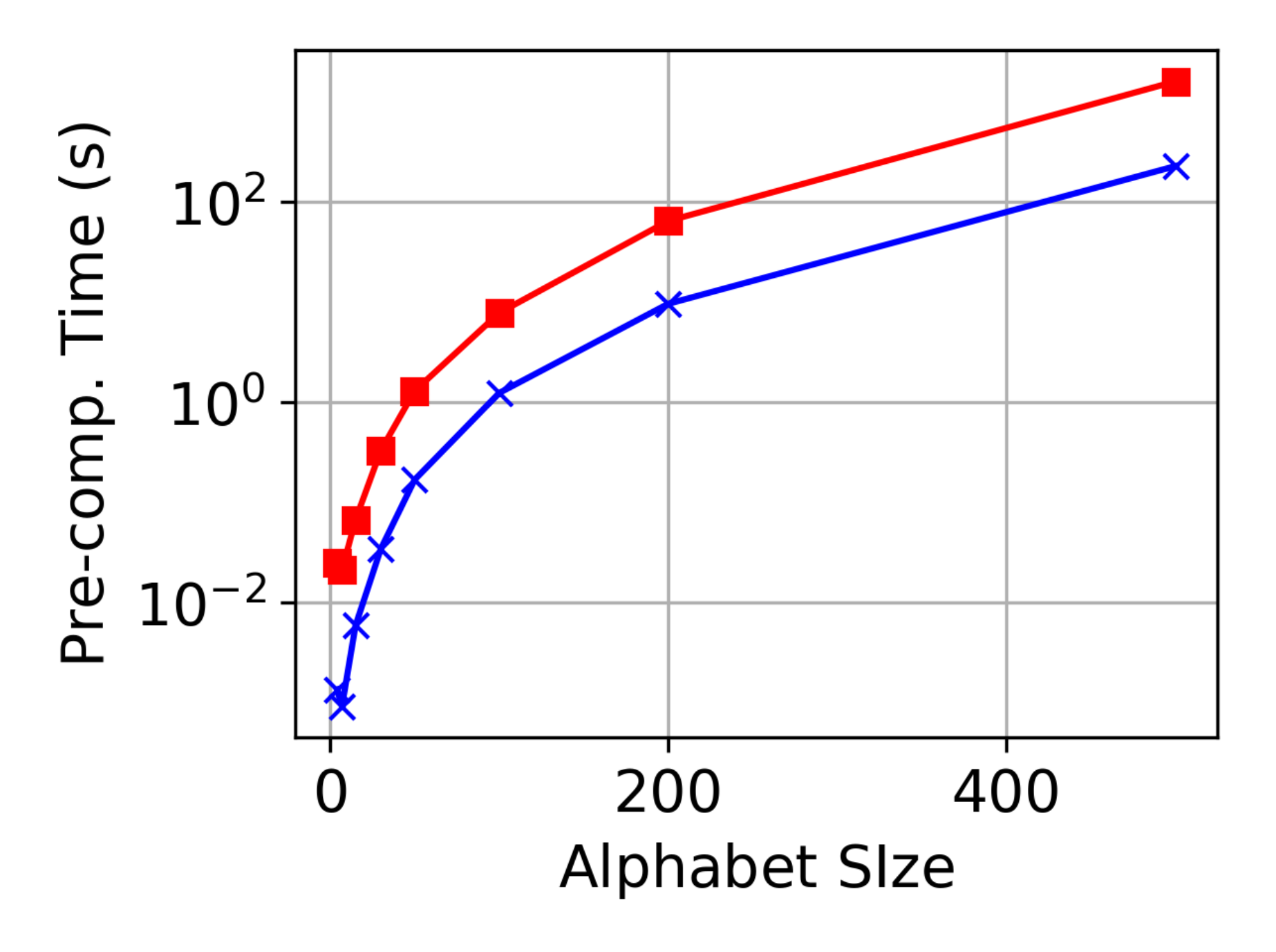}
      \caption{}\label{fig:time_space_a}
    \end{subfigure}
    \hspace{0.07\linewidth}
    \begin{subfigure}[t]{0.25\linewidth}
      \centering
      \includegraphics[trim={1.45cm 0.8cm 1.2cm 0.8cm},clip,width=0.97\textwidth]{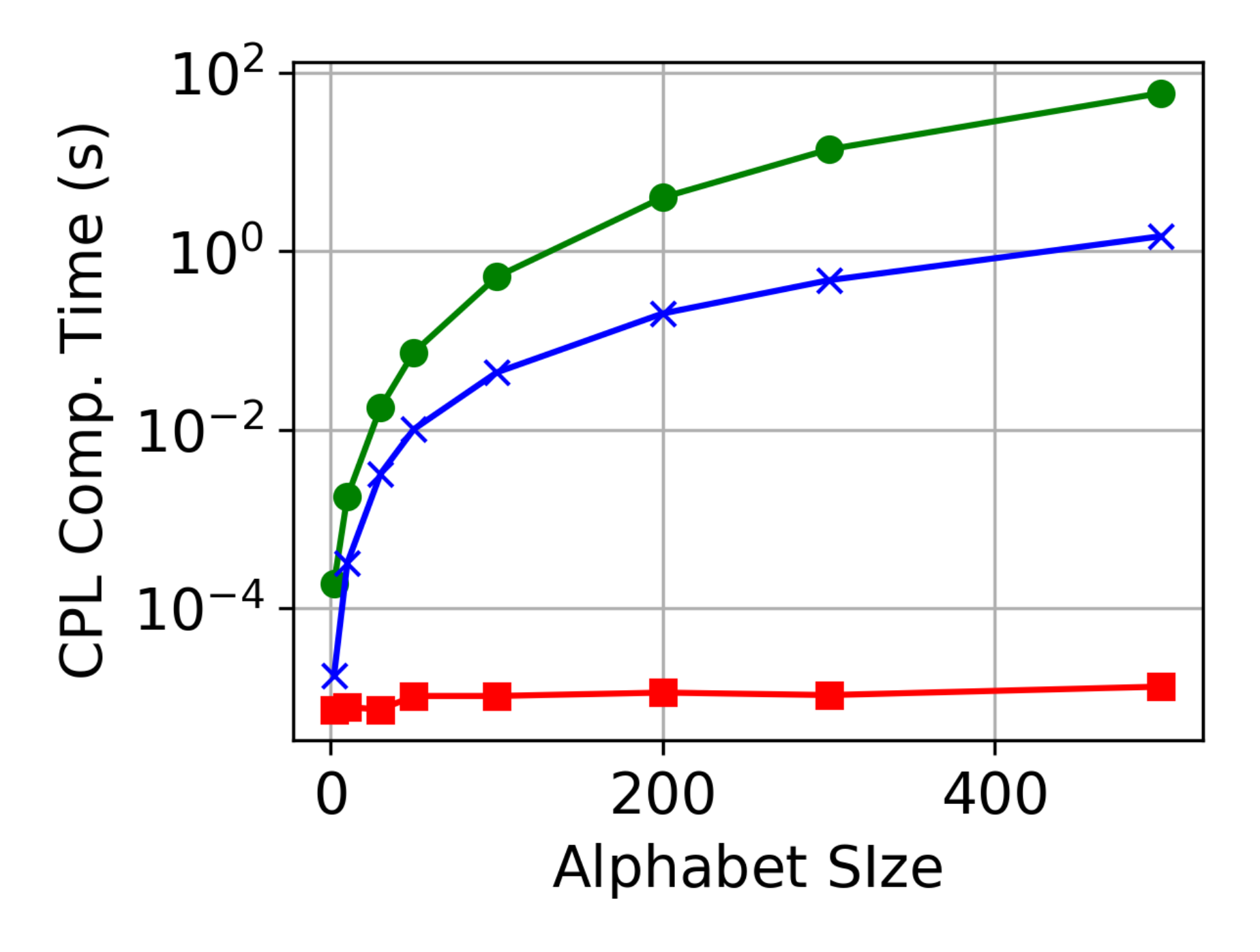}
      \caption{}\label{fig:time_space_b}
    \end{subfigure}
    \hspace{0.07\linewidth}
    \begin{subfigure}[t]{0.25\linewidth}
      \centering
      \includegraphics[trim={1.45cm 0.8cm 1.2cm 0.8cm},clip,width=0.97\textwidth]{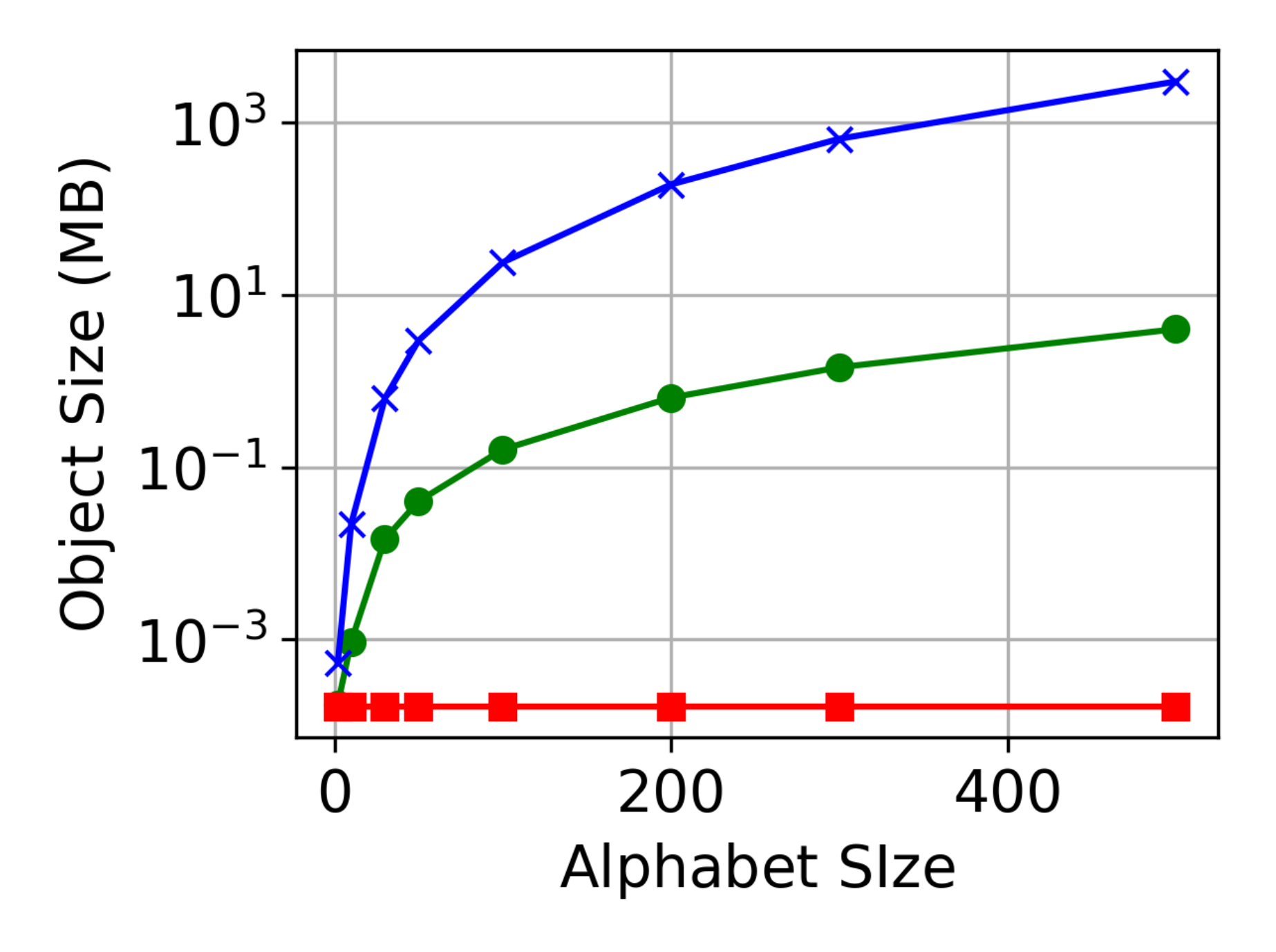}
      \caption{}\label{fig:time_space_c}
    \end{subfigure}%
  }
  \caption{Empirically validate the space and time complexities of \acronym. Here, space is measured with the Python object allocated to store the values.}
    \label{fig:space_time_empirical}
\end{figure*}

\section{Proof of Theorem~\ref{thm:upperbound_dist_bias}}\label{proof:thm:upperbound_dist_bias}

We use Corollary~\ref{Corollary:composition_apl} to prove Theorem~\ref{thm:upperbound_dist_bias}.
\begin{corollary}\label{Corollary:composition_apl}
        Let ($\alpha_1,\beta_1,\delta_1$) and ($\alpha_2,\beta_2,\delta_2$) be two characterisations for the Dependency budget. Then, ($\overline{\alpha},\overline{\beta},\overline{\delta}$) gives an upper bound for both characterisations. Here, $\overline{\alpha} = \max \{\alpha_1,\alpha_2\}, \overline{\beta} = \max \{\beta_1,\beta_2\}$, and $\overline{\delta} = \max \{\delta_1,\delta_2\}$.
\end{corollary}

\begin{proof}
    It can be easily proved that the estimated CPL given by Corollary~\ref{corollary:alpha_beta_delta_apl} monotonically increases with $\alpha,\beta$, and $\delta$ using Lemma~\ref{lemma:monotone_fn}.
\end{proof}

\begin{proof} 
From Corollary~\ref{Corollary:composition_apl}, we need to compute the maximum values of $\alpha$ and $\delta$ for all possible distributions within the uncertainty. Let $\hat{\alpha}$ and $\hat{\delta}$ be the updated values of $\alpha$ and $\delta$, respectively.
    
\textbf{Compute $\hat{\alpha}$}.
Since we consider a $\Delta$ uncertainty between known distribution (i.e., $P_K$) and actual distribution (i.e., $P_A$), $|p_A(a|x)-p_K(a|x)| \leq \Delta_{a,x}$ for all $a\in \mathcal{\hat{X}},\ x \in \mathcal{X}_k$, where $\Delta_{a,x} \in \Delta,\ p_A(\cdot|\cdot) \in P_A,\ p_K(\cdot|\cdot) \in P_K$. Therefore, updated $\hat{\alpha}$ should be the maximum of all possible $\alpha$ which is given by,
\begin{equation}
    \hat{\alpha} = \max_{a\in \mathcal{\hat{X}},x,x'\in\mathcal{X}_k,p(a|x')-\Delta_{a,x'}> 0} \left\{\frac{\min\{1, p(a|x) + \Delta_{a,x}\}}{p(a|x')- \Delta_{a,x'}}\right\}.
\end{equation}
    
\textbf{Compute $\hat{\delta}$}. Similarly the maximum $\delta$ can be computed as
\begin{equation}
    \hat{\delta} = \max_{x,x' \in \mathcal{X}_k} \sum_{a \in \mathcal{\hat{X}},\ p(a|x')-\Delta_{a,x'} \leq 0} p(a|x)+ \Delta_{a,x}.
\end{equation}
    
\textbf{Compute $\hat{\beta}$}. From Corollary~\ref{corollary:beta_upper_bound}, $\hat{\beta}$ can be computed for given distribution by,
\begin{equation}
        \hat{\beta} =\ln\frac{1 + A-B + e^{\hat{\alpha}}(-1 + \hat{\delta}) - 2\hat{\delta}}{1 + (-1 + A-B) e^{\hat{\alpha}} - \hat{\delta}}.
\end{equation}
Since $\hat{\beta}$ is monotonically increasing with $A-B$ (can be proved using Lemma~\ref{lemma:monotone_fn} similar to the proof~\ref{proof:corollary:apl_monotonially_increases_with_beta}), we need to find the maximum $A-B$ for all possible distributions within the uncertainty to compute $\hat{\beta}$. We can formulate this as a relaxed linear program, which provides a polynomial-time upper bound to the exact MILP formulation. 
\begin{equation} \label{eqn:optimisation_max_beta_with_uncertainty}
\begin{split}
\underset{Z,G,G'}{\textbf{maximize}} \quad &\sum_{z_i \in Z, \ i\in [b]} z_i \\[0.4em]
\textbf{subject to} \quad
& 0\leq z_i \leq \frac{M^+_i}{M^+_i + M^-_i}( G_i - G'_i +M^-_i);\quad \forall i\in [b],\\
& L \leq G \le U,\quad \mathbf{1}^\top G=1,\\
& L' \leq G' \le U',\quad \mathbf{1}^\top G'=1.
\end{split}
\end{equation} 
Here, $M^+_i=\max\{0,\ U_i-L'_i\},\ M^-_i =\max\{0,\ U'_i-L_i\},\ \forall i \in [b],$ and $\max_{G,G'}\sum_{i\in [b]} z_i$ gives the maximum $A-B$ for all possible distributions within the uncertainty. We can compute the solution using standard LP algorithms such as the simplex method, the interior point method, etc.

Therefore, from Theorem~\ref{Corollary:composition_apl}, $\hat{\alpha},\hat{\beta},\hat{\delta}$ estimates an upper bound for $\operatorname{CPL}^\star$ for any privacy budget level.

This completes the proof of Theorem~\ref{thm:upperbound_dist_bias}.
\end{proof}


\section{Privacy Budget Calibration}\label{appendix:Unequal Privacy Budgets Across Attributes}
We model the privacy budget calibration as a maximization problem as
\begin{equation*}
    \begin{split}
        \textbf{minimize} & \quad \operatorname{Loss}\\
        \textbf{subject to} & \quad \overline{\varepsilon}_k + \sum_{i \in [n]\setminus \{k\}} L_{X_i\rightarrow X_k} \leq \tilde{\varepsilon}; \quad\forall k \in [n],\\
        & \quad \overline{\varepsilon}_k \geq \frac{\tilde{\varepsilon}}{n}; \quad\forall k \in [n].\\
    \end{split}
\end{equation*}

Here, $\overline{\varepsilon}_k$ is the privacy budget of the $X_k$ attribute and $\frac{\tilde{\varepsilon}}{n}$ is the value of the basic privacy budget splitting (SPL) algorithm.
Since SPL is assigning privacy budgets assuming perfect correlation, we can always find a better or equal privacy budget for each attribute.
Note that the loss function is a flexible choice based on the requirements.
For example, we use~\eqref{eqn:loss_fn}.
Here, intuitively, we need to minimize the difference between the target overall privacy and the privacy leakage of each attribute.
Someone can define a different loss function by assigning weights for some attributes that are significant in privacy/utility compared to the rest. Since $L_{X_i\rightarrow X_k}$ is a non-convex function, this is a non-convex optimization problem.
In our experiments, we use a standard gradient descent algorithm to solve this optimization problem with~\eqref{eqn:loss_fn} loss function (differentiable).
\begin{equation}\label{eqn:loss_fn}
    \operatorname{Loss} = \sum_{k\in [n]} \left(\overline{\varepsilon}_k + \sum_{i \in [n]\setminus \{k\}} L_{X_i\rightarrow X_k} - \tilde{\varepsilon} \right)^2.
\end{equation}

\section{Experimental Results for Time and Space Analysis}\label{appendix:Experimental Results for Time and Space Analysis}

To validate the space and time complexity analysis, we empirically compute the time and space with varying the alphabet size as shown in Figure~\ref{fig:space_time_empirical}a--c. Here, we use a logarithmic scale on the time and space axes to validate polynomial-time/space and constant-time/space complexities. Empirical results validate the computational and space efficiency of \acronym\ compared to other algorithms.

\begin{figure}[tbh]
    \centering
\includegraphics[trim={0.1cm 0.2cm 0cm 0.2cm},clip,width=0.65\linewidth]{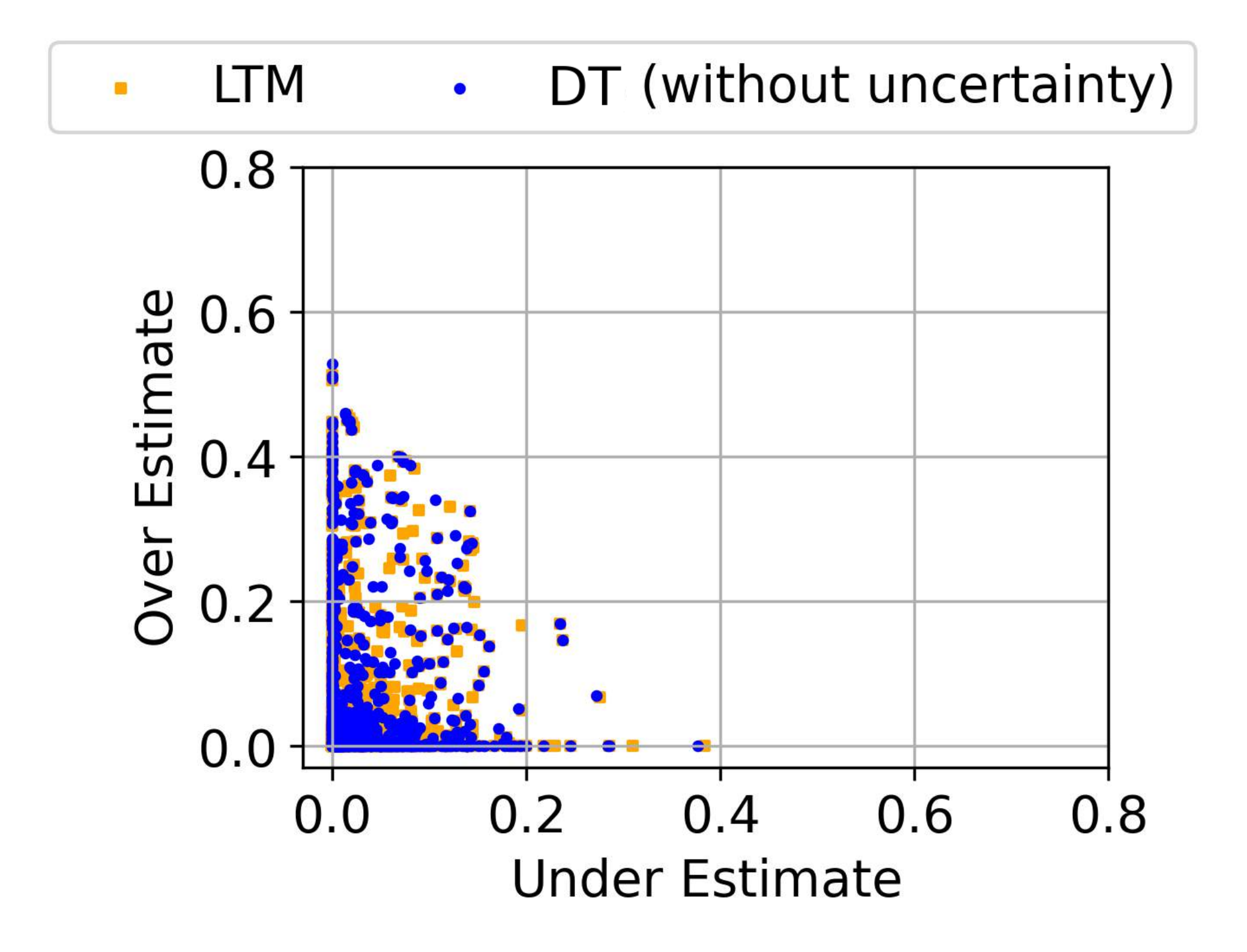}
    \caption{Normalized over estimation and under estimation of estimated CPL of \acronym\ and LTM with respect to ground truth $\operatorname{CPL}^\star$. Here, \acronym\ does not incorporate uncertainty.}
    \label{fig:overestimatio_underestimation_DB_LTM}
\end{figure}

\begin{figure*}[t]
  \centering
  \includegraphics[trim={0.8cm 0.cm 0.7cm 0.2cm},width=0.35\linewidth]{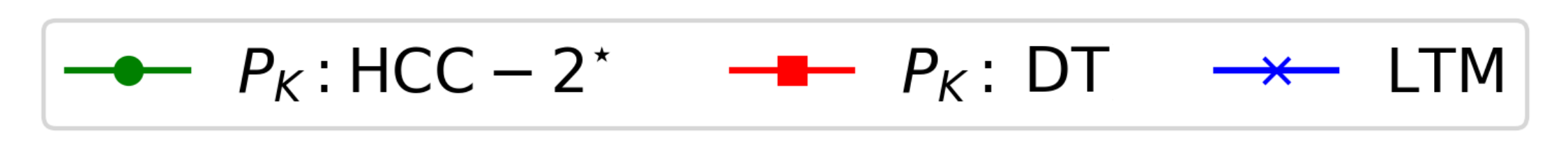}\par
  \makebox[0.9\linewidth][c]{%
    \begin{subfigure}[t]{0.23\linewidth}
      \centering
      \includegraphics[trim={1.45cm 0.8cm 1.2cm 0.8cm},clip,width=1.09\textwidth]{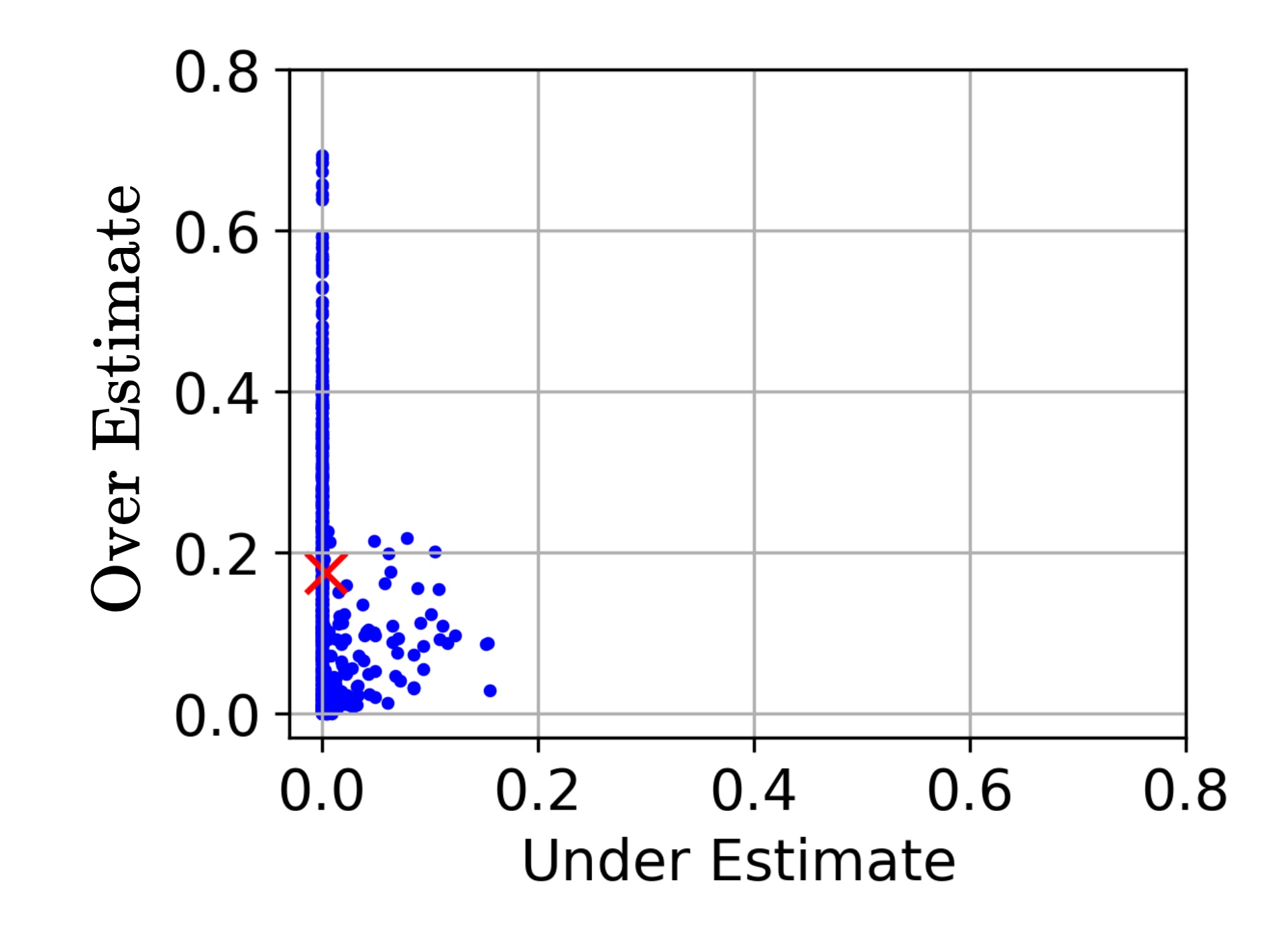}
      \caption{Range}\label{fig:over_under_est_a}
    \end{subfigure}
    \hspace{0.02\linewidth}
    \begin{subfigure}[t]{0.23\linewidth}
      \centering
      \includegraphics[trim={1.45cm 0.8cm 1.2cm 0.8cm},clip,width=0.97\textwidth]{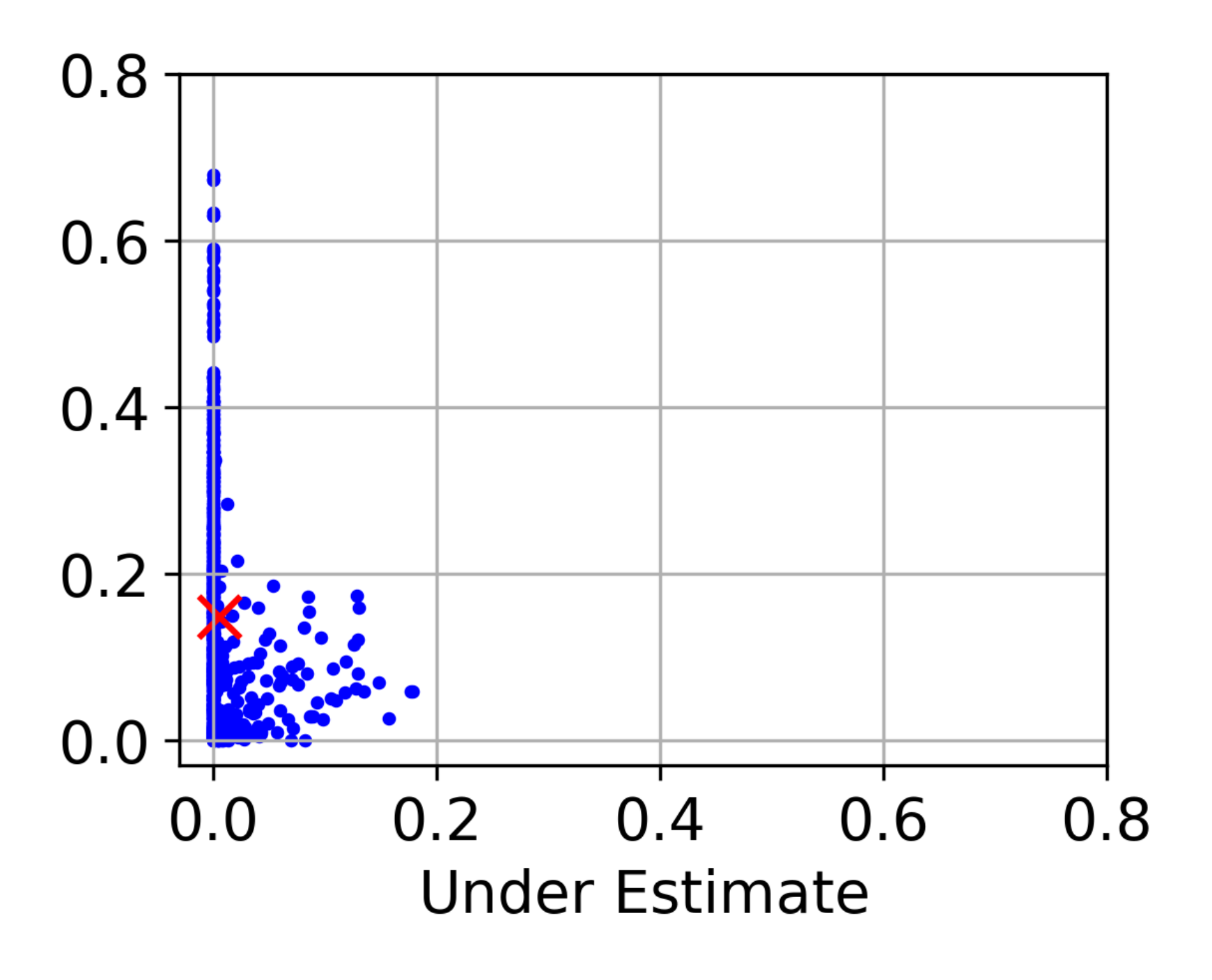}
      \caption{Standard deviation}\label{fig:over_under_est_b}
    \end{subfigure}
    \hspace{0.02\linewidth}
    \begin{subfigure}[t]{0.23\linewidth}
      \centering
      \includegraphics[trim={1.45cm 0.8cm 1.2cm 0.8cm},clip,width=0.97\textwidth]{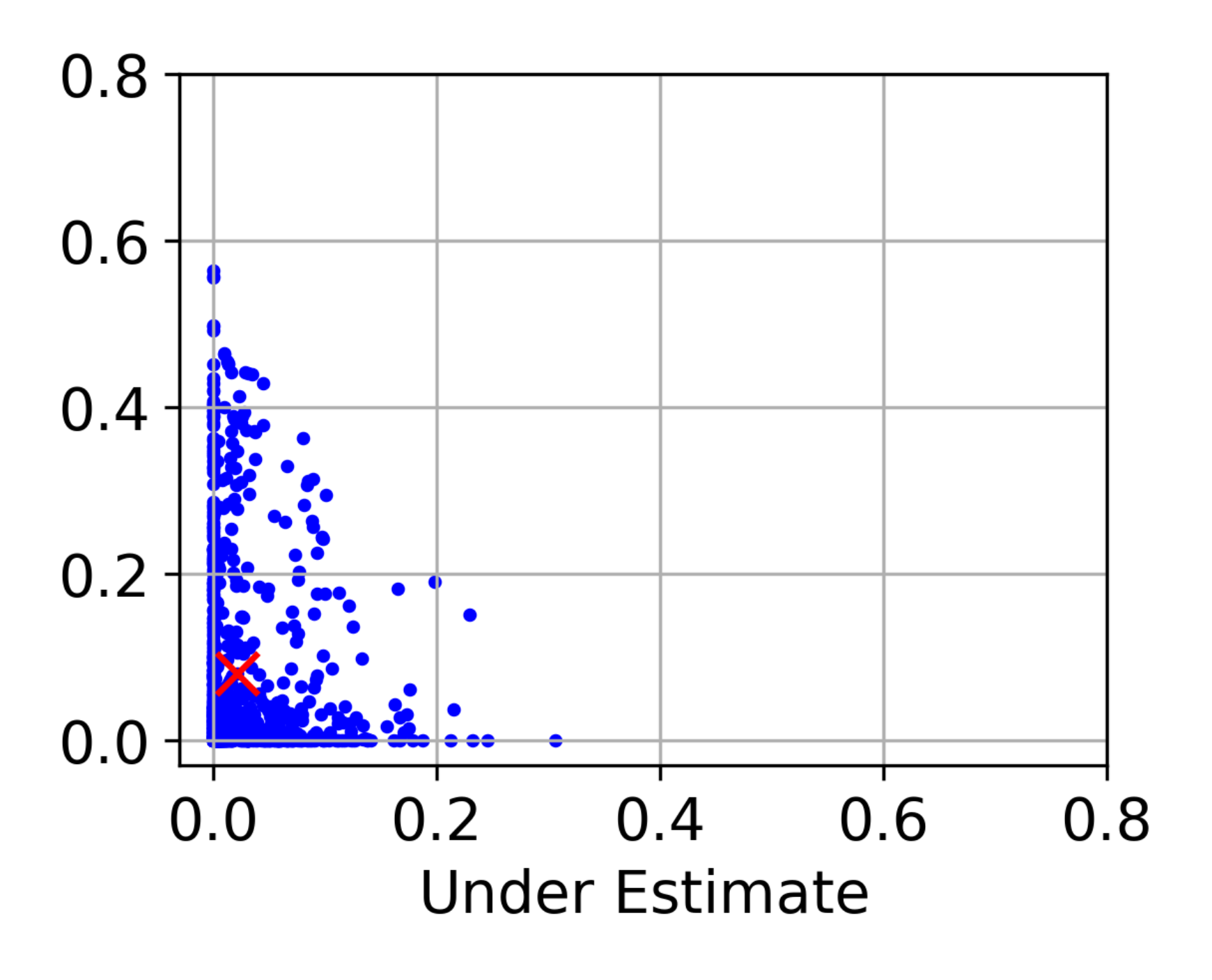}
      \caption{Variance}\label{fig:over_under_est_c}
    \end{subfigure}
    \hspace{0.02\linewidth}
    \begin{subfigure}[t]{0.23\linewidth}
      \centering
      \includegraphics[trim={1.45cm 0.8cm 1.2cm 0.8cm},clip,width=0.97\textwidth]{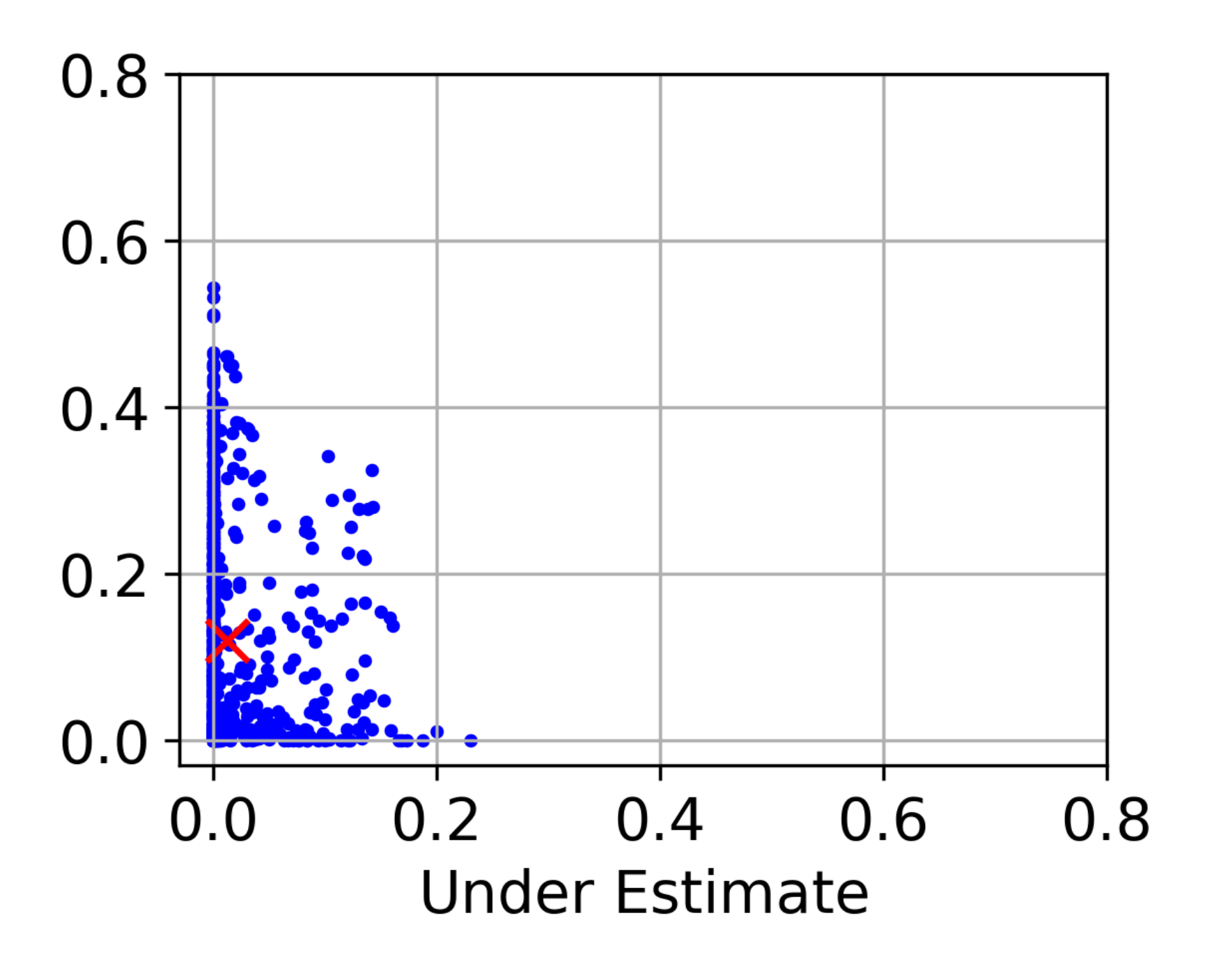}
      \caption{MAD}\label{fig:over_under_est_d}
    \end{subfigure}
  }
  \caption{Normalized overestimation and underestimation of estimated CPL with respect to ground truth $\operatorname{CPL}^\star$ for different uncertainty methods.}
    \label{fig:DB_different_uncertainty}
\end{figure*}

\section{Experimental Results for Distribution Uncertainty Defining Methods}\label{appendix:Experimental Results for Distribution Uncertainty Defining Methods}

\textbf{Statistically defined uncertainty.} 
We begin by evaluating the four statistical methods introduced in Section~\ref{section:Uncertainty Metrics}, outlining their respective strengths and limitations. To ground the analysis, we consider the motivational example: a state-wise household income survey in the United States conducted under LDP, which is explained in Section~\ref{section:introduction}. To instantiate this setup, we use the real-world \textit{SPM} dataset~\cite{SPM}. We evaluate the \emph{normalized underestimation} and \emph{overestimation} of estimated CPL (based on prior knowledge, $\operatorname{\tilde{CPL}}$) relative to the ground-truth CPL (based on actual distributional knowledge, $\operatorname{CPL^\star}$). 
\begin{align*}
\text{Normalized Underestimation} &= \frac{1}{\varepsilon }\max \{\operatorname{CPL^\star}- \operatorname{\tilde{CPL}},0\},\\
\text{Normalized Overestimation} &= \frac{1}{\varepsilon }\max \{\operatorname{\tilde{CPL}}-\operatorname{CPL^\star},0\}.
\end{align*}
Next, we plot all the attributes on an x-y graph using overestimation on the y-axis and underestimation on the x-axis.

We first evaluate \acronym\ without incorporating distributional uncertainty, using LTM as a baseline, as shown in Figure~\ref{fig:overestimatio_underestimation_DB_LTM}. 
Each point in the plot corresponds to the normalized error of a given attribute across states, excluding those used as prior knowledge.
The results indicate that both \acronym\ and LTM yield a considerable number of underestimated attributes.


Next, we incorporate uncertainty into \acronym, as shown in Figure~\ref{fig:DB_different_uncertainty}. In this figure, the \textcolor{red}{red crosses} denote average estimation. Points farther from the $x=0$ axis indicate an underestimation of CPL relative to the ground truth, which may introduce privacy vulnerabilities. Conversely, points farther from the $y=0$ axis represent overestimation of CPL, which reduces data utility. Since our goal is to ensure strong privacy protection while preserving utility, the ideal outcome is to keep points as close as possible to both the $x=0$ axis (minimizing potential privacy vulnerabilities from underestimation) and the origin (increasing data utility).

The results show that CPL estimated using \emph{range}- and \emph{standard deviation}-based uncertainty estimates exhibits less underestimation than variance- and MAD-based methods. This implies that range- and std-based uncertainty definitions provide stronger privacy guarantees. 
We further analyze the impact of prior knowledge by varying the set of prior states. 
Here, we observe that range-based uncertainty is highly sensitive to the choice of prior knowledge, while std-based uncertainty demonstrates more stable behavior.

\textbf{Naively defined uncertainty.}
In this setting, we use the publicly available LFW dataset to learn correlations between facial attributes, which are then used to estimate CPL for the CelebA dataset. 
Since \emph{only a single source} is available to infer correlations, we adopt a naive method to define uncertainty. Specifically, we vary the error parameter $\mathbf{e}$ across ${10\%, 15\%, 20\%, 30\%}$.

We observe that values of $10\% \leq \mathbf{e} \leq 20\%$ yield a better balance between utility and privacy. 
To further validate this observation, we repeat the experiment on additional datasets by using $20\%$ of the data to learn the distributions and the remainder for testing. 
The results consistently confirm that $10\% \leq \mathbf{e} \leq 20\%$ achieves the most favorable trade-off between utility and privacy.

\section{Attribute Inference Attacks}\label{appendix:section:attacks}

\noindent\textbf{TabTransformer attacker.}
We model a strong learned adversary using TabTransformer, which embeds each released categorical attribute as a token and applies a Transformer encoder to capture cross-attribute interactions, followed by a multi-class softmax head for predicting $X_j$~\cite{tabtransformer}.
We train a separate attacker $g_j$ for each target attribute using auxiliary pairs $(\tilde{\mathbf{Y}}_{-j}, X_j)$ and report the attribute-averaged Macro-F1 via $\mathrm{AIA}$. We train for 10 epochs and select a batch size of 16384.

\noindent\textbf{Macro-F1.}
Macro-F1 is the unweighted mean of the per-class $F_1$ scores (one-vs-rest), giving equal weight to each class
\begin{equation*}
    \mathrm{MacroF1} = \frac{1}{|\mathcal{C}|}\sum_{c\in\mathcal{C}} \frac{2P_cR_c}{P_c+R_c}.
\end{equation*}
Here, $P_C$ and $R_C$ are the precision and recall for class $C$.

\end{appendices}

\end{document}